\documentclass[11pt,reqno]{amsart}
\usepackage[greek.polutoniko, english]{babel}
\usepackage{mathtools,amsmath,amssymb}
\usepackage{mathrsfs,dsfont}
\usepackage{enumerate}
\usepackage{soul,xcolor}
\usepackage{tikz-cd}
\usepackage{stmaryrd}
\usepackage{slashed}
\usepackage[utf8]{inputenc}
\usepackage{aligned-overset}
\usepackage[sort]{natbib}
\usepackage[pdfencoding=auto, psdextra]{hyperref}
\hypersetup{colorlinks=true,allcolors=[rgb]{0,0,0.6}}
\usepackage[hmargin=2.7cm,vmargin=2.6cm]{geometry}
\usepackage{comment}

\usepackage{dsfont}
\usepackage{amsthm}
\usepackage{xcolor}
\usepackage[normalem]{ulem}

\newcommand{\N}{\mathbb{N}}
\newcommand{\R}{\mathbb{R}}
\newcommand{\C}{\mathbb{C}}

\newcommand{\Hi}{\mathscr{H}}
\newcommand{\eps}{\varepsilon}

\newcommand\norm[1]{\left\lVert#1\right\rVert}
\newcommand\abs[1]{\left\lvert#1\right\rvert}

\newcommand{\di}{\mathrm{d}}

\newcommand\ket[1]{\left| #1 \right\rangle}

\numberwithin{equation}{section}

\theoremstyle{plain} 
\newtheorem{theorem}{Theorem}[section]
\newtheorem{corollary}[theorem]{Corollary}
\newtheorem{lemma}[theorem]{Lemma}
\newtheorem{proposition}[theorem]{Proposition}
\newtheorem*{proposition*}{Proposition}
\newtheorem{thm}[theorem]{Theorem}

\newtheorem{lem}[theorem]{Lemma}
\newtheorem{prop}[theorem]{Proposition}

\theoremstyle{definition}
\newtheorem{definition}[theorem]{Definition}

\newtheorem*{qstn*}{Question 1}
\theoremstyle{remark}
\newtheorem{rem}[theorem]{Remark}
\newtheorem{remark}[theorem]{Remark}
\newtheorem*{remark*}{Remark}
\newtheorem*{remarks*}{Remarks}

\begin{document}
\tolerance=2000 \setlength{\emergencystretch}{1em}

\newenvironment{sistema}%
{\left\lbrace\begin{array}{@{}l@{}}}%
    {\end{array}\right.}
\title{Measuring a quantum system thanks to the correspondence principle}

\author[M.\ Falconi]{Marco Falconi} \address{Politecnico di
  Milano\\D-Mat\\P.zza da Vinci 32\\20133 Milano\\Italy}
\email{marco.falconi@polimi.it} \urladdr{https://www.mfmat.org/}

\author[C.\ Fermanian Kammerer]{Clotilde Fermanian Kammerer}
\address{LAREMA\\UMR 6093\\Université d'Angers, CNRS\\France}
\email{clotilde.fermanian@univ-angers.fr}

\author[R.\ Gautier]{Raphaël Gautier} \address{Université de
  Rennes\\CNRS\\IRMAR - UMR 6625\\F-35000 Rennes\\France \& Politecnico di
  Milano\\P.zza da Vinci 32\\20133 Milano\\Italy}
\email{raphael.gautier@univ-rennes.fr}
\email{raphael.gautier@polimi.it}

% \keywords{}

\date{\today}

\begin{abstract}
  The measurement problem is one of the longest-standing open problems in
  quantum mechanics: noncommutative observables are incompatible with the
  Born and collapse rules, prescribing that the outcome of their observation
  is distributed according to the laws of classical probability. In this
  paper, we solve the measurement problem through the correspondence principle with a semiclassical approach. We 
  prove that Born rule, the wavefunction collapse, and macroscopic
  observability all hold true, approximately, but with arbitrary precision,
  whenever a measuring device (in a semiclassical state) is coupled
  adiabatically to the observed system. Our measurement device is based on the propagation properties of semiclassical bosonic coherent states through an adiabatic two level system on scale times of the order of the Ehrenfest time.
\end{abstract}

\maketitle
\tableofcontents

\section{Introduction}
\label{sec:introduction-1}

Quantum mechanics is one of the most successful physical theories in terms of
predictive power. Nevertheless, the very process of confirming such
predictions -- \emph{i.e.}, the process of measurement and observation -- suffers from
a serious conceptual and mathematical issue that takes the name of \emph{measurement
  problem}. The measurement problem is almost as old as quantum
mechanics itself: it is already discussed, more or less implicitly, in the
final chapter of J.~von Neumann's foundational book about the mathematics of
quantum mechanics \citep{MR223138}, and it is, most assuredly, the driving
force behind N.~Bohr's decision of considering measuring devices as purely
classical objects \citep[see, \emph{e.g.},][]{bohr1928}. In this introduction, we
first review the measurement problem in its simplest and most compelling
form, and then we introduce and discuss our solution to it.

\subsection{The measurement problem}
\label{sec:measurement-problem}

The measurement problem stems from the need to resolve a strange dichotomy
involving dynamical processes, that is inherent to quantum mechanics. Quantum
mechanics is the foremost example of a theory of noncommutative probability:
its (real) random variables (in physical terms, \emph{observables}) form (the
self-adjoint elements of) a noncommutative C*-algebra, and they are thus
always representable as linear operators in a Hilbert space $\mathscr{H}$; the
associated noncommutative probability measures (in physical terms, \emph{states})
are then represented by normalized vectors of $\mathscr{H}$, or more generally by
positive trace-class operators of unit trace (in physical terms, \emph{density
  matrices}). However, the quantum mechanical time evolution of
observables, and by duality states, is perfectly \emph{deterministic}: it is
dictated by a weakly continuous group of *-automorphisms best known as the
Schrödinger unitary evolution. There is one notable exception to the above;
and it is such exception that, as already mentioned, lies at the heart of the
measurement problem: the process of measurement for a quantum observable is a
\emph{stochastic process}, dictated by the so-called \emph{Born and collapse rules}.

Let $\mathcal{O}$ be a measured self-adjoint observable with (discrete and
non-degenerate) spectrum $\lambda_0,\lambda_1,\dotsc,\lambda_j,\dotsc$ and associated normalized
eigenvectors $\psi_0,\psi_1,\dotsc,\psi_j,\dotsc$; then, on the state $\psi$ (or the density matrix
$\varrho$):
\begin{itemize}
\item[{\bf [BR]}] {\emph{Born rule}}: $\lambda_j$ is detected with (classical) probability
  $\bigl\lvert\langle \psi_j  , \psi \rangle_{}\bigr\rvert^2$ (or $\langle \psi_j,\varrho\psi_j\rangle$).
\item[{\bf[CR]}] \emph{Collapse rule}: whenever $\lambda_j$ is detected, the system
  transitions (is ``collapsed'') on the state $ \psi_j$ of the corresponding
  eigensubspace.
\end{itemize}
This dichotomy was very soon recognized as unsatisfactory, and the consensus
among physicists\footnote{As mentioned above, N.~Bohr notoriously advocated
  for the need of introducing a classical measuring device to make sense of
  the measurement process. Although there is an ongoing debate among
  historians \citep[see, \emph{e.g.},][]{laudisa2024epjh}, his stance was most
  probably practical, in the sense that the measuring device shall be a
  quantum object that for all purposes can be treated as classical
  \citep[see][for an analogous position]{landau1965III}.} was that the
process of measurement shall be described by a quantum dynamical process that
is deterministic itself, and this can be achieved by taking into
consideration the role of the \emph{measuring device}. To this extent, J.~von
Neumann introduced the following {\it measurement scheme} \citep[][Chap.\
VI]{MR223138}: suppose that we want to measure the self-adjoint quantum
observable\footnote{For simplicity, the reader shall always have in mind the
  case of a quantum observable with discrete and non-degenerate spectrum: the
  measurement problem can already be seen for the simplest quantum
  observables ($2\times2$ matrices on $\mathbb{C}^2$).} $\mathcal{O}$ on the vector state $\psi\in \mathscr{H}_{\mathcal{S}}$
of a physical system $\mathcal{S}$ (more generally, on the density matrix $\varrho$), through
a quantum probe $\mathcal{P}$ with Hilbert space $\mathscr{H}_{\mathcal{P}}$, whose initial state is a
vector state $\Psi\in \mathscr{H}_{\mathcal{P}}$ (more generally, a density matrix $\Gamma$), and we want
to do so with a deterministic process. Then, there shall exist a unitary map
$\boldsymbol{U}$ on $\mathscr{H}_{\mathcal{S}}\otimes \mathscr{H}_{\mathcal{P}}$ such that the measurement process is given by
\begin{equation*}
  \psi\otimes \Psi \longmapsto \boldsymbol{U}(\psi\otimes \Psi)\;,
\end{equation*}
or more generally for density matrices,
\begin{equation*}
  \varrho \otimes \Gamma \longmapsto \boldsymbol{U}(\varrho \otimes \Gamma) \boldsymbol{U}^{*}\;.
\end{equation*}
The measurement outcome shall be of the form
\begin{equation}
  \label{eq:16}
  \boldsymbol{U}(\varrho \otimes \Gamma) \boldsymbol{U}^{*}= \sum_j^{} \langle \psi_j  , \varrho\,\psi_j \rangle_{\mathcal{S}} \, \lvert\psi_j\rangle\langle\psi_j\rvert \otimes \lvert\Psi_j\rangle\langle\Psi_j\rvert\;,
\end{equation}
where the $\{\Psi_j\}_j$ are eigenvectors of a self-adjoint observable $\mathcal{Z}$ on
$\mathcal{P}$, called the {\it pointer}.

The data of $(\mathscr H_{\mathcal P},\psi\,\mbox{or}\,\varrho,{\mathcal
  Z},\boldsymbol{U})$ constitutes a measurement scheme for the observable
$\mathcal O$ on $\psi$ (or $\rho$). Its final state complies with the Born rule
{\bf[BR]}: $\lambda_j$ is detected with the probability $|\langle\psi_j,\psi\rangle|^2$ (or $\langle \psi_j ,
\varrho\,\psi_j \rangle_{\mathcal{S}}$). It also satisfies the collapse rule {\bf[CR]} on $\mathcal{S}$ since no
mixed states of the form $|\psi_k\rangle \langle\psi_j|$ with $k\not=j$ appear
in~\eqref{eq:16}. Moreover, the final state can be macroscopically observed
through $\mathcal{P}$ by reading the pointer, whose value is unambiguously related to
the value of the observed $\mathcal{O}$ on~$\mathcal{S}$. We call this last property the {\it
  $\mathcal P$-observability}:
\begin{itemize}
    \item[{\bf [$\mathcal P$-O]}] \emph{$\mathcal P$-observability}: $\mathcal O$ is measured by reading the pointer $\mathcal Z$ on $\mathcal P$.
\end{itemize}

A very simple argument of linear algebra -- originally due to E.P.~Wigner --
shows that there can be no unitary map $\boldsymbol{U}$ on the Hilbert space such that \eqref{eq:16}
holds: this is the core of the \emph{measurement problem}. Its proof is given in Appendix~\ref{appendixcommutatorestimates}.
\begin{thm}[Measurement problem -- E.P.~Wigner {\citep{wigner1963ajp}}]
  \label{thm:5}
  The measurement of a quantum system $\mathcal{S}$ cannot be carried out, by means of
  a probe $\mathcal{P}$, through a deterministic quantum dynamical process on $\mathcal{S}+\mathcal{P}$.
\end{thm}

Note that adding further measuring devices to the scheme above, such as considering
    $\mathcal{S}+\mathcal{P}_1+\mathcal{P}_2$ where $\mathcal{P}_1$ measures $\mathcal{S}$ and $\mathcal{P}_2$ measures $\mathcal{P}_1$, merely
    shifts but does not solve the problem, as already observed by von Neumann
    \citep{MR223138}.

Actually, 
 Wigner's argument for theorem~\ref{thm:5} was soon extended to more general notions of
    measurements \citep[see, \emph{e.g.},][]{PhysRevD.2.2783}; and even though the
    theory of quantum measurements has progressed a lot since the
    time \citep[see,
\emph{e.g.}][]{lueders1950,davies1970cmp,ozawa1984jmp,MR1419313,busch2016tmp,MR4696874},
%\footnote{The reader interested in a contemporary and detailed
 %     presentation of the mathematics of quantum measurements could, \emph{e.g.},
  %    refer to \citep{busch2016tmp}},
  the measurement problem still persists.

\subsection{Solving the measurement problem \emph{within error limits}}
\label{sec:solut-with-error}

In spite of the measurement problem, there is no experimental evidence
against the validity of ``orthodox'' quantum mechanics. Nonetheless, in view
of the fact that the measurement problem is a theorem, what does it mean to
solve it? It is clear that any attempt at a solution shall modify the
axiomatic framework that leads to Theorem~\ref{thm:5}; however, these
modifications shall be as conservative as possible, to not disrupt the
predictive power of quantum mechanics.

The fact that the deterministic Schrödinger process is implemented by a
unitary operator is expression of the general theory of symmetries in
(quantum) physics, whose role is crucial at many levels; it is thus advisable
not to modify such an axiom. On the other hand, the Born and collapse rules
only involve the process of measurement, and their weakening would not change
the mathematical framework of quantum mechanics so radically. Notably, E.\
Schrödinger's own definition of measurements requires that these rules are
true only within error limits:
\begin{quote}
  \emph{The systematically arranged interaction of two systems (measured
    object and measuring instrument) is called a measurement on the first
    system, if a directly-sensible variable feature of the second (pointer
    position) is always reproduced \underline{within certain error limits}
    when the process is immediately repeated (on the same object, which in
    the meantime must not be exposed to any additional influences).}
  \citep[][underlining added]{schroedingercat}
\end{quote}
We remark that, indeed, any deviation from the Born rule must be very small:
there is, up to now, no experimental evidence that such rule is violated
\citep[see, \emph{e.g.},][and references
therein]{sorkin1994,foerderer2025,sinha2018,PhysRevA.95.012107,PhysRevA.97.023601,Conlon_2024}. While
there are exceptions to the collapse rule (most notably, Haroche's indirect
measurements \citep[see, \emph{e.g.},][]{nogues1999nature,haroche2006oup}), we focus
here on so-called \emph{projective measurements}, where indeed the wavefunction
collapse is also observed experimentally.

The question of finding a solution to the measurement problem can thus be
reformulated as follows. As above, let $\varrho$ denote the initial state of the
system $\mathcal{S}$, $\Gamma$ the initial state of the probe~$\mathcal{P}$, and $\boldsymbol U$ the unitary map
on $\mathscr{H}_{\mathcal{S}}\otimes \mathscr{H}_{\mathcal{P}}$ implementing the measurement process. Let us denote the
measurement outcome by
\begin{equation*}
  \boldsymbol{\varrho}:= \boldsymbol{U}(\varrho \otimes \Gamma) \boldsymbol{U}^{*}\;;
\end{equation*}
and by
\begin{equation*}
  \boldsymbol{\varrho}_{\mathrm{ideal}}
\end{equation*}
an object\footnote{We are here purposefully vague: usually,
  $\boldsymbol{\varrho}_{\mathrm{ideal}}$ is taken to be a density matrix on $\mathscr{H}_{\mathcal{S}}\otimes \mathscr{H}_{\mathcal{P}}$;
  however, we take it to be a different object (see \textsection\ref{sec:bohr-type-ideal} below).} on $\mathcal{S}+\mathcal{P}$
that has the  properties {\bf [BR]}, {\bf [CR]} and {\bf [$\mathcal P$-O]}.
%\begin{enumerate}[i.]
%\item\label{item:1} \underline{$\mathcal{P}$\emph{-observability}}: $\mathcal{O}$ is measured by reading a pointer
%  observable on $\mathcal{P}$;
%\item\label{item:2} \underline{\emph{Born rule}}: $\lambda_j$ is detected with probability $\langle \psi_j  , \varrho\, \psi_j \rangle_{\mathcal{S}}$;
%\item\label{item:3} \underline{\emph{Collapse rule}}: if one reads the value $\lambda_j$ on the $\mathcal{P}$-pointer, then $\mathcal{S}$
%  is in the state $\lvert\psi_j\rangle\langle\psi_j\rvert$.
%\end{enumerate}
Theorem~\ref{thm:5} \emph{forbids} the equality $\boldsymbol{\varrho}= \boldsymbol{\varrho}_{\mathrm{ideal}}$,
but could it nevertheless hold within error limits? This leads to the
following question.
\begin{qstn*}[Solving the measurement problem within error limits]
  \label{qstn:1}
  Find $\boldsymbol{\varrho}_{\mathrm{ideal}}$ satisfying $\mathcal{P}$-observability, Born
  and collapse rules, and a topological space $\mathfrak{E}$ so that, on $\mathfrak{E}$,
  \begin{equation*}
  \label{question}
    \boldsymbol{\varrho}= \boldsymbol{\varrho}_{\mathrm{ideal}}+ \boldsymbol{R}_{\varepsilon}\;\;\mbox{ where}\;\;
    \lim_{\varepsilon\to 0} \boldsymbol{R}_{\varepsilon}=0\;.
  \end{equation*}
\end{qstn*}
Question~1 is still \emph{unanswered}, and it is precisely the question that we
answer in this paper. Furthermore, we provide quantitative bounds on
$\boldsymbol{R}_{\varepsilon}$ that can be made arbitrarily small by tuning the
initial state of the probe $\Gamma_{\varepsilon}$; finally, the parameter $\varepsilon$ has the
physical meaning of a \emph{semiclassical parameter} (see
\textsection\ref{sec:corr-princ} below) and appears only in $\Gamma_{\varepsilon}$ and $\boldsymbol U_{\varepsilon}$ (which
we can both tune by suitably preparing the probe $\mathcal{P}_{\varepsilon}$ of our physical
experiment).

An important point relies on which sense one understands the convergence of
$\boldsymbol R_\eps$ to $0$. In his paper on the insolubility of the
measurement problem \citep{PhysRevD.2.2783}, A.\ Fine discusses some of the
germinal attempts at finding an answer to Question~1 %\ref{qstn:1}
\citep[in particular,][with a focus on the
latter]{daneri1962np,daneri1966ncb,weidlich1967zfp}. Since we believe that
his analysis raises some points of interest, let us review it here briefly. A
first attempt \citep{weidlich1967zfp}, was to suppose $\boldsymbol{R}_{\varepsilon}$
small by asking that for all $\xi\in \mathscr{H}_{\mathcal{S}}\otimes \mathscr{H}_{\mathcal{P}}$,
\begin{equation}
  \label{eq:17}
  \langle \xi , \boldsymbol{R}_{\varepsilon} \xi \rangle_{}\ll \langle \xi , \boldsymbol{\varrho}_{\mathrm{ideal}}\xi\rangle_{}\;,
\end{equation}
where here $\boldsymbol{\varrho}_{\mathrm{ideal}}$ is the density matrix \eqref{eq:16} (or a suitable
generalization). However, since $\mathrm{tr}\,\boldsymbol{\varrho}=
\mathrm{tr}\,\boldsymbol{\varrho}_{\mathrm{ideal}}=1$ and thus $\mathrm{tr}\,\boldsymbol{R}_{\varepsilon}=0$, it follows
that $\boldsymbol{R}_{\varepsilon}$ cannot be a positive nonzero operator, hence \eqref{eq:17}
cannot give a useful notion of smallness. Fine's proposal is to consider the
states $\boldsymbol{\varrho}$ and $\boldsymbol{\varrho}_{\mathrm{ideal}}$ to be approximately
equal, denoted by
\begin{equation}
  \label{eq:20}
  \boldsymbol{\varrho}\simeq \boldsymbol{\varrho}_{\mathrm{ideal}}\;,
\end{equation}
if they are $\mathfrak{A}$-equivalent, where $\mathfrak{A}\subset \mathscr{L}(\mathscr{H}_{\mathcal{S}}\otimes \mathscr{H}_{\mathcal{P}})$ is a
suitably chosen subalgebra of the algebra of all bounded observables on
$\mathscr{H}_{\mathcal{S}}\otimes \mathscr{H}_{\mathcal{P}}$:
\begin{equation*}
  \boldsymbol{\varrho}\simeq \boldsymbol{\varrho}_{\mathrm{ideal}} \text{ iff } \exists \mathfrak{A}\subset \mathscr{L}(\mathscr{H}_{\mathcal{S}}\otimes \mathscr{H}_{\mathcal{P}}) \;\bigl\lvert\;  \forall A\in \mathfrak{A}\;,\; \mathrm{tr}(A\,\boldsymbol{\varrho})= \mathrm{tr}(A\,\boldsymbol{\varrho}_{\mathrm{ideal}})\;. 
\end{equation*}
The idea is that the mathematical observables not belonging to $\mathfrak{A}$ are not
observable at the macroscopic level, and thus it would be impossible -- as a
macroscopic observer with access only to $\mathfrak{A}$-observables -- to distinguish
between the true final state $\boldsymbol{\varrho}$ and the ideal but forbidden
state $\boldsymbol{\varrho}_{\mathrm{ideal}}$. The choice of the right subalgebra
$\mathfrak{A}$ have proven, however, to be elusive, and there is no conclusive result in
this direction; one of the crucial problems is that on one hand a very
natural choice for $\mathfrak{A}$ is to contain only operators of the form $1\otimes A$, where
$A$ acts on $\mathscr{H}_{\mathcal{P}}$ (since the observer looks at the probe to measure $\mathcal{S}$),
but on the other hand the unitary dynamical process $U$ cannot preserve the
structure $1\otimes A$ and carry out the measurement at the same time; in addition,
unless $\mathfrak{A}$ is a commutative subalgebra (and that conforms to our own
solution, as it will appear evident below), the effects due to
noncommutativity make very difficult to satisfy
$\boldsymbol{\varrho}\simeq\boldsymbol{\varrho}_{\mathrm{ideal}}$. In this paper, we realize
the program of Question~1 for a specific system~$\mathcal S$ with $\mathscr
H_{\mathcal S}=\mathbb C^2$, while following a different path for the
convergence that we base on the correspondence principle. In the next
paragraphs, we describe our choice of $\boldsymbol{\varrho}_{\mathrm{ideal}}$
(section~\ref{sec:bohr-type-ideal}), the way the convergence is taken
(section~\ref{sec:corr-princ}) and the choice of the unitary map of the
measurement process (section~\ref{sec:unitary}).
%now. We start }

\subsection{Choice of a Bohr-type ideal measurement}
\label{sec:bohr-type-ideal}

%As already remarked above, our solution also relies on the fact that the Born
%and collapse rules hold only approximately. Compared to Fine's scheme,
%however, it differs both in the choice of $\boldsymbol{\varrho}_{\mathrm{ideal}}$, and in
%the choice of the topological space $\mathfrak{E}$ on which they are approximately equal. 
Let us start by discussing our choice of $\boldsymbol{\varrho}_{\mathrm{ideal}}$.
We follow N.\ Bohr, that advocated for the measuring device $\mathcal{P}$ to obey the
laws of \emph{classical physics}. Therefore, our choice of $\boldsymbol{\varrho}_{\mathrm{ideal}}$
will not be a density matrix on $\mathscr{H}_{\mathcal{S}}\otimes \mathscr{H}_{\mathcal{P}}$, such as \eqref{eq:16} or its
generalizations, but rather a mixed \emph{quantum/classical state}, behaving as a
quantum state on~$\mathcal{S}$ and as a classical state on $\mathcal{P}$. Let us stress the fact
that such a quantum/classical state is not admissible within quantum
mechanics; nonetheless, we prove below that it approximates arbitrarily well
the true measurement outcome $\boldsymbol{\varrho}$ \emph{in a semiclassical sense}
(see \textsection~\ref{sec:corr-princ} below), at least for a suitable choice of $\Gamma$ and
$\boldsymbol U$.

Let us denote by $\mathfrak{X}_{\mathcal{P}_0}$ the phase space of $\mathcal{P}_0$,\footnote{We use the
  subscript $_0$ to emphasize the fact that we are here considering a
  classical probe.} seen as a classical system; we assume $\mathfrak{X}_{\mathcal{P}_0}$ to be the
continuous dual of a pre-Hilbert and (Hausdorff) topological vector space
$\Xi_{\mathcal{P}_0}$. The most natural choices are in fact $\Xi_{\mathcal{P}_0}=\mathbb{C}^d\cong\mathbb{R}^{2d}$ for the position
and momenta of classical particles, and $\Xi_{\mathcal{P}_0}=\mathscr{S}(\mathbb{R}^d)$ for the
configurations of a classical field; we will use the latter here. We define a $\mathcal{S}$-\emph{state-valued
  probability} $\boldsymbol{\mu}$ on $\mathfrak{X}_{\mathcal{P}_0}$ to be a vector measure of the
form
\begin{equation*}
  \mathrm{d}\boldsymbol{\mu}(x)= \varrho(x)\mathrm{d}\mu(x)\;,
\end{equation*}
where $\mu$ is a Borel Radon probability measure on $\mathfrak{X}_{\mathcal{P}_0}$, and $x\mapsto\varrho(x)$ is
a $\mu$-a.e.\ defined function that takes values in the density matrices of $\mathcal{S}$
(\emph{i.e.}, for $\mu$-a.a.\ $x\in \mathfrak{X}_{\mathcal{P}_0}$, $\varrho(x)$ is an operator on $\mathscr{H}_{\mathcal{S}}$ such that
$\varrho(x)\geq 0$, and $\mathrm{tr}\varrho(x)=1$). State-valued probabilities appear in the
semiclassical analysis literature \citep[see, for
instance,][]{gerard1997cpam,MR3357117}, and they are the natural mathematical
object to describe mixed quantum/classical systems \citep[see][for a detailed
presentation]{correggi2022quasi,correggi2023apde}: the measure~$\mu$ encodes
the probability of each event ``$\mathcal{P}_0$ is in configuration $x\in \mathfrak{X}_{\mathcal{P}_0}$'', while
the function $\varrho(x)$ in turn yields the $\mathcal{S}$-quantum state conditioned to that
$\mathcal{P}_0$-event.

Before the measurement, the system $\mathcal{S}$ is in the vector state $\psi\in \mathscr{H}_{\mathcal{S}}$, or
more generally in the statistical mixture of states described by a density
matrix $\varrho$. The classical probe $\mathcal{P}_0$ is instead prepared in its ``rest
state'' (zero of the scale), modeled by some probability $\mu$ on $\mathfrak{X}_{\mathcal{P}_0}$. It
is natural to assume (although not necessary for our argument below) that we
have full knowledge, in a classical sense, of the probe $\mathcal{P}_0$ while it is at
rest, since we are engineering our measuring device ourselves.  Therefore, we
can suppose that there exists a phase space point $y\in \mathfrak{X}_{\mathcal{P}_0}$ (the ``zero''
of the pointer) such that, before the measurement, the probe is in the state
\begin{equation*}
  \mu=\delta_y\;.
\end{equation*}
Therefore, the initial state-valued probability for $\mathcal{S}+\mathcal{P}_0$ is \emph{unentangled},
and given by
\begin{equation*}
  \varrho\, \mathrm{d}\delta_y(x)\;.
\end{equation*}
After the measurement, the system shall transition to a final state-valued
probability in which Born and collapse rules, and $\mathcal{P}$-observability  (see  {\bf [BR]}, {\bf [CR]} and {\bf [$\mathcal P$-O]} above) are satisfied. Among the possible ways to
comply with the above, the one that suits our needs is given by:
\begin{equation}
  \label{eq:18}
  \boldsymbol{\varrho}_{\mathrm{ideal}}= \sum_j^{} \langle \psi_j  , \varrho \,\psi_j \rangle_{\mathcal{S}} \: \lvert\psi_j\rangle\langle\psi_j\rvert \, \delta_{x_j}\;,
\end{equation}
where the $\{x_j\}_j\subset \mathfrak{X}_{\mathcal{P}_0}$ are distinct phase space points (\emph{i.e.}, $x_j\neq
x_k$ for indices such that $\lambda_j\neq\lambda_k$) 
encoding the different values of the classical
pointer observable (a pointer observable is then any real-valued phase-space
function $Z:\mathfrak{X}_{\mathcal{P}_0}\to \mathbb{R}$ such that $Z(x_j)\neq Z(x_k)$ whenever $x_j\neq
x_k$). Observe that, when restricted to $\mathcal{S}$ and $\mathcal{P}_0$, the state-valued
probability $\boldsymbol{\varrho}_{\mathrm{ideal}}$ becomes, respectively, the $\mathcal{S}$-density
matrix and $\mathcal{P}_0$-probability measure given by:
\begin{gather*}
\boldsymbol{\varrho}_{\mathrm{ideal}}\bigr\rvert_{\mathcal{S}}=\sum_j^{} \langle \psi_j  , \varrho \,\psi_j \rangle_{\mathcal{S}} \: \lvert\psi_j\rangle\langle\psi_j\rvert\;,\quad\boldsymbol{\varrho}_{\mathrm{ideal}}\bigr\rvert_{\mathcal{P}_0}= \sum_j^{}\langle \psi_j  , \varrho \,\psi_j \rangle_{\mathcal{S}} \: \delta_{x_j}\;.
\end{gather*}
The state $\boldsymbol{\varrho}_{\mathrm{ideal}}\bigr\rvert_{\mathcal{S}}$ is indeed the
statistical mixture one expects to obtain on $\mathcal{S}$ after a measurement of the
observable $\mathcal{O}$ is performed \citep[see][p.\ 347 of the English
translation]{MR223138}. It  is the noncommutative expectation $\mathbb{E}[\,\varrho\,\vert\,\mathcal{O}\,]$ of $\varrho$,
conditioned to the observation of $\mathcal{O}$ \citep[see][where it shall be intended
as the conditional expectation w.r.t.\ the subalgebra generated by~$\mathcal{O}$]{MR68751}:
\begin{equation}\label{expect}
  \mathbb{E}[\,\varrho\,\vert\,\mathcal{O}\,]= \sum_j^{} \langle \psi_j  , \varrho \psi_j\rangle_{\mathcal{S}}\: \lvert\psi_j\rangle\langle\psi_j\rvert=\boldsymbol{\varrho}_{\mathrm{ideal}}\bigr\rvert_{\mathcal{S}}\;.
\end{equation}
This quantity is the output of our measurement.

\subsection{The correspondence principle and the convergence in the measurement problem}
\label{sec:corr-princ}

%Now that we introduced our ideal measurement outcome
%$\boldsymbol{\varrho}_{\mathrm{ideal}}$, we shall define topological space $\mathfrak{E}$, and
%how the small parameter $\varepsilon$ comes into play. The ``true'' probe will now depends on the parameter $\varepsilon$, denoted $\mathcal{P}_\eps$, and it  is
%quantum, while the ``ideal'' probe $\mathcal{P}_0$ is classical. They are related by means of 
It is well-known that quantum systems behave as classical in the limit of
large quantum numbers, a property that goes under the name of
\emph{correspondence principle}. In its mathematical formulation, a
semiclassical parameter $\varepsilon$ is introduced, measuring the noncommutativity of
quantum canonical observables (by making their commutator proportional to
$\varepsilon$); such a parameter notably enters in the {\it quantization maps} and the definition of \emph{coherent states} and \emph{Wigner measures} which are among the 
tiny fraction of the (very powerful) tools
of \emph{semiclassical analysis} that we use here. While it would be nigh impossible to
cite all relevant works, let us mention \citep{GerLeich93,lions1993rmi,gerard1997cpam,fermanian2002bsmf,ammari2008,ammari2014jsp,ammari2017sima,falconi2017arxiv,correggi2022quasi,correggi2023apde}, see also references
therein.   Let us review here briefly the core ideas we will use.

\subsubsection{Quantization maps}
Let $\Xi$ be a
pre-Hilbert and topological vector space with inner product $\langle \cdot , \cdot \rangle_{}$, then
\begin{equation*}
  \bigl(\Xi^{(\mathbb{R})}, \mathrm{Im}\langle \cdot , \cdot \rangle_{}\bigr)\;,
\end{equation*}
is a \emph{real symplectic space}, where $\Xi^{(\mathbb{R})}$ is $\Xi$ seen as a real
pre-Hilbert space with inner product $\mathrm{Re}\langle \cdot , \cdot \rangle_{}$ by ``doubling
the basis''\footnote{Let $\{e_\alpha\}_{\alpha\in \mathcal{A}}$ be a complex basis of $\Xi$; then
  $\{e_\alpha, i\,e_\alpha\}_{\alpha\in \mathcal{A}}$ is a real basis of $\Xi^{(\mathbb{R})}$.}. The symplectic
structure on $\Xi$ makes it possible to define Hamiltonian mechanics on the continuous
dual $\mathfrak{X}^{(\mathbb{R})}= (\Xi^{(\mathbb{R})})'$, and
thus to consider it a phase space of classical mechanics in the usual
sense. It also allows to define the \emph{Heisenberg group} $\mathbf{H}(\Xi)$,
\emph{i.e.}\ $\mathbf{H}(\Xi)=\Xi\times \mathbb{R}$ is the 2-step nilpotent Hausdorff
topological group with group law
\begin{equation*}
  \bigl(\xi,s\bigr)\cdot \bigl(\xi',s'\bigr)= \bigl(\xi+\xi',s+s'- \mathrm{Im}\langle \xi  , \xi' \rangle_{}\bigr)\;.
\end{equation*}
Its unitary dual $\hat{\mathbf{H}}(\Xi)$ can be fibered by its central
character\footnote{The center of $\mathbf{H}(\Xi)$ is given by: $Z=\{(0,s)\;,\;
  s\in \mathbb{R}\}$.}: given an irreducible unitary representation $\pi$ of
$\mathbf{H}(\Xi)$, then there exists $\varepsilon\in \mathbb{R}$ such that for all $s\in \mathbb{R}$,
\begin{equation*}
  \pi(0,s)= e^{i \varepsilon s}\, 1\;,
\end{equation*}
where $1$ stands for the identity operator on the Hilbert space $\mathscr{H}_{\pi}$ of
the representation. Therefore, we can write
\begin{equation*}
  \hat{\mathbf{H}}(\Xi)= \bigcup_{\varepsilon\in \mathbb{R}} \hat{\mathbf{H}}_{\varepsilon}(\Xi)\;,
\end{equation*}
where $\hat{\mathbf{H}}_{\varepsilon}(\Xi)$ is the set of equivalence classes of
irreducible representations with central character~$\varepsilon$.

Let $\mathfrak{X}=\Xi'$ be the continuous dual of $\Xi$, then $S\subset \mathbb{C}^{\mathfrak{X}}$ is a ``\emph{symbol
  class}'', \emph{i.e.}\ a subspace of phase-space functions (called symbols
in semiclassical analysis and classical observables in physics) of suitable
regularity; we do not give a precise definition of the class $S$, most
notably because it depends on whether $\mathfrak{X}$ is finite or infinite
dimensional. Let us also denote by $\mathscr{L}\bigl(\hat{\mathbf{H}}_{\varepsilon}(\Xi)\bigr)$ the
\emph{universal operators} on Heisenberg group $\varepsilon$-representations (in
physics, quantum observables): $\mathrm{T}\in
\mathscr{L}\bigl(\hat{\mathbf{H}}_{\varepsilon}(\Xi)\bigr)$ is a collection
$\mathrm{T}=\{T_{[\pi_{\varepsilon}]}\}_{[\pi_{\varepsilon}]\in \hat{\mathbf{H}}_{\varepsilon}(\Xi)}$ of
equivalence classes of operators on irreducible $\varepsilon$-representations,
\emph{i.e.}\ given $\pi\in [\pi_{\varepsilon}]$, $T_{\pi}\in T_{[\pi_{\varepsilon}]}$ is a linear operator on
$\mathscr{H}_{\pi}$, and whenever $\pi,\pi'\in [\pi_{\varepsilon}]$ with $\mathscr{H}_{\pi'}=U \mathscr{H}_{\pi}$, then
$T_{\pi},T_{\pi'}\in T_{[\pi_{\varepsilon}]}$ with $T_{\pi'}= UT_{\pi}U^{-1}$. Then a
\emph{quantization map} $\mathrm{Op}_{\varepsilon}$ is a map
\begin{align*}
  \mathrm{Op}_{\varepsilon}: S&\longrightarrow \mathscr{L}\bigl(\hat{\mathbf{H}}_{\varepsilon}(\Xi)\bigr)\\
  a&\longmapsto \mathrm{Op}_{\varepsilon}(a)= \{\mathrm{Op}_{\varepsilon}(a)_{[\pi_{\varepsilon}]}\}_{[\pi_{\varepsilon}]\in \hat{\mathbf{H}}_{\varepsilon}(\Xi)}\;,
\end{align*}
that associates a quantum observable (concretely represented as a linear
operator in all irreducible representations of the Heisenberg group) to each
classical observable (in the symbol class). If $\Xi=\mathfrak{X}=\mathbb{C}^d$ is finite
dimensional, $\hat{\mathbf{H}}_{\varepsilon}(\Xi)$ has only \emph{one element} (for $\varepsilon\neq0$), by the
celebrated Stone-von Neumann theorem, and therefore the quantization maps can
be written very explicitly by working on the well-known Schrödinger
representation. If $\Xi$ is infinite dimensional, on the other hand, there are
\emph{uncountably many} equivalence classes of irreducible representations (\emph{e.g.},
corresponding to free fields with different masses, or to interacting
fields), and thus it is crucial to have a quantization procedure that works
for any representation \citep[see, \emph{e.g.},][]{falconi2017arxiv}. The best known
quantization map is the so-called \emph{Weyl quantization}: roughly speaking, it is
the one such that the Fourier characters\footnote{Here $\langle \xi , \cdot \rangle_{}$
  denotes, with a small abuse of notation, the duality pairing between $\Xi$
  and $\mathfrak{X}$, \emph{i.e.}, $\langle \xi , x \rangle_{}=x(\xi)$ for any $\xi\in \Xi$ and $x\in \mathfrak{X}=\Xi'$.}
$e^{ i \sqrt{2}\,\mathrm{Re}\langle \xi , \cdot \rangle_{}}$ are mapped to the Heisenberg group elements; \emph{i.e.}\
(formally) satisfying
\begin{equation*}
  \mathrm{Op}^{\mathrm{Weyl}}_{\varepsilon}(e^{i \sqrt{2}\, \mathrm{Re}\langle \xi , \cdot \rangle_{}})_{\pi} = \pi(\xi,0)\;.
\end{equation*}

Observe that the canonical operators have commutator proportional to $\varepsilon$
whenever they are represented by $\pi\in [\pi_{\varepsilon}]\in \hat{\mathbf{H}}_{\varepsilon}(\Xi)$ with
central character $\varepsilon$, and therefore: on one hand the $\varepsilon=0$ (or trivial)
representations shall be associated to classical mechanics (see below), and
on the other hand the limit $\varepsilon\to 0$\footnote{It is possible to define
  representations for all $\varepsilon\in \mathbb{R}$; however, when used as a semiclassical
  parameter, a very useful rule-of-thumb is to put $\varepsilon$ wherever the Planck's
  constant $h$ physically appears. It is therefore customary to restrict to
  the case $\varepsilon>0$.} shall be intended as a \emph{classical limit}, that
implements the correspondence principle. In order to do so, however, one
shall focus on the behavior of \emph{quantum states}, since they dictate
whether quantum numbers in the system are ``large'' (\emph{i.e.}, how small
the parameter $\varepsilon$ can be considered). In fact, there is a well-known tool to
construct representations from states, called the \emph{GNS construction}. To
uniquely define a quantum state\footnote{The strongly continuous
  representations discussed here are the ones yielded by so-called
  \emph{regular quantum states}; non-regular quantum states play a marginal
  role in physics, and no role in this paper.}  it suffices to specify how it
acts on the Heisenberg group \citep{segal1959mfmdvs,segal1961cjm}: there is a
bijection between states and continuous\footnote{To be precise, there is a
  bijection between regular states and \emph{pseudo-continuous} (continuous
  on all finite dimensional subspaces) group positive definite functions. All
  the states we consider here, however, correspond to continuous functions.}
positive definite functions on the group $\mathbf{H}(\Xi)$, the latter being
functions $G: \mathbf{H}(\Xi)\to \mathbb{C}$ satisfying
\begin{equation*}
  \sum_{j,k=1}^n \alpha_j \bar{\alpha}_k G\bigl((\xi_k,s_k)^{-1}\cdot (\xi_j,s_j)\bigr)\geq 0
\end{equation*}
for all $\{\alpha_j\}_{j=1}^n\subset \mathbb{C}$ and $\{(\xi_j,s_j)\}_{j=1}^n\subset \mathbf{H}(\Xi)$ (recall that
$(\xi_k,s_k)^{-1}= (-\xi_k,-s_k)$). The GNS construction associates to each such
function $G$ a cyclic representation $(\mathscr{H}_G,\pi_G,\Omega_G)$ of the Heisenberg group,
in which the cyclic vector $\Omega_G\in \mathscr{H}_G$ is the state vector uniquely identified
by $G$. Irreducible representations are associated to the so-called \emph{indecomposable
  functions}, or \emph{pure states}; more generally a function $G_{\varepsilon}$ has
central character $\varepsilon$ if and only if its induced representation has central
character $\varepsilon$, and in that case there exists a function $g_{\varepsilon}:\Xi\to \mathbb{C}$ such
that:
\begin{equation}\label{def:Geps}
G_{\varepsilon}(\xi,s)= g_{\varepsilon}(\xi)\, e^{i\varepsilon s}\;\;\mbox{and}\;\;
\sum_{j,k}^{}\alpha_j\bar{\alpha}_k \,g_{\varepsilon}(\xi_j-\xi_k)e^{ i \frac{\varepsilon}{2} \mathrm{Im}\langle \xi_j , \xi_k\rangle_{}}\geq 0.
\end{equation}
Now, if $\varepsilon=0$, $G_0(\xi,s)=g_0(\xi)$ is independent of $s$, and the function
$g_0$ is (pseudo-)continuous and positive definite on $\Xi$; these functions
are in bijection, by Bochner's theorem, with \emph{cylindrical measures} on
$\mathfrak{X}$ (of which Borel Radon measures are a subset): $\varepsilon=0$ states are
\emph{classical states}. This is the main intuition behind the correspondence
principle: quantum states converge, in the limit $\varepsilon\to 0$, to classical states.

\subsubsection{Coherent states}\label{sec:coherent}
Before pushing the mathematics further, let us discuss the above from a
physical perspective. How strong the quantum effects are in a physical system
is dictated by the state, and any state induced by $G_{\varepsilon}$ with $\varepsilon\ll1$ shall
behave almost as a classical state. Most importantly, the state of a physical
system can be \emph{prepared} to a very good degree of precision (there are many
well-established experimental techniques to do that), and thus the parameter
$\varepsilon$ can be tuned to our needs.

The most famous example of semiclassical states in physics -- and
coincidentally, the ones we use to solve the measurement problem -- are
\emph{coherent states} $\lvert x_{\varepsilon}\rangle$, attached to some point $x\in \mathfrak{X}$. Coherent
states are the ones obtained by fixing $g_\eps$ to the standard Gaussian
function in~\eqref{def:Geps}, thus yielded by a shifted-Gaussian group
function $G_{\varepsilon}[x]$:
  \begin{equation*}
    G_{\varepsilon}[x](\xi,s)= e^{- \frac{\varepsilon}{2}\lVert \xi  \rVert_{}^2}e^{i \sqrt{2} \mathrm{Re}\langle \xi  , x \rangle_{}}e^{i\varepsilon s}\;,
  \end{equation*}
  and they are also called {\it Gaussian wave-packets} (centered around the
  classical configuration $x\in \mathfrak{X}$\footnote{Note that when $\mathfrak X=\mathbb R^{2d}$, $x=(q,p)$ represents the average position/momentum of the state, and $G_\eps[x] (\xi,0)$ is the Fourier transform of the wave packet as defined in~\cite{Combescure_Robert_book}.}). Coherent states are crucial in quantum
  optics: they are the ``most classical'' quantum states (they have minimal
  uncertainty)~\cite{Hagedorn_71,Combescure_Robert_book,Fe_LeRo}, and they
  are used to model the lasing states that can be prepared in the lab. The
  average number of excitations (photons in the case of lasers) in a coherent
  state $\lvert x_{\varepsilon}\rangle$ is proportional to $\varepsilon^{-1}$, and the most powerful laser
  pulses that can be manufactured have $\varepsilon \sim 10^{-20}$: the quantum effects on
  such lasers are indeed negligible, and they can be considered for all
  practical purposes as classical electromagnetic pulses.

  \subsubsection{Wigner measures} Coming back to the mathematical side, the
  identification of quantum states with group-functions with central
  character $\varepsilon>0$, and of classical states with group-functions with central
  character $\varepsilon=0$, allows to put quantum and classical states on the same
  grounds, and thus to give a meaning to the equality
  $\boldsymbol{\varrho}=\boldsymbol{\varrho}_{\mathrm{ideal}}+\boldsymbol{R}_{\varepsilon}$ and to
  $\lim_{\varepsilon\to 0}\boldsymbol{R}_{\varepsilon}=0$ (we discuss the addition of the $\mathcal{S}$
  degrees of freedom below). Now, let $\Gamma_{\varepsilon}$ be a quantum density matrix
  (\emph{i.e.}, a quantum state represented in a representation $(\mathscr{H},\pi_{\varepsilon})$
  with central character $\varepsilon>0$), and denote by $\hat{\Gamma}_{\varepsilon}: \Xi\to \mathbb{C}$ its
  \emph{noncommutative Fourier transform}:
\begin{gather*}
  \hat{\Gamma}_{\varepsilon}(\xi)= \mathrm{tr}\bigl(\Gamma_{\varepsilon}\pi_{\varepsilon}(\xi,0)\bigr)\;,\\
  \intertext{and therefore,}
  G_{\Gamma_{\varepsilon}}(\xi,s)= \hat{\Gamma}_{\varepsilon}(\xi)\,e^{i\varepsilon s}\;.
\end{gather*}
Furthermore, let $\mu$ be a probability measure on $\mathfrak{X}$, and denote by
$\hat{\mu}:\Xi\to \mathbb{C}$ its Fourier transform:
\begin{gather*}
  \hat{\mu}(\xi)= \int_{\mathfrak{X}}^{}e^{i \sqrt{2}\,\mathrm{Re}\langle \xi  , x \rangle_{}}  \mathrm{d}\mu(x)\;.
\end{gather*}
Then the state $\Gamma_{\varepsilon}$ \emph{converges to }$\mu$ \emph{in the sense of Fourier
  transforms}, denoted by
\begin{equation*}
  \Gamma_{\varepsilon} \xrightarrow[\varepsilon\to 0]{\mathcal{F}}\mu\;,
\end{equation*}
if and only if for all $\xi\in \Xi$,
\begin{equation*}
  \hat{\Gamma}_{\varepsilon}(\xi) \xrightarrow[\varepsilon\to 0]{} \hat{\mu}(\xi)\;.
\end{equation*}
In this case, $\mu$ is called the \emph{Wigner measure} associated with $\Gamma_{\varepsilon}$. 

The
convergence in the sense of Fourier transforms implies convergence of quantum
expectations to classical expectations as well. 
%\begin{prop}
  %\label{prop:wignerfourierimplic}
Let $\Gamma_{\varepsilon}$ be a density matrix and $\mu$ a classical state; if
  %\begin{equation*}
  $\displaystyle{\Gamma_{\varepsilon} \xrightarrow[\varepsilon\to 0]{\mathcal{F}}\mu\;}$,
 % \end{equation*}
  then for all $a\in S$,
  \begin{equation*}
    \lim_{\varepsilon\to 0} \mathrm{tr}\bigl(\Gamma_{\varepsilon} \mathrm{Op}^{\mathrm{Weyl}}_{\mathrm{\varepsilon}}(a)\bigr)= \int_{\mathfrak{X}}^{}a(x)  \mathrm{d}\mu(x)\;.
  \end{equation*}
%\end{prop}

This is the most compelling mathematical formulation of the correspondence
principle: in the limit of large quantum numbers $\varepsilon\to 0$ (cf.~\textsection\ref{sec:coherent}), quantum systems behave as classical. Observe that,
for all $x\in \mathfrak{X}$,
\begin{equation*}
  \lvert x_{\varepsilon}\rangle\langle x_{\varepsilon}\rvert \xrightarrow[\varepsilon\to 0]{\mathcal{F}} \delta_x\;,
\end{equation*}
\emph{i.e.}\ the coherent Gaussian wave-packet of \textsection\ref{sec:coherent} becomes a delta distribution in the
classical limit. Moreover, thanks to the correspondence principle, \emph{quantum
  superpositions are suppressed as $\varepsilon\ll 1$} (for coherent states, they are
even suppressed exponentially fast): this is one of the two crucial
ingredients that we use to positively answer Question~1. The proposition
below reflects that phenomenon (its proof is performed in
Appendix~\ref{appendixcommutatorestimates} for the concrete model under study,
see \textsection\ref{sec:meas-scheme-thro}).
\begin{prop}
  \label{prop:superpsuppr}
  Let $\Psi_{\varepsilon},\Phi_{\varepsilon}$ be vector states, and let $\mu,\nu$ be probability measures
  on $\mathfrak{X}$. If
  \begin{equation*}
    \lvert \Psi_{\varepsilon}\rangle\langle \Psi_{\varepsilon}\rvert \xrightarrow[\varepsilon\to 0]{\mathcal{F}} \mu\;,\; \lvert \Phi_{\varepsilon}\rangle\langle \Phi_{\varepsilon}\rvert \xrightarrow[\varepsilon\to 0]{\mathcal{F}} \nu
  \end{equation*}
  and $\mu\perp\nu$ (the measures are mutually singular), then:
  \begin{align*}
&\forall \alpha,\beta\in \mathbb{C},\;\;
      \lvert \alpha \Psi_{\varepsilon}+\beta\Phi_{\varepsilon}\rangle\langle\alpha\Psi_{\varepsilon}+\beta\Phi_{\varepsilon}\rvert \xrightarrow[\varepsilon\to 0]{\mathcal{F}} \lvert \alpha  \rvert_{}^2 \,\mu + \lvert \beta  \rvert_{}^2 \, \nu\;;
 \\
&\forall a\in S,\;\;\forall  \xi\in \Xi ,
   \;\;
      \lim_{\varepsilon\to 0}\langle \Psi_{\varepsilon}  , \mathrm{Op}^{\mathrm{Weyl}}_{\mathrm{\varepsilon}}(a) \Phi_{\varepsilon}\rangle_{}= 0\;\;\mbox{and}\;\;
      \lim_{\varepsilon\to 0}\langle \Psi_{\varepsilon}  , \pi_{\varepsilon}(\xi,0) \Phi_{\varepsilon}\rangle_{}= 0\;.
 \end{align*}
\end{prop}

\subsubsection{The convergence in the  measurement problem}
In our setup, we have two quantum subsystems, the measured system $\mathcal{S}$ and the
probe $\mathcal{P}_\varepsilon$, and we want only the probe to be \emph{semiclassical}. 
%(from now on, we
%denote it by $\mathcal{P}_{\varepsilon}$ to emphasize the fact that it is prepared in a
%semiclassical state with central character $\varepsilon$). 
All the discussion above can
be generalized to this case, using the state-valued measures of
\textsection\ref{sec:bohr-type-ideal} \citep[see][for a detailed
discussion]{correggi2022quasi}. Consider a bipartite system $\mathscr{H}_{\mathcal{S}}\otimes
\mathscr{H}_{\mathcal{P}_{\varepsilon}}$ in which $\mathscr{H}_{\mathcal{P}_{\varepsilon}}$ carries a representation $\pi_{\varepsilon}$ of the
Heisenberg group. Then a state $\boldsymbol{\sigma}_{\varepsilon}$ converges to a
state-valued measure $\mathrm{d}\boldsymbol{\nu}(x)=\sigma(x)\mathrm{d}\nu(x)$ \emph{in the sense of Fourier transforms}, denoted again
\begin{equation*}
  \boldsymbol{\sigma}_{\varepsilon}\xrightarrow[\varepsilon\to 0]{\mathcal{F}} \boldsymbol{\nu}\;,
\end{equation*}
if and only if for all $K$ compact on $\mathscr{H}_{\mathcal{S}}$, and $\xi\in \Xi$,
\begin{equation}\label{eq:fourier_convergence}
  \mathrm{tr}_{\mathcal{S}}\bigl(K\hat{\boldsymbol{\sigma}}_{\varepsilon}(\xi)\bigr) \xrightarrow[\varepsilon\to 0]{}\mathrm{tr}_{\mathcal{S}}\bigl(K\hat{\boldsymbol{\nu}}(\xi)\bigr) \;,
\end{equation}
\emph{i.e.},
\begin{equation*}
  \mathrm{tr}_{\mathcal{S}+\mathcal{P}_{\varepsilon}}\Bigl(\bigl(K\otimes \pi_{\varepsilon}(\xi ,0)\bigr)\,\boldsymbol{\sigma}_{\varepsilon}\Bigr) \xrightarrow[\varepsilon\to 0]{} \int_{\mathfrak{X}}^{} \mathrm{tr}_{\mathcal{S}}\bigl(K\sigma(x)\bigr)e^{ i \sqrt{2} \mathrm{Re}\langle \xi   , x \rangle_{}}   \mathrm{d}\nu(x)\;.
\end{equation*}
Such a convergence implies convergence of expectations
\begin{equation*}
  \lim_{\varepsilon\to 0} \mathrm{tr}\Bigl(\bigl(K\otimes \mathrm{Op}^{\mathrm{Weyl}}_{\mathrm{\varepsilon}}(a)\bigr)\, \boldsymbol{\sigma}_{\varepsilon}\Bigr)= \int_{\mathfrak{X}}^{}\mathrm{tr}_{\mathcal{S}}\bigl(K\sigma(x)\bigr)  a(x)  \mathrm{d}\nu(x)
\end{equation*}
for all $a\in S$. If the initial quantum state is unentangled, \emph{i.e.}\
%\begin{equation*}
 $ \boldsymbol{\sigma}_{\varepsilon}= \varrho\otimes \Gamma_{\varepsilon}\;$,
%\end{equation*}
and $\Gamma_{\varepsilon}\xrightarrow[\varepsilon\to 0]{\mathcal{F}}\mu$, then
\begin{equation*}
  \varrho\otimes \Gamma_{\varepsilon} \xrightarrow[\varepsilon\to 0]{\mathcal{F}} \varrho \,\mathrm{d}\mu\;.
\end{equation*}
Concretely, we set up our measuring device to have a semiclassical probe
$\mathcal{P}_{\varepsilon}$, whose zero-of-the-scale-state is a coherent state $\Gamma_{\varepsilon}= \lvert x_{\varepsilon}\rangle\langle
x_{\varepsilon}\rvert$. Suppose that initially $\mathcal{S}$ is in the state $\varrho$; then, $\mathcal{S}+\mathcal{P}_{\varepsilon}$ is in the state
\begin{equation*}
\varrho\otimes \Gamma_\eps=
  \varrho\otimes \lvert x_{\varepsilon}\rangle\langle x_{\varepsilon}\rvert\;,
\end{equation*}
and
\begin{equation*}
  \varrho\otimes \lvert x_{\varepsilon}\rangle\langle x_{\varepsilon}\rvert\xrightarrow[\varepsilon\to 0]{\mathcal{F}} \varrho\, \mathrm{d}\delta_x\;,
\end{equation*}
where $\varrho \,\mathrm{d}\delta_x$ is exactly the initial state of $\mathcal{S}+\mathcal{P}_0$ in our ideal
Bohr-type measurement of \textsection\ref{sec:bohr-type-ideal}.

\subsection{Choice of a unitary process}\label{sec:unitary}
The unitary process involved in our measurement scheme is constructed via an evolution operator realizing an {\it adiabatic coupling} between the system $\mathcal S$ and the probe $\mathcal P_\varepsilon$, and letting the time evolve up to  {\it large time asymptotics of  Ehrenfest's type}.

\subsubsection{The adiabatic coupling}
\label{sec:adiabatic-coupling}

It remains to show that we can choose the unitary process $\boldsymbol{U}_{\varepsilon}$ on $\mathscr{H}_{\mathcal{S}}\otimes
\mathscr{H}_{\mathcal{P}_{\varepsilon}}$ in such a way that
\begin{equation*}
  \boldsymbol{\varrho}_{\varepsilon}= \boldsymbol{U}_{\varepsilon}\bigl(\varrho\otimes \lvert x_{\varepsilon} \rangle\langle x_{\varepsilon}\rvert\bigr)\boldsymbol{U}^{*}_{\varepsilon} \xrightarrow[\varepsilon\to 0]{\mathcal{F}} \boldsymbol{\varrho}_{\mathrm{ideal}}\;.
\end{equation*}
Let us define $\boldsymbol{U}_{\varepsilon}(t)$ as the unitary Schrödinger group on
$\mathscr{H}_{\mathcal{S}}\otimes \mathscr{H}_{\mathcal{P}_{\varepsilon}}$ generated by an Hamiltonian $\boldsymbol{H}_{\varepsilon}$ with
semiclassical time, \emph{i.e.}\
\begin{equation}
  \label{eq:19}
  \boldsymbol{U}_{\varepsilon}(t)= e^{-i \frac{t}{\varepsilon}\boldsymbol{H}_{\varepsilon}}\;,
\end{equation}
and
\begin{equation*}
  \boldsymbol{H}_{\varepsilon}= \mathcal{O}\otimes 1 + 1\otimes H_{\varepsilon} + \boldsymbol{K}_{\varepsilon}\;,
\end{equation*}
where $\mathcal{O}$ is the $\mathcal{S}$-observable we want to measure, $H_{\varepsilon}$ is the probe's
free Hamiltonian (acting only on $\mathcal{P}_{\varepsilon}$), and $\boldsymbol{K}_{\varepsilon}$ is the
interaction potential, coupling $\mathcal{S}$ to $\mathcal{P}_{\varepsilon}$. Both $H_{\varepsilon}$ and
$\boldsymbol{K}_{\varepsilon}$ depend on $\varepsilon$ only through the \emph{canonical
  variables}, \emph{i.e.}\ the self-adjoint generators of the Heisenberg
group $\pi_{\varepsilon}(\xi,0)$, $\xi\in \Xi_{\mathcal{P}_0}$. It is convenient to choose the eigenvectors
of $\mathcal{O}$ as the reference orthonormal basis of $\mathscr{H}_{\mathcal{S}}$, so that $\mathcal{O}$ is
diagonal.

Suppose, for simplicity, that $\mathscr{H}_{\mathcal{S}}\cong \mathbb{C}^n$ is finite dimensional: the total
Hamiltonian $\boldsymbol{H}_{\varepsilon}$ can be seen as an operator valued matrix
\begin{equation*}
  \boldsymbol{H}_{\varepsilon}=
  \begin{pmatrix}
    H_{11}&H_{12}&\dotsm&\dotsm&H_{1n}\\
    H_{21}&H_{22}&\dotsm&\dotsm&H_{2n}\\
    \vdots&\vdots&\ddots&\vdots&\vdots\\
    \vdots&\vdots&\vdots&\ddots&\vdots\\
    \dotsm&\dotsm&\dotsm&\dotsm&H_{nn}
  \end{pmatrix}
\end{equation*}
with formal $\mathscr{H}_{\mathcal{P}_{\varepsilon}}$-operator-valued eigenvalues
$\{\lambda^{\varepsilon}_j\}_{j=1}^n$. These operators play an important role in the theory
of \emph{adiabatic approximation} \citep[see][]{Teufel_book}, that is
dictated by the scaling in \eqref{eq:19}. Indeed, within the semiclassical
time scale $\tau= \frac{t}{\varepsilon}$ physical quantities $F(t)$ exhibit slow
variations: when $\tau$ is fixed, any $\tau$-derivative of $F(t)=F(\eps\tau)$ is
damped by the small parameter $\eps$. These small variations are indicative
of adiabatic thermodynamic processes. More precisely, the term
\emph{adiabatic} derives from the Greek term {\tt
  \textgreek{<ad$\mspace{-2mu}$i$\mspace{-2.3mu}$'abatoc}}, which means
\emph{impassable}. Therefore, $\boldsymbol{H}_{\varepsilon}$ must be chosen so that the
$\{\lambda^{\varepsilon}_j\}_{j}$ are separated by an impassable gap (in our concrete
example, we have $n=2$ and $\lambda^\eps_1-\lambda^\eps_2\geq 2>0$). It is proved that for
adiabatic standard systems of semiclassical PDEs
\citep[see][]{Spohn_Teufel,Martinez_Sordoni,Combescure_Robert_book}, the
general evolution is split at leading order in $\eps$ along the eigenmodes
and, in each of them, governed by Hamiltonians given by the eigenvalues of
the symbols. This idea arises in the early 30-s in the context of the
Born-Oppenheimer approximation for molecular dynamics~\cite{born1927adp}. It
is at the roots of our scheme: we will prove that, at leading order
in~$\eps$, the evolution driven by $\boldsymbol U_\eps(t)$ in each mode
reduces to those dictated by the operators $\{\lambda^\eps_j\}_j$.

A crucial point is to combine the adiabatic mode splitting of eigenvalues
$\{\lambda^\eps_j\}_j$ with semiclassical analysis: the semiclassical structure of
$\mathscr{H}_{\mathcal{P}_{\varepsilon}}$ yields the fact that for each $j$,
\begin{equation*}
  \lambda^{\varepsilon}_j = \mathrm{Op}^{}_{\varepsilon}(\lambda_j)\;,
\end{equation*}
where $x\mapsto\lambda_j(x)$ are classical symbols (in class $S$)\footnote{To be precise,
  the quantization $\mathrm{Op}_{\varepsilon}$ used here is not the Weyl quantization,
  but rather the Wick quantization; this does not affect the semiclassical
  picture.}; each symbol $\lambda_j$ can in turn be seen as a classical
Hamiltonian, inducing a flow $t\mapsto x_j(t)$ in the phase space $\mathfrak{X}_{\mathcal{P}_0}$. These
classical flows enter in the propagation of \emph{coherent structures}: an initial
$\mathcal{P}_{\varepsilon}$--Gaussian wave-packet centered around $y\in \mathfrak{X}_{\mathcal{P}_0}$ is split (up to
order $\sqrt{\varepsilon}$ errors) in multiple Gaussian wave-packets, one for each mode,
centered around $x_j(t)$ (where the initial condition for all flows is the
same, namely $y$); on~$\mathcal{S}$, the adiabatic coupling acts, again up to
$\sqrt{\varepsilon}$-errors, as a ``superadiabatic projection''. More precisely, in $\mathscr{H}_{\mathcal{S}}\otimes
\mathscr{H}_{\mathcal{P}_{\varepsilon}}$-norm topology, given an $\mathcal{S}$-state vector $\psi=\sum_j^{}\alpha_j \psi_j$, we
have 
\begin{equation*}
  \boldsymbol{U}_{\varepsilon}(t)\bigl(\psi \otimes \lvert y_{\varepsilon}\rangle\bigr)\overset{\mathscr{H}_{\mathcal{S}}\otimes \mathscr{H}_{\mathcal{P}_{\varepsilon}}}{=} \sum_j^{} \alpha_j \pi_j\bigl(x_j(t)\bigr)\mathcal{R}_j(t) \psi_j \otimes  \lvert \widetilde{x_j(t)}_{\varepsilon}\rangle + \boldsymbol{O}_t(\sqrt{\varepsilon})\;.
\end{equation*}
Here: on $\mathcal{S}$, $\pi_j\bigl(x_j(t)\bigr)$ is the aforementioned superadiabatic
projector, that acts on the $j$-th mode $\psi_j$ evolved by the effective
adiabatic unitary evolution $\mathcal{R}_j(t)$ (see \textsection\ref{sectioncoherentstates} for
detailed definitions); while, on $\mathcal{P}_{\varepsilon}$, the state $\lvert
\widetilde{x_j(t)}_{\varepsilon}\rangle$ is a so-called \emph{squeezed coherent state}
\citep[see, \emph{e.g.},][]{Combescure_Robert_book}, centered around $x_j(t)$
(see again \textsection\ref{sectioncoherentstates}). It turns out that the squeezed
coherent states $\lvert\widetilde{x_j(t)}_{\varepsilon}\rangle$ converge semiclassically to Dirac
measures as standard ones:
\begin{equation*}
  \lvert\widetilde{x_j(t)}_{\varepsilon}\rangle\langle \widetilde{x_j(t)}_{\varepsilon}\rvert\xrightarrow[\varepsilon\to 0]{\mathcal{F}} \delta_{x_j(t)}\;,
\end{equation*}
and therefore, as long as $x_j(t)\neq x_k(t)$ for $j\neq k$ (and in our concrete
model this is true, $\varepsilon$-uniformly, at any time), the quantum superpositions
are suppressed in the limit $\varepsilon\to 0$ (with exponential rate). Hence if we
define
\begin{gather*}
  \boldsymbol{\varrho}(t)= \boldsymbol{U}_{\varepsilon}(t)\bigl(\lvert\psi\otimes  y_{\varepsilon}\rangle\langle\psi\otimes y_{\varepsilon}|)\boldsymbol{U}^{*}_{\varepsilon}(t)\;,\\
  \boldsymbol{\varrho}_{\mathrm{ideal}}(t)= \sum_j^{} \lvert\alpha_j\rvert^2\; \bigl\lvert\pi_j\bigl(x_j(t)\bigr)\mathcal{R}_j(t) \psi_j\bigr\rangle\bigl\langle\pi_j\bigl(x_j(t)\bigr)\mathcal{R}_j(t) \psi_j\bigr\rvert \:\:\delta_{x_j(t)}\;,
\end{gather*}
we have, in the sense of semiclassical Fourier transforms:
\begin{equation}
  \label{eq:21}
  \boldsymbol{\varrho}(t)\overset{\mathcal{F}}{=} \boldsymbol{\varrho}_{\mathrm{ideal}}(t)+ \boldsymbol{O}_t(\sqrt{\varepsilon})\;.
\end{equation}
By \eqref{eq:21}, we mean as in~\eqref{eq:fourier_convergence}, that for all $\xi\in \Xi$ and $K$ compact on $\mathscr{H}_{\mathcal{S}}$,
\begin{equation*}
  \mathrm{tr}_{\mathcal{S}}\bigl(K\widehat{\boldsymbol{\varrho}(t)}(\xi)\bigr) = \mathrm{tr}_{\mathcal{S}}\bigl(K\widehat{\boldsymbol{\varrho}_{\mathrm{ideal}}(t)}(\xi)\bigr) + O_{t,K}(\sqrt{\varepsilon})\;,
\end{equation*}
where $O_{t,K}(\sqrt{\varepsilon})\in \mathbb{C}$ and there exists a constant $C_t>0$ such that
\begin{equation*}
  \lvert O_{t,K}(\sqrt{\varepsilon})  \rvert_{}^{}\leq C_t \lVert K  \rVert_{}^{} \sqrt{\varepsilon}\;,
\end{equation*}
with $\lVert K \rVert_{}^{}$ the operator norm. However, this is \emph{not} (yet)
an answer to Question~1: %\ref{qstn:1}: 
in fact, neither
\begin{equation*}
  \pi_j\bigl(x_j(t)\bigr)\mathcal{R}_j(t) \psi_j= \psi_j
\end{equation*}
nor
\begin{equation*}
  \pi_j\bigl(x_j(t)\bigr)\mathcal{R}_j(t) \psi_j \; \perp\; \pi_k\bigl(x_k(t)\bigr)\mathcal{R}_k(t) \psi_k
\end{equation*}
for $j\neq k$. Hence, neither collapse nor Born rules are satisfied by this
measurement process.

\subsubsection{Long-time asymptotics}
\label{sec:long-time-asympt}

This problem is solved by taking a probe that scatters (at least in the
semiclassical limit), and an interaction that vanishes for long times. We
also need to choose a suitable asymptotic time $t_{\varepsilon}$: since
$\boldsymbol{O}_t(\sqrt{\varepsilon})$ above is exponential in $t$, we need to have
\begin{equation*}
  t_{\varepsilon}= -\frac{\ln \varepsilon}{C}\;, 
\end{equation*}
with $C>0$ large enough. This way, on $\mathcal{S}$ we have, in trace class norm (see \textsection\ref{sec:longtime}),
\begin{equation*}
  \lim_{t\to \infty}\pi_j(x_j(t))\mathcal{R}_j(t) = \lvert \psi_j\rangle\langle\psi_j\rvert\;.
\end{equation*}
On $\mathcal{P}_0$, scattering is usually implemented in the following sense: on each
mode, the classical Hamiltonian generating the flow $x_j(t)$ is of the form
\begin{equation*}
  \lambda_j(x)= q_0(x) + h_{j}(x)\;,
\end{equation*}
where $q_0(x)= \langle x , h_0 x \rangle_{\mathfrak{X}_{\mathcal{P}_0}}$ is a quadratic form with \emph{dispersive}
self-adjoint generator $h_0$. By having a dispersive generator, we mean that
we require at least that the limit
\begin{equation*}
  \lim_{t\to \infty} e^{ith_0}x_j(t) = x^{\infty}_j
\end{equation*}
exists on $\mathfrak{X}_{\mathcal{P}_0}$ for a suitable class of initial conditions (and defines
the $j$-th asymptotic state of the probe). We must also have  $x_j^{\infty}\neq
x^{\infty}_k$ for all $j\neq k$, and indeed in the concrete model we prove in \textsection\ref{sec:longtime} that $\lVert
x^{\infty}_j-x^{\infty}_k \rVert_{\mathfrak{X}_{\mathcal{P}_0}}^{}\geq C>0$ (uniformly with respect to $\varepsilon$).

\subsection{Main result}
With the setup described above, we can choose the measurement process as
\begin{equation}
  \label{eq:22}
  \boldsymbol{U}_{\varepsilon}= e^{i \frac{t_{\varepsilon}}{\varepsilon} 1\otimes \mathrm{Op}^{}_{\varepsilon}(q_0)} \boldsymbol{U}_{\varepsilon}(t_{\varepsilon})\;;
\end{equation}
such a unitary operator describes the evolution of $\mathcal{S}$ and the \emph{asymptotic
  state} of the probe $\mathcal{P}_{\varepsilon}$, at time $t_{\varepsilon}$. Furthermore, we set the
zero-of-the-scale state $\Gamma_{\varepsilon}$ on $\mathcal{P}_{\varepsilon}$ to be the coherent state
\begin{equation*}
  \Gamma_{\varepsilon}= \lvert y_{\varepsilon}\rangle\langle y_{\varepsilon}\rvert\;,
\end{equation*}
with $y\in \mathfrak{X}_{\mathcal{P}_0}$ chosen suitably among the class of initial conditions that
yield distinct asymptotic states. By tuning the measurement device in such a
way, and by waiting a time $t_{\varepsilon}$ for the measurement process to take place,
we can finally answer positively to Question~1.
\begin{thm}[Solution to the measurement problem within error limits]
  \label{thm:6}
  For any density matrix $\varrho$ on $\mathcal{S}$, let us define
  \begin{equation*}
    \boldsymbol{\varrho}= \boldsymbol{U}_{\varepsilon}\bigl(\varrho\otimes \lvert y_{\varepsilon}\rangle\langle y_{\varepsilon}\rvert\bigr) \boldsymbol{U}^{*}_{\varepsilon}\;,
  \end{equation*}
  and the state-valued measure
  \begin{equation*}
    \boldsymbol{\varrho}_{\mathrm{ideal}}= \sum_j^{} \langle \psi_j  , \varrho\,\psi_j \rangle_{\mathcal{S}}\:\lvert\psi_j\rangle\langle\psi_j\rvert\: \delta_{x_j^{\infty}}\;,
  \end{equation*}
  satisfying $\mathcal{P}$-observability, Born and collapse rules (cf. {\bf [$\mathcal P$-O]}, {\bf [BR]} and {\bf [CR]}, \textsection\ref{sec:measurement-problem}). 
  Then,
  \begin{gather*}
    \boldsymbol{\varrho}\:\overset{\mathcal{F}}{=}\: \boldsymbol{\varrho}_{\mathrm{ideal}}+ \boldsymbol{O}\bigl(\lvert\ln\varepsilon\rvert^{-1/2}\bigr)\;, \;\; \mbox{i.e.}\;\; 
     \boldsymbol{\varrho} \xrightarrow[\varepsilon\to 0]{\mathcal{F}} \boldsymbol{\varrho}_{\mathrm{ideal}}\;.
  \end{gather*}
\end{thm}
Thanks to the explicit bounds on $\boldsymbol{O}\bigl(\lvert\ln\varepsilon\rvert^{-1/2}\bigr)$,
we can further prove that, when restricted to~$\mathcal{S}$, we have a stronger
convergence, namely convergence in trace class norm. 

\begin{corollary}\label{cor:2}
Consider the noncommutative
conditional expectation  $\mathbb{E}[\,\varrho\,\vert\, \mathcal{O}\,]=\boldsymbol{\varrho}_{\rm ideal}\bigr\rvert_{\mathcal{S}}$ defined in~\eqref{expect},  
%we have   convergence
  %to $\mathbb{E}[\,\varrho\,\vert\, \mathcal{O}\,]$ 
%  in trace-class norm:
 % \begin{equation*}
 % $\displaystyle{  
 \[
 \lim_{\varepsilon\to 0} \mathrm{tr}\,\Bigl\lvert\, \boldsymbol{\varrho}\bigr\rvert_{\mathcal{S}} - \mathbb{E}[\,\varrho\,\vert\, \mathcal{O}\,] \,\Bigr\rvert=0\;.
 \]
 %}$
 % \end{equation*}
\end{corollary}

\bigskip

\bigskip

\bigskip

\begin{remark}\label{rem1}
  $\phantom{i}$
  \begin{enumerate}[(i)]
    \setlength{\itemsep}{2mm}
  \item\label{item:4} Theorem~\ref{thm:6} is proved below for the observable
    $\mathcal{O}=\sigma_z=
    \begin{pmatrix}
      1&0\\
      0&-1
    \end{pmatrix}
    $ on $\mathscr{H}_{\mathcal{S}}=\mathbb{C}^2$. The extension to any observable on $\mathscr{H}_{\mathcal{S}}=\mathbb{C}^n$ would
    require only minimal modifications (mostly to take care of the eventual
    degeneracies in the spectrum). The case of observables with purely
    discrete spectrum on separable Hilbert spaces shall also be easy to
    reproduce with the same technique. Observables with continuous spectrum,
    on the other hand, might require considerable additional effort.
  \item\label{item:5} In the proof, we use a scalar field with
    non-relativistic dispersion as the semiclassical probe $\mathcal{P}_{\varepsilon}$. Although
    this choice is very natural from a physical standpoint, any semiclassical
    system (including, \emph{e.g.}, a system of material points) should work
    equally well.
  \item\label{item:6} We expect that neither the time scale $t_{\varepsilon}\sim-\ln\varepsilon$ nor
    the rate of convergence $\boldsymbol{O}(\lvert\ln \varepsilon\rvert^{-1/2})$ to be
    optimal. As a matter of fact, the aforementioned experiments on the
    validity of Born rule suggest that the optimal rate of convergence shall
    be much faster. To get an improved rate, one shall be able to improve on
    the logarithmic semiclassical asymptotic time scale; this could be
    achieved either by proving that the asymptotic time scale is independent
    from the semiclassical scale \citep[cf.][where indeed the classical and
    asymptotic time limits commute]{ammari2021arxiv}, or by taking into
    account the decoherence caused on the probe by the environment, that is
    known to improve, in suitable cases, semiclassical asymptotic times
    \citep{hernandez2025cmp}. This shall also relate, from an experimental
    perspective, with the process of amplification of the measurement
    outcome.
  \end{enumerate}
\end{remark}

The no-go theorem \ref{thm:5} (called ``insolubility proof of the measurement
problem'' in the literature) which motivates our work, greatly influenced mathematicians, physicists,
and philosophers, in so that most attempts of overcoming the problem focused
on new interpretations or axioms that modify (in a more or less radical way)
standard quantum mechanics, and especially the fact that dynamical processes
are always deterministic. In view of the fact that quantum mechanics is one
of the most predictive, accurate, and powerful theories about the physical
world, these attempts always create a great debate, and they are often met
with a good dose of skepticism (and this is why we highlight the fact that
our solution is fully within the standard quantum mechanical
framework). Nonetheless, let us briefly discuss (some of) them.
\smallskip 

%\subsubsection*
\noindent {\it Weakening objectiveness}.
A first approach to the measurement problem is to try to ``bypass'' it: the
measurement problem becomes less relevant -- or dissipates completely -- if we
weaken the notion of objectively assigning values to quantum
observable. These weakened notions of objectiveness have been formulated
within modal and many-world interpretations of quantum mechanics, however as
far as we know they are mostly speculative: the interested reader might refer
to the monograph \citep{MR1464691}, where there even is an attempt at
quantifying how much objectiveness is weakened, through value-attribution
rules.
\smallskip

%\subsubsection*{
\noindent {\it Measurements via stochastic processes}. 
Another important aspect of the measurement problem concerns the
incompatibility of a unitary quantum evolution with the wavefunction collapse
on the measured system. A way of achieving the wavefunction collapse is by
modeling the measurement coupling (at least restricted to $\mathcal{S}$) with a
suitable \emph{continuous-time stochastic process}. This idea has been
formulated by N.~Gisin \citep{MR741987, MR761901} from a physical
perspective, while, notably, J.~Fröhlich and Z.~Gang recently proved that it
indeed induces the wavefunction collapse asymptotically \citep{MR4699907},
see also \citep{froehlich2026sigal} where a model of the two-slit experiment
is discussed. As noted by Gisin in his original paper, it remains unclear
what deterministic ``hidden'' dynamics -- if any -- would induce this
stochastic process, but he infers it shall be [\dots]\emph{very sensitive to
  the microscopic state of the system plus apparatus} \citep{MR741987};
Fröhlich, Gang and Pizzo motivate the appearance of such a stochastic process
through the so-called ``ETH approach'' to quantum mechanics: a proposed
completion of the quantum mechanical first principles developed by the
authors and collaborators \citep[see][]{MR4381182, MR4927824}. It would be
interesting to understand if such a process might emerge from our
semiclassical scheme, however this goes beyond the scope of the present
paper. Let us also mention that there are other attempts at overcoming the
measurement problem through stochastic processes, such as the stochastic
filtering equation upon which V.~Belavkin's \emph{nondemolition
  superselection principle} is based \citep[see][]{belavkin2005arxiv}.
\smallskip 

%\subsubsection*{
\noindent {\it Quantum-to-classical transitions of measurement devices via decoherence}.
The last approach that we want to discuss is the one that might be seen as
closer in spirit to our solution: the quantum-to-classical transition of the
measurement device induced by coupling with an environment -- the so-called
\emph{decoherence}. Decoherence is a quantum-to-classical transition induced
by the environment: in real lab situations, a system is never truly isolated,
and it is the exchange with the environment that makes it become classical,
or more properly a classical statistical mixture of rank one projectors
\citep[see, \emph{e.g.}][]{zee1970foop,habib1998prl}. Since its formulation,
decoherence has been strongly linked with measurements and the measurement
problem \citep{zurek2005sp} (see also \citep{glauber1986anyas} for a
different approach, using quantum optical amplifiers), and it is still widely
used in physics in this context \citep[see,
\emph{e.g.}][]{Everitt_2011,2md2-nxz3}. It is, however, unclear how the
coupling with an environment $\mathcal{E}$ could make the system $\mathcal{S}+\mathcal{P}+\mathcal{E}$ yield a
quantum-to-classical transition for $\mathcal{P}$ only, as long as the full microscopic
dynamics of the total system is unitary: von Neumann already noted in his
book \citep[][\textsection6.3 of the english translation]{MR223138} that adding a third
or further subsystems shall not change significantly the scheme and its
outcome; even more so, due the added layer of complication introduced by
adding the environment, there is no rigorous and complete formulation of such
a ``decoherence measurement scheme''. In the light of our present results, it
appears clear that the correspondence principle itself is a sufficient -- and
in our opinion, the \emph{most natural} -- actor for the quantum-to-classical
transition of probes $\mathcal{P}$ in measurement schemes; coupling with an environment
-- for example by making the probe dynamics a Markov semigroup -- might help in
improving the semiclassical accuracy: as recently proved in
\citep{hernandez2025cmp, MR4954260}, it allows to go beyond Ehrenfest times.

\subsection{Organisation of the paper}
The rest of the paper is organized as follows. Since the mathematics of
quantum measurements has been widely developed, we start in
\textsection\ref{sec:measurement-schemes} by briefly overview the mathematics of quantum
measurements, and the related bibliography, in order to state our result
within the formalism of measurement. This converges to a reformulation of
Theorem~\ref{thm:6} (see Theorem~\ref{thm:2}).  Then,
in~\textsection\ref{sec:meas-scheme-thro} we introduce the precise model used to prove
Theorem~\ref{thm:6} in the frame described in (i) of Remarks~\ref{rem1}
above, and discuss its physical relevance, togeter with the schedule of its
proof. As already explained, the proof of Theorem~\ref{thm:6} (and thus
Theorem~\ref{thm:2}) consists in two steps: in~\textsection\ref{sectioncoherentstates}
we study the adiabatic limit at finite time, and in \textsection\ref{sec:longtime} we
focus on the long-time asymptotics, completing the proof. Finally, in
Appendix~\ref{appendixcommutatorestimates}, we collect the proof of the
Theorem~\ref{thm:6} and Proposition~\ref{prop:superpsuppr} stated in this
introduction.

\subsection*{Acknowledgments}
\label{sec:acknowledgements}

{% \footnotesize
  The authors warmly thank R.~Carles, V.~Jak\v{s}ić, and
  M.~Zworski for fruitful discussions.  MF acknowledges the support of PNRR
  Italia Domani and Next Generation EU through the ICSC National Research
  Centre for High Performance Computing, Big Data and Quantum Computing. MF
  and RG also acknowledge the support: by the Italian Ministry of University
  and Research (MUR) through the grant “Dipartimento di Eccellenza 2023-2027”
  of Dipartimento di Matematica, Politecnico di Milano, and the PRIN 2022
  grant ``OpeN and Effective quantum Systems (ONES)''; and by the Research
  Institute for Mathematical Sciences (RIMS) of the University of Kyoto
  through the RIMS research project 2026 ``The Mathematical Roads to
  QFT''. CFK acknowledges the support of the Région Pays de la Loire via the
  Connect Talent Project HiFrAn 2022 07750, the grants ANR-23-CE40-0016
  (OpART) and ANR-25-CE40-7296 (La Gabare).  CFK and RG benefit from the
  France 2030 program, Centre Henri Lebesgue ANR-11-LABX-0020-01.}

\section{The Mathematics of Quantum Measurements}
\label{sec:measurement-schemes}

In this section we briefly overview the mathematical theory of quantum
measurements, following mainly \citep[][see also references
therein]{MR223138,davies1970cmp,ozawa1984jmp,busch2016tmp}; then we
formalize, in the same spirit, the ideal measurements with a classical probe,
introduced in \textsection\ref{sec:bohr-type-ideal}, and connect the two through the semiclassical limit;
finally, we review the existing literature on the measurement problem.

\subsection{von Neumann schemes}
\label{sec:meas-schem-foll}

Following von Neumann's original ideas, one would like to consider the action
of the measurement process on $\mathcal{S}$ as conditional expectations (as
in~\eqref{expect}), that in turn are a consequence of the coupling with the
probe $\mathcal{P}$ carrying out the measurement. While von Neumann's theory works well
for observables with purely discrete spectrum, it requires some
generalization for observables that have a continuous spectrum. This has been
done by E.~B.\ Davies and J.~T.\ Lewis \citep{davies1970cmp} with the
introduction of \emph{completely positive instruments} on $\mathcal{S}$, while some
years later M.\ Ozawa \citep{ozawa1984jmp} related $\mathcal{S}$-instruments with
measurement processes on $\mathcal{S}+\mathcal{P}$.

Let us start with generalizing the concept of observable. By the spectral
theorem, there is a one-to-one correspondence between $\mathcal{S}$-observables $\mathcal{O}$
(\emph{i.e}, densely defined self-adjoint operators) and spectral measures
$\mathrm{d}\mathrm{O}$, the latter being \emph{projection-valued}
measures. We use the roman letter~${\rm O}$ for denoting the spectral measure
of the observable $\mathcal O$, written in calligraphic style. For the
purpose of defining conditional expectations and joint distributions, it is
useful \citep{davies1970cmp} to relax the assumption that
$\mathrm{d}\mathrm{O}$ is projection-valued, and only require it to be valued
in positive bounded operators. This leads to the concept of
\emph{semiobservable} \citep{ozawa1984jmp}.

\begin{definition}[Semiobservable]
  \label{def:9}
  A \emph{semiobservable} $\mathrm{O}$ on $\mathcal{S}$ is a positive operator-valued
  measure (POVM): let $(\Omega,\mathfrak{S})$ be a measurable space, and $\mathscr{B}(\mathscr{H}_{\mathcal{S}})$ the space
  of bounded linear operators on $\mathscr{H}_{\mathcal{S}}$; then $0$ is a map from $\mathfrak{S}$ onto
  $\mathscr{B}(\mathscr{H}_{\mathcal{S}})$
%  \begin{equation*}
 %   \mathrm{O}: \mathfrak{S}\longrightarrow \mathscr{B}(\mathscr{H}_{\mathcal{S}})\;,
  %\end{equation*}
  such that
  \begin{itemize}
  \item $\mathrm{O}(\Omega)=1$;
  \item $\forall \Sigma\in \mathfrak{S}$, $0\leq \mathrm{O}(\Sigma)\leq \mathrm{O}(\Omega)$;
  \item $\forall \{\Sigma_i\}_{i\in \mathbb{N}}\subset \mathfrak{S}$ mutually disjoint, then in the strong operator
    topology
    \begin{equation*}
      \mathrm{O}\Bigl(\bigcup_{i\in \mathbb{N}}^{} \Sigma_i\Bigr)= \sum_{i\in \mathbb{N}}^{} \mathrm{O}(\Sigma_i)\;.
    \end{equation*}
  \end{itemize}
  Furthermore, if $\forall \Sigma\in \mathfrak{S}$, $\mathrm{O}(\Sigma)$ is an orthogonal projection
  (\emph{i.e.}, it is self-adjoint and such that $\mathrm{O}(\Sigma)^2=\mathrm{O}(\Sigma)$), then $\mathrm{O}$ is
  an \emph{observable}.
\end{definition}

Semiobservables are more general noncommutative random variables than
observables, and the expectation of states (\emph{i.e.}, positive trace-class
operators on $\mathscr{H}_{\mathcal{S}}$) conditioned to semiobservables are called
(CP-)\emph{instruments}: let $\mathcal{L}^1(\mathscr{H}_{\mathcal{S}})$ be the ideal of trace class operators and
$\mathscr{B}\bigl(\mathcal{L}^1(\mathscr{H}_{\mathcal{S}})\bigr)$ the space of bounded linear maps on it, then any
semiobservable $\mathrm{O}$, induces a (completely) positive map
$$\mathfrak{I}_{\mathrm{O}}: \mathfrak{S}\to \mathscr{B}\bigl(\mathcal{L}^1(\mathscr{H}_{\mathcal{S}})\bigr)$$ through
\begin{equation*}
  \mathrm{tr}\bigl(\mathfrak{I}_{\mathrm{O}}(\Sigma)\varrho\bigr)= \mathrm{tr}\bigl(\mathrm{O}(\Sigma)\varrho\bigr)\;,
\end{equation*}
which satisfies the two properties of  the next definition.
%satisfies
%\begin{itemize}
%\item $\mathrm{tr}(\mathfrak{I}_{\mathrm{O}}(\Omega)\varrho)= \mathrm{tr}\,\varrho$,
%\item $\forall \{\Sigma_i\}_{i\in \mathbb{N}}\subset \mathfrak{S}$ mutually disjoint, then in the strong operator
 % topology of $\mathscr{B}(\mathcal{L}^1)$,
  %\begin{equation*}
    %\mathfrak{I}%_{\mathrm{O}}\Bigl(\sum_{i\in \mathbb{N}}^{} \Sigma_i\Bigr)= \sum_{i\in \mathbb{N}}^{} \mathfrak{I}_{\mathrm{O}}(\Sigma_i)\;.
 % \end{equation*}
%\end{itemize}
%This leads to the following definition, and then characterization of the relation
%between instruments and observables.
\begin{definition}[Instruments]
  \label{def:10}
  A (CP-)\emph{instrument} $\mathfrak{I}$ on $\mathcal{S}$ is a (completely) positive map
  \begin{equation*}
    \mathrm{O}: \mathfrak{S}\longrightarrow \mathscr{B}\bigl(\mathcal{L}^1(\mathscr{H}_{\mathcal{S}})\bigr)
  \end{equation*}
  such that $\mathrm{tr}\bigl(\mathfrak{I}(\Omega)\varrho\bigr)= \mathrm{tr}(\varrho)$, and for all $\{\Sigma_i\}_{i\in \mathbb{N}}\subset \mathfrak{S}$ mutually disjoint, then in the strong operator
    topology of $\mathscr{B}\bigl(\mathcal{L}^1(\mathscr{H}_{\mathcal{S}})\bigr)$,
    \begin{equation*}
      \mathfrak{I}\Bigl(\bigcup_{i\in \mathbb{N}}^{} \Sigma_i\Bigr)= \sum_{i\in \mathbb{N}}^{} \mathfrak{I}(\Sigma_i)\;.
    \end{equation*}
%  \begin{itemize}
%  \item $\mathrm{tr}\bigl(\mathfrak{I}(\Omega)\varrho\bigr)= \mathrm{tr}(\varrho)$;
%  \item $\forall \{\Sigma_i\}_{i\in \mathbb{N}}\subset \mathfrak{S}$ mutually disjoint, then in the strong operator
 %   topology of $\mathscr{B}\bigl(\mathcal{L}^1(\mathscr{H}_{\mathcal{S}})\bigr)$,
%    \begin{equation*}
%      \mathfrak{I}\Bigl(\sum_{i\in \mathbb{N}}^{} \Sigma_i\Bigr)= \sum_{i\in \mathbb{N}}^{} \mathfrak{I}(\Sigma_i)\;.
 %   \end{equation*}
  %\end{itemize}
\end{definition}

The relation between instruments and observables is characterized by the next
definition.

\begin{proposition}[{\citep[][Thm.\ 1]{davies1970cmp}}]
  \label{prop:5}
  For every instrument $\mathfrak{I}$, there exists a unique semiobservable~$\mathrm{O}$
  such that for all $\Sigma\in \mathfrak{S}$,
  \begin{equation*}
    \mathrm{tr}\bigl(\mathfrak{I}(\Sigma)\varrho\bigr)= \mathrm{tr}\bigl(\mathrm{O}(\Sigma)\varrho\bigr)\;;
  \end{equation*}
  and $\mathfrak{I}$ is said to be $\mathrm{O}$-compatible. Conversely, for every
  semiobservable $\mathrm{O}$ there exist at least one $\mathrm{O}$-compatible
  CP-instrument.
\end{proposition}

Instruments encode conditional expectations of quantum states: given a state
$\varrho$ on $\mathcal{S}$, we can define the conditional expectation $\mathbb{E}[\,\varrho\,\vert\, \mathrm{O}\in \Sigma\,]$ by
\begin{equation*}
  \mathbb{E}[\,\varrho\,\vert\,\mathrm{O}\in \Sigma\,]:= \mathfrak{I}(\Sigma)\varrho\;,
\end{equation*}
where $\mathfrak{I}$ is any $\mathrm{O}$-compatible instrument (and clearly
$\mathbb{E}[\,\varrho\,\vert\,\mathrm{O}\,]=\mathfrak{I}(\Omega)\varrho$). 

Since we interpret these conditional
expectations as the effect of a quantum measurement process, we would like to
put in relation instruments and measurement processes. In order to do that,
let us consider a different kind of conditional expectation: the conditional
expectation in a bipartite quantum system. Let $\mathscr{H}_{\mathcal{S}}\otimes \mathscr{H}_{\mathcal{P}}$ be a tensor
Hilbert space, and let $\varrho$ be a positive trace class operator on $\mathscr{H}_{\mathcal{S}}$ and
$\Gamma$ a positive trace class operator on $\mathscr{H}_{\mathcal{P}}$; finally, let $\boldsymbol{T}\in \mathscr{B}(\mathscr{H}_{\mathcal{S}}\otimes \mathscr{H}_{\mathcal{P}})$. Then we define the conditional
expectation $\mathbb{E}_{\Gamma}(\boldsymbol{T})\in \mathscr{B}(\mathscr{H}_{\mathcal{S}})$ by
\begin{equation*}
  \mathrm{tr}_{\mathcal{S}}(\mathbb{E}_{\Gamma}(\boldsymbol{T})\varrho)= \mathrm{tr}_{\mathcal{S}+\mathcal{P}}\bigl(\boldsymbol{T}(\varrho\otimes \Gamma)\bigr)\;.
\end{equation*}
We use such a conditional expectation to define \emph{measurement schemes} for
semiobservables.
\begin{definition}[Measurement scheme]
  \label{def:1}
  $\phantom{i}$
  \begin{enumerate}[(i)]
  \item Let us consider the quadruple
  \begin{equation*}
    \mathscr{M}=(\mathscr{H}_{\mathcal{P}},\Gamma,\mathcal{Z},\boldsymbol{U})\;,
  \end{equation*}
  where the Hilbert space $\mathscr{H}_{\mathcal{P}}$ is the \emph{state space of the probe}, the
  $\mathcal{P}$-state $\Gamma$ is the \emph{zero-of-the-scale state}, the $\mathcal{P}$-observable $\mathcal Z$
  (projection-valued measure on $(\Omega,\mathfrak{S})$) is the \emph{pointer observable}, and the
  unitary operator $\boldsymbol{U}$ on $\mathscr{H}_{\mathcal{S}}\otimes \mathscr{H}_{\mathcal{P}}$ is the \emph{measurement coupling}. Then
  $\mathscr{M}$ is a \emph{measurement scheme for the $\mathcal{S}$-semiobservable $\mathrm{O}$} (POVM on
  $(\Omega,\mathfrak{S})$) if and only if for all $\Sigma\in \mathfrak{S}$,
  \begin{equation*}
    \mathrm{O}(\Sigma)= \mathbb{E}_{\Gamma}\Bigl(\boldsymbol{U}^{*}\bigl(1\otimes \mathrm{Z}
(\Sigma)\bigr)\boldsymbol{U}\Bigr)\;.
\end{equation*}
\item\label{item:7} Two measurement schemes $\mathscr{M}=(\mathscr{H}_{\mathcal{P}},\Gamma,\mathrm{Z},\boldsymbol{U})$ and $\mathscr{M}'=(\mathscr{H}_{\mathcal{P}},\Gamma',\mathrm{Z}',\boldsymbol{U}')$ are \emph{equivalent} if and only if for
all $T\in \mathscr{B}(\mathscr{H}_{\mathcal{S}})$ and $\Sigma\in \mathfrak{S}$,
\begin{equation*}
  \mathrm{tr}_{\mathcal{S}+\mathcal{P}}\Bigl(\bigl(T\otimes
  \mathrm{Z}
  (\Sigma)\bigr)\boldsymbol{U}(\varrho\otimes \Gamma)\boldsymbol{U}^{*}\Bigr)=\mathrm{tr}_{\mathcal{S}+\mathcal{P}}\Bigl(\bigl(T\otimes \mathrm{Z}'
  (\Sigma)\bigr)\boldsymbol{U}'(\varrho\otimes \Gamma')(\boldsymbol{U}')^{*}\Bigr)\;.
\end{equation*}
  \end{enumerate}
% (i)  
% \noindent  (ii) 
\end{definition}

Note that if $\mathscr M=(\mathscr{H}_{\mathcal{P}},\Gamma,\mathcal{Z},\boldsymbol{U})$
is a measurement scheme, then by part (i) of the above definition
\[
\mathrm{tr}_{\mathcal S}(O(\Sigma)\varrho)=
\mathrm{tr}_{\mathcal S}\Bigl(\mathbb E_{\Gamma}(\boldsymbol U^* \bigl(1\otimes {\mathrm Z}
  %\mathrm{Z}
  (\Sigma)\bigr)\boldsymbol U) \varrho\Bigr)= 
\mathrm{tr}_{\mathcal{S}+\mathcal P}\Bigl(\bigl(1\otimes {\mathrm Z}
  %\mathrm{Z}
  (\Sigma)\bigr)\boldsymbol{U}(\varrho\otimes \Gamma)\boldsymbol{U}^{*}\Bigr).
\]

As proved by M.\ Ozawa \citep{ozawa1984jmp}, there is a one-to-one
correspondence between equivalence classes of measurement schemes and
$\mathrm{O}$-compatible completely positive instruments.
\begin{proposition}[{\citep[][Thm.\ 5.1]{ozawa1984jmp}}]
  \label{prop:6}
  There is a one-to-one correspondence between equivalence classes of
  $\mathrm{O}$-measurement schemes and families of $\mathrm{O}$-compatible
  CP-instruments, given for all $\Sigma\in \mathfrak{S}$ by
  \begin{equation*}
    \mathrm{tr}_{\mathcal{S}+\mathcal{P}}\Bigl(\bigl(T\otimes \mathrm{Z}
(\Sigma)\bigr)\boldsymbol{U}(\varrho\otimes \Gamma)\boldsymbol{U}^{*}\Bigr)= \mathrm{tr}_{\mathcal{S}}\bigl(T\,\mathfrak{I}(\Sigma)\varrho\bigr)\;.
  \end{equation*}
\end{proposition}
Combining Propositions~\ref{prop:5}~and~\ref{prop:6}, it follows that it is
always possible to define a measurement scheme $\mathscr{M}$, for any semiobservable
$\mathrm{O}$. Wigner's theorem as stated in Theorem~\ref{thm:5}, and the
measurement problem (for a discrete countable value space $\Omega$) can be easily
reformulated as follows in this language.

\begin{thm}[Measurement problem]
  \label{thm:7}
  There exists no scheme $\mathscr{M}=(\mathscr{H}_{\mathcal{P}},\Gamma,\mathrm{Z},\boldsymbol{U})$ such that
  \begin{equation*}
    \boldsymbol{U}(\varrho\otimes \Gamma)\boldsymbol{U}^{*}= \sum_j (1\otimes \mathrm{Z}_j)\boldsymbol{U}(\varrho\otimes \Gamma)\boldsymbol{U}^{*}(1\otimes \mathrm{Z}_j)\;,  
  \end{equation*}
  where $\mathrm{Z}_j$ are the pointer spectral projectors associated with
  the $j$-th spectral value.
\end{thm}
In other words, while it is always possible to construct a measurement scheme
that satisfies Born and collapse rules {\bf [BR]} and {\bf [CR]} (since, indeed, there exist 
instruments that satisfy them), it cannot also satisfy $\mathcal{P}$-observability {\bf [$\mathcal P$-O]} at
the same time; this is why the measurement problem is also called the problem
of \emph{objectification} (of measurement outcomes) \citep{busch2016tmp}.
Theorem \ref{thm:7} is the result of years of refinement,
that started with the already cited paper by Wigner
\citep{wigner1963ajp,PhysRevD.9.2321, PhysRevD.2.2783, brown1986foop,
  MR1445744, MR1029208, MR1419313}.

\subsection{Bohr schemes}
\label{sec:bohrmeas}

Let us adapt the above definition of measurement scheme to the case in which
the probe is classical; to the best of our knowledge, this was not done
previously, at least in a systematic and rigorous way.

Since the system $\mathcal{S}$ stays quantum, the notion of semiobservables and
instruments carry over unchanged. The notion of measurement $\mathscr{M}$, on the other
hand, must be changed to accommodate a classical probe $\mathcal{P}_0$:
\begin{itemize}
  \setlength{\itemsep}{2mm}
\item[(i)] The quantum state space must be substituted with a {\it classical phase space}
  $\mathfrak{X}_{\mathcal{P}_0}$, that is the continuous dual of a pre-Hilbert and topological
  vector space as in \textsection\ref{sec:bohr-type-ideal}.
\item[(ii)] The zero-of-the-scale quantum state is replaced by a {\it classical
  zero-of-the-scale state}, \emph{i.e.}\ a Borel Radon probability measure $\mu$ on
  $\mathfrak{X}_{\mathcal{P}_0}$.
\item[(iii)] To define a {\it classical pointer observable}, we shall identify phase space
  points $x\in \mathfrak{X}_{\mathcal{P}_0}$ with spectral values $\omega\in \Omega$; we do so by supposing
  there exists a measurable embedding $\imath: \Omega\to \mathfrak{X}_{\mathcal{P}_0}$, where the latter is
  endowed with the Borel $\sigma$-algebra. A pointer observable is then any
  injective random variable $Z:\imath(\Omega)\to \mathbb{R}$: in fact, from a physical
  perspective, it is the numerical value $z$ of the pointer observable $Z$
  that is concretely measured, and through~$z$ one can identify the phase
  space point $\imath(\omega)$ so that $Z(\imath(\omega))=z$, and subsequently the spectral value
  $\omega$ of the semiobservable $\mathrm{O}$.

  \medskip

  \noindent What most closely resembles the quantum spectral measure $\textrm{d}Z$ is
  simply the collection $\mathrm{X}=\{\chi_{\Sigma}\}_{\Sigma\in \mathfrak{S}}$ of characteristic
  functions
  \begin{equation*}
    \chi_{\Sigma}(x)=
    \begin{cases}
      \: 1&\text{if }x\in \imath(\Sigma)\\
      \: 0&\text{if }x\in \mathfrak{X}_{\mathcal{P}_0}\smallsetminus \imath(\Sigma)
    \end{cases}\;.
  \end{equation*}
\item[(iv)] The unitary measurement coupling $\boldsymbol{U}$ cannot be substituted
  with its most natural counterpart, \emph{i.e.}\ an Hamiltonian flow. In fact, as
  discussed in \textsection\ref{sec:bohr-type-ideal}, we need a map that transforms a Dirac measure into a
  measure supported on many distinct points. Therefore, we take the {\it classical
  measurement coupling} to be a linear map $\boldsymbol{\Upsilon}$ on
  $\mathcal{S}$-state-valued measures on $\mathfrak{X}_{\mathcal{P}_0}$\footnote{A linear map
    $\boldsymbol{\Upsilon}$ on state valued measures maps linearly the state-valued
    measure $\mathrm{d}\boldsymbol{\mu}(x)= \varrho(x)\mathrm{d}\mu(x)$ to another
    state valued measure
    $\boldsymbol{\Upsilon}\bigl(\mathrm{d}\boldsymbol{\mu}(x)\bigr)=
    \varrho_{\boldsymbol{\Upsilon}}(x)\mathrm{d}\mu_{\boldsymbol{\Upsilon}}(x)$.}.
\end{itemize}
The above discussion leads to the following
definition.

\bigskip

\begin{definition}[Bohr measurement scheme]
  \label{def:4}
  $\phantom{i}$
  \begin{enumerate}[(i)]
  \item Let us consider the quadruple
  \begin{equation*}
    \mathscr{M}_0= (\mathfrak{X}_{\mathcal{P}_0}, \mu, \mathrm{X},\boldsymbol{\Upsilon})\;,
  \end{equation*}
  where the phase space $\mathfrak{X}_{\mathcal{P}_0}=\Xi_{\mathcal{P}_0}'$ is the continuous dual of the
  pre-Hilbert and topological vector space $\Xi_{\mathcal{P}_0}$, $\mu$ a Radon Borel
  probability measure on $\mathfrak{X}_{\mathcal{P}_0}$, $\mathrm{X}$ the family of
  $\mathfrak{X}_{\mathcal{P}_0}$-characteristic-functions associated to a measurable embedding
  $\imath:\Omega\to \mathfrak{X}_{\mathcal{P}_0}$, and $\boldsymbol{\Upsilon}$ a map on $\mathcal{S}$-state-valued measures on
  $\mathfrak{X}_{\mathcal{P}_0}$. Then $\mathscr{M}_0$ is a \emph{Bohr measurement scheme for the $\mathcal{S}$-semiobservable $\mathrm{O}$}
  if and only if for all $\Sigma\in \mathfrak{S}$, and $\mathcal{S}$-state $\varrho$,
  \begin{equation*}
    \mathrm{tr}_{\mathcal{S}}(\mathrm{O}(\Sigma)\varrho)= \mathrm{tr}_{\mathcal{S}}\Bigl(\int_{\mathfrak{X}_{\mathcal{P}_0}}^{} \chi_{\Sigma}(x)\,\boldsymbol{\Upsilon}\bigl(\varrho\,\mathrm{d}\mu(x)\bigr)\Bigr)\;.
  \end{equation*}
\item\label{item:8}   Two Bohr measurement schemes $\mathscr{M}_0,\mathscr{M}_0'$ are said to be equivalent if for all
$T\in \mathscr{B}(\mathscr{H}_{\mathcal{S}})$,
\begin{equation*}
  \mathrm{tr}_{\mathcal{S}}\Bigl(T\,\int_{\mathfrak{X}_{\mathcal{P}_0}}^{}\chi_{\Sigma}(x)\,\boldsymbol{\Upsilon}\bigl(\varrho\,\mathrm{d}\mu(x)\bigr)\Bigr)=\mathrm{tr}_{\mathcal{S}}\Bigl(T\,\int_{\mathfrak{X}_{\mathcal{P}_0}}^{}\chi'_{\Sigma}(x)\,\boldsymbol{\Upsilon}'\bigl(\varrho\,\mathrm{d}\mu'(x)\bigr)\Bigr)\;.
\end{equation*}
\item\label{item:9} Furthermore, a Bohr measurement scheme $\mathscr{M}_0$ \emph{solves the measurement problem}
  (for a discrete countable value space $\Omega=\{\omega_0,\omega_1,\dotsc,\omega_j,\dotsc\}$) if and only if
  \begin{equation*}
    \boldsymbol{\Upsilon}(\varrho \mathrm{d}\mu(x)) = \sum_j^{} \mathfrak{I}(\{\omega_j\})\varrho \, \mathrm{d}\delta_{\imath(\omega_j)}(x)\;,
  \end{equation*}
  for some $\mathrm{O}$-compatible instrument $\mathfrak{I}$.
  \end{enumerate}
\end{definition}

One can exhibit conditions for 
correspondence between equivalence classes of Bohr measurement schemes and
$\mathrm{O}$-compatible instruments.

\begin{proposition}
  \label{prop:7}
  Suppose there exists an embedding of $(\Omega,\mathfrak{S})$ into a separable Hilbert
  space. Then, there is a one to one correspondence between equivalence
  classes of Bohr measurement schemes of $\mathrm{O}$, and $\mathrm{O}$-compatible instruments.
\end{proposition}
\begin{proof}
  If $\varrho\mapsto \int_{\mathfrak{X}_{\mathcal{P}_0}}^{}\chi_{\Sigma}(x)\boldsymbol{\Upsilon}\bigl(\varrho \mathrm{d}\mu(x)\bigr)$
  defines an instrument, then by Proposition~\ref{prop:5} it is an
  $\mathrm{O}$-compatible instrument. Let us now prove that the collection
  $\{\mathfrak{I}_{\mathscr{M}_0}(\Sigma)\}_{\Sigma\in \mathfrak{S}}$, with
  \begin{equation*}
    \mathfrak{I}_{\mathscr{M}_0}(\Sigma) \,\cdot \, := \int_{\mathfrak{X}_{\mathcal{P}_0}}^{}\chi_{\Sigma}(x)\boldsymbol{\Upsilon}\bigl(\,\cdot \, \mathrm{d}\mu(x)\bigr)\;,
  \end{equation*}
  is an instrument. We first observe that indeed $\mathfrak{I}_{\mathscr{M}_0}(\Sigma)\in \mathscr{B}(\mathcal{L}^1(\mathscr{H}_{\mathcal{S}}))$.
  In fact, it is easily verified to be linear, and it is bounded:
  $\boldsymbol{\Upsilon}$ is a linear map of state-valued measures, so for all $\Sigma\in \mathfrak{S}$ and all $\varrho\in \mathcal{L}^1(\mathscr{H}_{\mathcal{S}})$,
  \begin{equation*}
    \lVert \mathfrak{I}_{\mathscr{M}_0}(\Sigma)\,\varrho  \rVert_{\mathcal{L}^1}^{}\leq 4\lVert \varrho  \rVert_{\mathcal{L}^1}^{}\;,
  \end{equation*}
  since any trace class operator can be decomposed into a linear combination
  of positive trace-class operators. In addition, by the general properties
  of state-valued measures \citep[see][]{correggi2022quasi}, the map
  $\mathfrak{I}_{\mathscr{M}_0}(\Sigma)$ is positive, and
  \begin{equation*}
    \mathrm{tr}\bigl(\mathfrak{I}(\Omega)\,\varrho\bigr)= \mathrm{tr}(\varrho)\;.
  \end{equation*}
  The $\sigma$-additivity with respect to $\mathfrak{S}$ is also a straightforward
  consequence of the fact that $\mathfrak{I}_{\mathscr{M}_0}(\Sigma)\,\varrho$ comes from a (vector) integral.

  Furthermore, since the semiobservable $\mathrm{O}$ associated to the
  instrument $\mathfrak{I}_{\mathscr{M}_0}$ by
  \begin{equation*}
    \mathrm{tr}\bigl(\mathfrak{I}_{\mathscr{M}_0}(\Sigma)\,\varrho\bigr)= \mathrm{tr}\bigl(\mathrm{O}(\Sigma)\,\varrho\bigr)
  \end{equation*}
  is unique, it follows that the correspondence between equivalence classes
  of Bohr measurements and $\mathrm{O}$-compatible instruments is one-to-one.
  
  It
  remains to show that given an instrument $\mathfrak{I}$, there exists a Bohr
  measurement scheme $\mathscr{M}_0$ such that $\mathfrak{I}_{\mathscr{M}_0}=\mathfrak{I}$. This is easily done by choosing
  $\boldsymbol{\Upsilon}$ so that, denoting by $\boldsymbol{\mu}_{\varrho}$ the state-valued
  measure
 % \begin{equation*}
   $ \mathrm{d}\boldsymbol{\mu}_{\varrho}(x)= \varrho \,\mathrm{d}\mu(x)\;$,
%  \end{equation*}
  we have that for all $\Sigma\in \mathfrak{S}$,
  %\begin{equation*}
   $ \boldsymbol{\Upsilon}\bigl(\boldsymbol{\mu}_{\varrho}\bigr)(\imath(\Sigma))= \mathfrak{I}(\Sigma)\varrho\;$.
%  \end{equation*}
\end{proof}
\begin{corollary}
  \label{cor:3}
  For any semiobservable $\mathrm{O}$ on $\mathcal{S}$, there exists a Bohr
  measurement scheme~$\mathscr{M}_0$. If the value space $\Omega$ is discrete and countable,
  there always exists a measurement scheme that solves the measurement
  problem.
\end{corollary}

\subsection{Semiclassical measurement schemes and reformulation of the main result}
\label{sec:bridging-von-neuman}

As discussed in \textsection\ref{sec:corr-princ}, we shall link von Neumann and Bohr schemes through the
correspondence principle and semiclassical analysis.

Let $\mathscr{M}_0$ be a Bohr measurement scheme, with phase space $\mathfrak{X}_{\mathcal{P}_0}=\Xi_{\mathcal{P}_0}'$; then a
triple $(\mathscr{H}_{\mathcal{P}},\Gamma,\boldsymbol{U})$ (cf.\ Definition~\ref{def:1}) is \emph{$\varepsilon$-compatible}
with $\mathscr{M}_0$ if $\mathscr{H}_{\mathcal{P}}$ carries a (irreducible) representation of the
Heisenberg group $\mathbf{H}(\Xi_{\mathcal{P}_0})$ with central character $\varepsilon>0$. Since we cannot
overcome, in view of Theorem~\ref{thm:7}, the measurement problem at the quantum
level, the pointer observable $\mathcal{Z}$ of Definition~\ref{def:1} (even if it exists) is
not the right semiclassical object to look at. If $\mathrm{O}$ has a
countable and discrete value space $\Omega=\{\omega_0,\omega_1,\dotsc,\omega_j,\dotsc\}$, we can define an
approximate pointer observable $\mathrm{Z}_{\varepsilon}$ as a function
$\mathrm{Z}_{\varepsilon}:\Omega\to \mathscr{B}(\mathscr{H}_{\mathcal{P}})$ that generates, in the limit $\varepsilon\to 0$ an
$\mathrm{O}$-compatible instrument. This leads to the following definition.

\begin{definition}[Semiclassical measurement scheme]
  \label{def:5}
(i)  Let
  \begin{equation*}
    \mathscr{M}_0=(\mathfrak{X}_{\mathcal{P}_0}, \mu, \mathrm{X},\boldsymbol{\Upsilon})
  \end{equation*}
  be a Bohr scheme for the $\mathcal{S}$-semiobservable $\mathrm{O}$. A
  \emph{($\mathscr{M}_0$-)semiclassical measurement scheme $\mathscr{M}_{\varepsilon}$} is an $\varepsilon$-compatible
  triple
  \begin{equation*}
    \mathscr{M}_{\varepsilon}=(\mathscr{H}_{\mathcal{P}_{\varepsilon}},\Gamma_{\varepsilon},\boldsymbol{U}_{\varepsilon})\;,
  \end{equation*}
  where $\mathscr{H}_{\mathcal{P}_{\varepsilon}}$ is a Hilbert space carrying a representation of
  $\mathbf{H}(\Xi_{\mathcal{P}_0})$ with central character $\varepsilon$, $\Gamma_{\varepsilon}$ a positive trace-class
  operator on $\mathscr{H}_{\mathcal{P}_{\varepsilon}}$, $\boldsymbol{U}_{\varepsilon}$ a unitary operator on $\mathscr{H}_{\mathcal{S}}\otimes
  \mathscr{H}_{\mathcal{P}_{\varepsilon}}$, such that:
  \begin{itemize}
    \setlength{\itemsep}{2mm}
  \item $\Gamma_{\varepsilon}\xrightarrow[\varepsilon\to 0]{\mathcal{F}} \mathrm{d}\mu(x)$ ;
  \item $\forall \varrho\in \mathcal{L}^1(\mathscr{H}_{\mathcal{S}})$,
    \begin{equation*}
      \boldsymbol{U}_{\varepsilon}\bigl(\varrho\otimes \Gamma_{\varepsilon}\bigr)\boldsymbol{U}^{*}_{\varepsilon}\xrightarrow[\varepsilon\to 0]{\mathcal{F}} \boldsymbol{\Upsilon}\bigl(\varrho\, \mathrm{d}\mu(x)\bigr)
    \end{equation*}
  \end{itemize}
(ii)  A $\mathscr{M}_0$-semiclassical measurement scheme $\mathscr{M}_{\varepsilon}$ \emph{solves the measurement problem within error limits}
  if and only if $\mathscr{M}_0$ solves the measurement problem; it solves the
  measurement problem \emph{with rate~$f_{\varepsilon}$} if, in the sense of
  semiclassical Fourier transforms (cf.\ \eqref{eq:21}),
  \begin{equation*}
    \boldsymbol{U}_{\varepsilon}\bigl(\varrho\otimes \Gamma_{\varepsilon}\bigr)\boldsymbol{U}^{*}_{\varepsilon}\:\overset{\mathcal{F}}{=}\: \boldsymbol{\Upsilon}\bigl(\varrho\, \mathrm{d}\mu(x)\bigr) + \boldsymbol{O}(f_{\varepsilon})\;.
  \end{equation*}

  %\medskip

\noindent (iii) If the $\mathrm{O}$-value-space $\Omega=\{\omega_0,\omega_1,\dotsc,\omega_j,\dotsc\}$ is
  countably discrete, we define a \emph{strong semiclassical measurement
    scheme $\mathscr{M}_{\varepsilon}^+$} to be a quadruple
  \begin{equation*}
    \mathscr{M}_{\varepsilon}^+= (\mathscr{H}_{\mathcal{P}_{\varepsilon}},\Gamma_{\varepsilon},\mathrm{Z}_{\varepsilon},\boldsymbol{U}_{\varepsilon})\;,
  \end{equation*}
  where $(\mathscr{H}_{\mathcal{P}_{\varepsilon}},\Gamma_{\varepsilon},\boldsymbol{U}_{\varepsilon})$ is a $\mathscr{M}_0$-semiclassical
  measurement scheme and
  \begin{equation*}
    \mathrm{Z}_{\varepsilon}:\Omega\to \mathscr{B}(\mathscr{H}_{\mathcal{P}})\; \text{ with }0\leq \mathrm{Z}_{\varepsilon}(\omega_j)\leq 1\text{ for all }j\in \mathbb{N}\;,
  \end{equation*}
  an \emph{approximate pointer}, on which the following additional
  convergence holds: for all $ j\in \mathbb{N}$ and $\varrho\in \mathcal{L}^1(\mathscr{H}_{\mathcal{S}})$, in trace-class norm
  on $\mathscr{H}_{\mathcal{S}}$,
  \begin{equation}
    \label{eq:29bis}
    \lim_{\varepsilon\to 0} \,\mathbb{E}_{\Gamma_{\varepsilon}}\Bigl(\boldsymbol{U}_{\varepsilon}^{*}\bigl(1\otimes \mathrm{Z}(\omega_j)\bigr)\boldsymbol{U}_{\varepsilon}\Bigr)\,\varrho = \int_{\mathfrak{X}_{\mathcal{P}_0}}^{}\chi_{\{\omega_j\}}(x)  \boldsymbol{\Upsilon}\bigl(\varrho\,\mathrm{d}\mu(x)\bigr)=\mathfrak{I}_{\mathscr{M}_0}\bigl(\{\omega_j\}\bigr)\,\varrho\;.
  \end{equation}
\end{definition}
In view of the above, Theorem~\ref{thm:6} can be reformulated as follows.
\begin{thm}
  \label{thm:2}
  Let $\mathrm{O}$ be an observable with non-degenerate spectrum\footnote{The
    $\mathrm{O}$-value space is $\Omega=\{\omega_1,\dotsc,\omega_n\}$, and $\mathrm{O}(\{\omega_j\})=
    \lvert\psi_j\rangle\langle\psi_j\rvert$ for all $j$, with $\{\psi_j\}_j\subset \mathscr{H}_{\mathcal{S}}$ orthonormal basis. Note also that semiobservables and observables on $\mathbb C^n$ differ by a multiplicative constant.} on
  $\mathscr{H}_{\mathcal{S}}= \mathbb{C}^n$. Then, there exists a strong semiclassical measurement scheme
  that solves the measurement problem within error limits, and with rate
  $\lvert\ln \varepsilon\rvert^{-1/2}$.
\end{thm}

The sections below are devoted to the proof of Theorem~\ref{thm:2}, in the
case $n=2$ (the proof for a generic $n$ is identical); as remarked below
Theorem~\ref{thm:6}, the extension to observables with degenerate and
countable discrete spectrum shall present no particular difficulty, while the
extension to continuous spectra (where semiobservables shall play an
important role) is left for future work.

\section{Measuring a two-level system through semiclassical fluorescence}
\label{sec:meas-scheme-thro}

In this section we construct concretely both the ideal Bohr scheme that
solves the measurement problem, and the strong semiclassical scheme of
Theorem \ref{thm:2}, for $n=2$: we want to measure a self-adjoint $2\times 2$
matrix $\mathcal O$ on a vector space $\psi\in \mathscr H_{\mathcal S}=\mathbb
C^2$. A choice of orthonormal basis lets us assume $\mathcal O=\sigma_z$. As we
will see later, the measurement will be based on the use of a bosonic field
that models fluorescence.  Throughout the rest of the paper, we might omit the
argument of functional spaces $L^p, \mathscr{S},\, \dotsc$ as it should be clear from the
context.

\subsection{A concrete Bohr scheme based on square-integrable functions}
%$\mathscr{M}_0(\sigma_z)$}
\label{sec:bohr-subscheme}

We choose the probe $\mathcal{P}$ to be a force-carrying field, modeling an
\emph{electromagnetic pulse}. Its classical version $\mathcal{P}_0$ has an
infinite-dimensional phase space, thus we choose
\begin{equation*}
  \Xi_{\mathcal{P}_0}= \mathscr{S}(\mathbb{R}^3)\;,
\end{equation*}
pre-Hilbert through the $L^2$-inner product $\langle \cdot , \cdot \rangle_2$. Although its
continuous dual is $\mathscr{S}'(\mathbb{R}^3)$, we will only use measures that are concentrated
on $L^2(\mathbb{R}^3)\subset \mathscr{S}'(\mathbb{R}^3)$, so with a slight abuse of notation we can set
\begin{equation}
  \label{eq:23}
  \mathfrak{X}_{\mathcal{P}_0}:= L^2(\mathbb{R}^3)\;.
\end{equation}
We will denote by $\langle\cdot,\cdot\rangle_2$ the inner product in $L^2(\mathbb R)$. We assume
that initially the probe is in a completely deterministic state, \emph{i.e.}\
that the initial measure is a Dirac measure:
\begin{equation}
  \label{eq:26}
  \mu:= \delta_{\underline{u}}\;,
\end{equation}
with $\underline{u}\in L^2$ that will be chosen suitably below. The value space
for $\mathrm{O}$ is taken to be $\Omega=\{+,-\}$ (representing, \emph{e.g.}, the
two values that the spin can take for a spin one-half particle); it is embedded
into $\mathfrak{X}_{\mathcal{P}_0}$ by
\begin{equation*}
  \imath(+):= u^{\infty}_+\;,\; \imath(-):= u_-^{\infty}\;,
\end{equation*}
where the points $u^{\infty}_{\pm}$ are the asymptotic states associated to the
adiabatic scattering processes, that we define precisely below in \textsection\ref{sec:von-neumann-part} (cf.\
\textsection\ref{sec:long-time-asympt}). Furthermore, we choose the normalized eigenvectors of
$\mathrm{O}$ as the canonical orthonormal basis of $\mathscr{H}_{\mathcal{S}}$: in other words,
$\mathrm{O}$ can be always represented as the third Pauli matrix $\sigma_z=
\begin{pmatrix}
  1&0\\0&-1
\end{pmatrix}
$ on $\mathscr{H}_{\mathcal{S}}=\mathbb{C}^2$; for convenience, let us denote $\lvert+\rangle=(1,0)$ and
$\lvert-\rangle=(0,1)$. It follows that $\mathrm{X}$ is generated by
\begin{equation}
  \label{eq:25}
  \chi_+(x):=\chi_{\{+\}}(x)=
  \begin{cases}
    \: 1&\text{if }x=u^{\infty}_+\\
    \:0&\text{otherwise}
  \end{cases}\;,\quad   \chi_-(x):=\chi_{\{-\}}(x)
  \begin{cases}
    \: 1&\text{if }x=u^{\infty}_-\\
    \:0&\text{otherwise}
  \end{cases}\;.
\end{equation}
Finally, the map $\boldsymbol{\Upsilon}$ is any map on state-valued measures that
satisfies, for any $\varrho\in \mathcal{L}^1(\mathbb{C}^2)$, $\varrho\geq 0$, $\mathrm{tr}\varrho=1$,
\begin{equation}
  \label{eq:27}
  \boldsymbol{\Upsilon}\bigl(\varrho\, \mathrm{d}\delta_{\underline{u}}\bigr):=  \langle +  , \varrho\, + \rangle_{\mathscr{H}_{\mathcal{S}}}\, \lvert +\rangle\langle +\rvert \,\delta_{u^{\infty}_+} + \langle -  , \varrho\, - \rangle_{\mathscr{H}_{\mathcal{S}}}\, \lvert -\rangle\langle -\rvert \,\delta_{u^{\infty}_-}\;.
\end{equation}
This leads to the following definition.
\begin{definition}[Concrete Bohr scheme]
  \label{def:11}
  The Bohr measurement scheme $\mathscr{M}_0(\sigma_z)$ for $\sigma_z$ is the quadruple
  \begin{equation*}
    \mathscr{M}_0(\sigma_z)= (\mathfrak{X}_{\mathcal{P}_0}, \mu, \mathrm{X}, \boldsymbol{\Upsilon})
  \end{equation*}
  defined through \eqref{eq:23}--\eqref{eq:27} above. The scheme $\mathscr{M}_0(\sigma_z)$
  solves the measurement problem.
\end{definition}

\subsection{A semiclassical scheme based on symmetric Fock space}% $\mathscr{M}^+_\varepsilon(\sigma_z)$}
\label{sec:von-neumann-part}

Let us now define the concrete quadruple $(\mathscr{H}_{\mathcal{P}_{\varepsilon}},
\Gamma_{\varepsilon},\mathrm{Z}_{\varepsilon},\boldsymbol{U}_{\varepsilon})$. 
%We prove that it is indeed a strong
%$\mathscr{M}_0(\sigma_z)$-semiclassical scheme in \textsection\ref{sec:limit-varepsilonto-0} below.

\subsubsection{The bosonic Fock space}
Firstly, we shall quantize the probe $\mathcal{P}_0$, \emph{i.e.}\ choose a representation of
the Heisenberg group $\mathbf{H}(\mathscr{S}(\mathbb{R}^3))$. We make use of the most common
representation, the so-called \emph{Fock representation} (or free field
representation). The Hilbert space for this representation is the symmetric
Fock space:
\begin{equation}
  \label{eq:24}
  \mathscr{H}_{\mathcal{P}_{\varepsilon}}:= \mathscr{F}\bigl(L^2(\mathbb{R}^3)\bigr)= \bigoplus_{n \in \N} L^2_{\mathrm{sym}}(\R^{3n})\;,
\end{equation}
where $L^2_{\mathrm{sym}}(\mathbb{R}^{3n})$ stands for the space of square-integrable functions
$\Psi(x_1,\dotsc,x_n)$ that are symmetric under the permutation of any triple of
variables $x_i$ and $x_j$ in $\mathbb{R}^3$ (and $L^2_{\mathrm{sym}}(\mathbb{R}^0)\equiv \mathbb{C}$). The Heisenberg
group is represented through the self-adjoint field operators $\phi_{\varepsilon}(\xi)$, $\xi\in
\mathscr{S}$, defined by: 
\[
\phi_{\varepsilon}(\xi)= \frac{1}{\sqrt{2}}\bigl(a^{*}_{\varepsilon}(\xi)+a_{\varepsilon}(\xi)\bigr)
\]
where
$a_{\varepsilon}(\cdot ),a^{*}_{\varepsilon}(\cdot )$ are the creation and annihilation operators:
satisfying the canonical commutation relations 
\[
[a_{\varepsilon}(\xi),a^{*}_{\varepsilon}(\xi')]= \varepsilon\langle
\xi , \xi' \rangle_2
\]
More precisely, the annihilation and creation operators $a,a^{*}:L^2\to
\mathrm{Cl}(\mathscr{F}(L^2))$ are maps from the one-particle space to the
closed operators in the Fock space defined, for all $f\in L^2$, on the common
core $\mathscr{F}_{\mathrm{fin}}$ by
\begin{gather*}
  \bigl(a(f)\Psi\bigr)_n(x_1,\dotsc,x_n)= \sqrt{n+1}\int_{\mathbb{R}^3}^{}\bar{f}(x)\Psi_{n+1}(x,x_1,\dotsc,x_n)  \mathrm{d}x\;,\\
  \bigl(a^{*}(f)\Psi\bigr)_n(x_1,\dotsc,x_n)= \frac{1}{\sqrt{n}} \sum_{j=1}^n f(x_j) \Psi_{n-1}(x_1,\dotsc,\hat{x}_j,\dotsc,x_n)\;,
\end{gather*}
where $\hat{x}_j$ stands for the fact that such variable is missing.  The
creation and annihilation $a(f)$ and $a^{*}(f)$ are adjoint one another, and
the semiclassical, $\varepsilon$-scaled creation and annihilation operators are
\begin{equation*}
  a_{\varepsilon}(f)= \sqrt{\varepsilon}\, a(f)\;,\; a_{\varepsilon}^{*}(f)= \sqrt{\varepsilon}\, a^{*}(f).
\end{equation*}
%(see \textsection\ref{sec:bosonic-fock-spaces} for additional details).
 The Heisenberg group is thence
represented by exponentiation of the field operators:
\begin{equation*}
  \pi(\xi,s)= e^{i\phi_{\varepsilon}(\xi)}e^{i\varepsilon s}\;.
\end{equation*}
Semiclassical analysis in the Fock space has been studied first in the
context of system of many bosons, in a series of seminal works by Z.\ Ammari
and F.\ Nier \citep{ammari2008,ammari2009jmp,ammari2011jmpa,ammari2015asns},
and further developed by one of the authors together with Z.\ Ammari
\citep{ammari2014jsp,ammari2017sima} and M.\ Correggi and M.\ Olivieri in
presence of additional degrees of freedom
\citep{correggi2022quasi,correggi2023apde}. 

\subsubsection{Anti-Wick quantization}
We make use mostly of one type of
quantization map: the so-called \emph{Anti-Wick (or Berezin)
  quantization}. Anti-Wick quantization has the advantages of preserving
positivity, and of continuously embedding smooth symbols (as a normed space
with the uniform norm) into bounded operators (with their operator norm). To
define it, let us first define \emph{cylindrical symbols}. A cylindrical symbol $a\in
\mathscr{S}(\mathfrak{X}_{\mathcal{P}_0})_{\mathrm{cyl}}$ is a function $a:\mathfrak{X}_{\mathcal{P}_0}\to \mathbb{C}$ so that there exists
an orthogonal projection $\mathbb{P}$, $\mathrm{rank}\,\mathbb{P}<+\infty$, and a function $b\in
\mathscr{S}(\mathbb{P}\mathfrak{X}_{\mathcal{P}_0})$ such that
\begin{equation*}
  a(x)= b(\mathbb{P}x)\;.
\end{equation*}
The definition might be easily extended to other function spaces such as
$\mathscr{D}(\mathfrak{X})_{\mathrm{cyl}}, L^1(\mathfrak{X})_{\mathrm{cyl}}$, $C_0(\mathfrak{X})_{\mathrm{cyl}}$.
\begin{definition}[Anti-Wick quantization]
  \label{def:12}
  Let $a=b\circ \mathbb{P} \in C_0(\mathfrak{X}_{\mathcal{P}_0})_{\mathrm{cyl}}$ be a rapidly decreasing cylindrical symbol. Then
  its \emph{anti-Wick quantization} is the bounded operator on $\mathscr{H}_{\mathcal{P}_{\varepsilon}}$
  defined by the Bochner integral
  \begin{equation*}
    \mathrm{Op}_{\varepsilon}^{\mathrm{a-Wick}}(a)= \int_{\mathbb{P} \mathfrak{X}_{\mathcal{P}_0}}^{} b(x) \lvert x_{\varepsilon}\rangle\langle x_{\varepsilon}\rvert  \mathrm{d}x\;, 
  \end{equation*}
  where $\lvert x_{\varepsilon}\rangle\in \mathscr{H}_{\mathcal{P}_{\varepsilon}}$ is the coherent state vector of
  \textsection\ref{sec:coherent}.
\end{definition}
Before stating the main properties of anti-Wick quantization, let us define
the \emph{partial trace} $\boldsymbol{\varrho}\bigr\rvert_{\mathcal{S}}$ (resp.\
$\boldsymbol{\varrho}\bigr\rvert_{\mathcal{P}_{\varepsilon}}$) of a state $\boldsymbol{\varrho}\in \mathcal{L}^1(\mathscr{H}_{\mathcal{S}}\otimes
\mathscr{H}_{\mathcal{P}_{\varepsilon}})$ to be the state on $\mathcal{S}$ (resp.\ $\mathcal{P}_{\varepsilon}$) defined by
\begin{equation*}
  \mathrm{tr}_{\mathcal{S}}\bigl(T\, \boldsymbol{\varrho}\bigr\rvert_{\mathcal{S}}\bigr)= \mathrm{tr}_{\mathcal{S}+\mathcal{P}_{\varepsilon}}\bigl((T\otimes 1)\boldsymbol{\varrho}\bigr)\;\;\;\;
  \biggl(\mbox{resp.}\:\mathrm{tr}_{\mathcal{P}_{\varepsilon}}\bigl(S\, \boldsymbol{\varrho}\bigr\rvert_{\mathcal{P}_{\varepsilon}}\bigr)= \mathrm{tr}_{\mathcal{S}+\mathcal{P}_{\varepsilon}}\bigl((1\otimes S)\boldsymbol{\varrho}\bigr)\:\biggr)
\end{equation*}
for all $T\in \mathscr{B}(\mathscr{H}_{\mathcal{S}})\bigl(S\in \mathscr{B}(\mathscr{H}_{\mathcal{P}_{\varepsilon}})\bigr)$. 

\begin{prop}
  \label{prop:1}
  The anti-Wick quantization is positive: 
 if  $a\geq 0$, we have  $\mathrm{Op}^{\mathrm{a-Wick}}_{\varepsilon}(a)\geq 0$. Moreover, for all  $a\in C_0(\mathfrak{X}_{\mathcal{P}_0})_{\mathrm{cyl}}$, we have
\[
\lVert \mathrm{Op}^{\mathrm{a-Wick}}_{\varepsilon}(a)  \rVert_{\mathscr{B}(\mathscr{H}_{\mathcal{P}_{\varepsilon}})}^{}\leq \lVert a
    \rVert_{\infty}^{}.
    \]
  \noindent Furthermore, let $\boldsymbol{\varrho}_{\varepsilon}\in \mathcal{L}^1(\mathscr{H}_{\mathcal{S}}\otimes \mathscr{H}_{\mathcal{P}_{\varepsilon}})$ be positive and trace-one,
  and let $\boldsymbol{\mu}$ be a $\mathcal{S}$-state-valued-measure on $\mathfrak{X}_{\mathcal{P}_0}$ so
  that
  \begin{equation*}
    \boldsymbol{\varrho}_{\varepsilon}\overset{\mathcal{F}}{=} \boldsymbol{\mu}+ \boldsymbol{O}(f_{\varepsilon})\;,
  \end{equation*}
  with some rate $f_{\varepsilon}\to 0$. Then, for all $a\in \mathscr{S}(\mathfrak{X}_{\mathcal{P}_0})_{\mathrm{cyl}}$, we have
  that, in trace-class norm,
  \begin{equation*}
    \Bigl(\mathrm{Op}^{\mathrm{a-Wick}}_{\varepsilon}(a) \boldsymbol{\varrho}_{\varepsilon}\Bigr)\Bigr\rvert_{\mathcal{S}}\overset{\mathcal{L}^1(\mathscr{H}_{\mathcal{S}})}{=} \int_{\mathfrak{X}_{\mathcal{P}_0}}^{}a(x)  \mathrm{d}\boldsymbol{\mu}(x) + \mathrm{O}(f_{\varepsilon})\;.
  \end{equation*}
\end{prop}

\subsubsection{Spin-bosons interactions}
Let us define the {\it initial state for the probe}, $\Gamma_{\varepsilon}$: let
$\underline{u}\in \mathfrak{X}_{\mathcal{P}_0}$ be the initial classical configuration of the probe,
cf.~\eqref{eq:26}, then
\begin{equation}
  \label{eq:28}
  \Gamma_{\varepsilon}:= \lvert \underline{u}_{\varepsilon} \rangle\langle \underline{u}_{\varepsilon}\rvert\;.
\end{equation}
To define the {\it measurement coupling} $\boldsymbol{U}_{\varepsilon}$, let us introduce the
spin-boson interaction on $\mathscr{H}_{\mathcal{S}}\otimes \mathscr{H}_{\mathcal{P}_{\varepsilon}}$. The spin-boson Hamiltonian is
formally defined as follows:
\begin{equation*}
  \boldsymbol{H}_{\varepsilon}(\mathfrak{g})= \sigma_z\otimes 1 + 1\otimes \mathrm{d}\Gamma_{\varepsilon}(-\Delta)+ \mathfrak{g}\, \sigma_x\otimes \phi_{\varepsilon}(g)\;;
\end{equation*}
where:
\begin{itemize}
  \setlength{\itemsep}{2mm}
\item[(i)] $\mathfrak{g}\in \mathbb{R}$ is the coupling constant, measuring the strength of the coupling
  between the system and the probe; the sign of $\mathfrak{g}$ does not determine the
  sign of the interaction, so for convenience we assume $\mathfrak{g}>0$.
\item[(ii)] $\sigma_x=
  \begin{pmatrix}
    0&1\\1&0
  \end{pmatrix}
  $ is the first Pauli matrix.
\item[(iii)] $g\in L^2(\mathbb{R}^3)$ is the charge distribution of $\mathcal{S}$
  (with respect to the force field, that is typically electromagnetic).
\item[(iv)] The second quantization operator $\mathrm{d}\Gamma_{\varepsilon}(-\Delta)$ models the
  Hamiltonian operator of a scalar field with non-relativistic dispersion
  (such a dispersion is chosen purely for mathematical convenience, a more
  physically realistic dispersion could be used, but the scattering theory on
  $\mathcal{P}_0$ shall be modified accordingly).
\end{itemize}

More precisely, let $T$ be a self-adjoint operator on $L^2$, with domain
$D(T)$. Then, we define $\mathscr{F}_{\mathrm{fin}}(T)$ as the set
\begin{multline*}
  \mathscr{F}_{\mathrm{fin}}(T)=\Bigl\{\Psi\in \mathscr{F}(L^2)\;,\; \exists \underline{n}\in \mathbb{N}\text{ such that } \forall n\geq \underline{n}\;,\; \Psi_n=0 \\\text{ and }\forall m<\underline{n}\;,\; \Psi_m\in \underbrace{D(T)\otimes_{\mathrm{s}}\dotsm\otimes_{\mathrm{s}}D(T)  }_m\Bigr\}\;,
\end{multline*}
and $\mathrm{d}\Gamma_{\varepsilon}(T)$ is the operator on
$\mathscr{F}(L^2)$, essentially self-adjoint on
$\mathscr{F}_{\mathrm{fin}}(T)$, defined as
\begin{equation*}
  \bigl(\mathrm{d}\Gamma_{\varepsilon}(T)\Psi\bigr)_n= \Bigl(\eps \sum_{j=1}^n\,1\otimes \dotsm\otimes 1\otimes \underbrace{T}_j\otimes 1\otimes \dotsm\otimes 1\Bigr)\Psi_n\;.
\end{equation*}
Here, 
we have denoted vectors $\Psi\in
\mathscr{F}(L^2)$ as
\begin{equation*}
  \Psi=(\Psi_0,\Psi_1\dotsc,\Psi_n,\dotsc)\;,\;\;\Psi_n\in L^2_n.
\end{equation*}
%The vector
%\begin{equation*}
 % \Omega=(1,0,\dotsc,0,\dotsc)
%\end{equation*}
%is called the \emph{Fock vacuum vector}; it belongs to an important dense
%subspace of vectors, the so-called \emph{finite particle vectors}
%\begin{equation*}
 % \mathscr{F}_{\mathrm{fin}}=\Bigl\{\Psi\in \mathscr{F}(L^2)\;,\; \exists \underline{n}\in \mathbb{N}\text{ such that } \forall n\geq \underline{n}\;,\; \Psi_n=0\Bigr\}\;.
%\end{equation*}

The spin-boson thus describes the interaction between the two-level system
$\mathcal{S}$ (for which the measured observable $\sigma_z$ is the free
Hamiltonian, so in this picture its eigenvalues correspond to energy levels),
and a force-carrying bosonic field $\mathcal{P}_{\varepsilon}$ whose
excitations are created and absorbed by $\mathcal{S}$, through the coupled
field operator $\sigma_x\otimes \phi_{\varepsilon}$.

In physics, the spin-boson model is referred to as ``ubiquitous''
\citep{leggett1987rmp}, as it describes many models of interest in quantum
optics, condensed matter and quantum chemistry, such as the interaction of a
qubit with its environment. Most notably, the idea of using the spin-boson to
address/test the measurement problem is not new \citep{leggett1985prl},
although this was done in a completely different way, by trying to consider
the two levels in $\mathcal{S}$ as ``macroscopically distinct'': the authors
argue that this cannot be the case, and thus $\mathcal{S}$ can indeed serve
as a \emph{bona fide} quantum system, that must be subjected to quantum
measurement.

%Although, as already noted, we postpone some precise (but somewhat standard)
%definitions about symmetric Fock spaces to
%\textsection\ref{sectioncoherentstates}, 
We are now in a position to state the
result concerning self-adjointness of $\boldsymbol{H}_{\varepsilon}(\mathfrak{g})$, whose
proof is standard \citep[see, \emph{e.g.},][for the essential
self-adjointness part]{falconi2015mpag}. We denote by
$\mathscr{F}_{\mathrm{fin}}$ the so-called \emph{finite particle vectors} in
the Fock space, \emph{i.e.}\ the vectors with only finitely many non-zero
components in the direct sum $\bigoplus_{n\in
  \mathbb{N}}L^2_{\mathrm{sym}}(\mathbb{R}^{3n})$.
\begin{proposition}
  \label{prop:2}
  For all $\mathfrak{g}\in \mathbb{R}$ and $g\in L^2(\mathbb{R}^3)$, the spin-boson Hamiltonian
  $\boldsymbol{H}_{\varepsilon}(\mathfrak{g})$ is essentially self-adjoint on $D(1\otimes \mathrm{d}\Gamma_{\varepsilon}(-\Delta))\cap
  \mathbb{C}^2\otimes \mathscr{F}_{\mathrm{fin}}$. If $g\in \dot{H}^{-1}(\mathbb{R}^3)\cap L^2(\mathbb{R}^3)$, then
  $\boldsymbol{H}_{\varepsilon}(\mathfrak{g})$ is self-adjoint on $D(1\otimes \mathrm{d}\Gamma_{\varepsilon}(-\Delta))$ and bounded
  from below.
\end{proposition}
The unitary evolution operator associated with the spin-boson Hamiltonian is
\begin{equation}\label{def:propagation_time}
  \boldsymbol{U}_{\varepsilon}(t)= e^{-i \frac{t}{\varepsilon}\boldsymbol{H}_{\varepsilon}(\mathfrak{g})}\;;
\end{equation}
where the $\varepsilon$-dependence reflects the so-called Born-Oppenheimer or adiabatic
regime, as we discussed in \textsection\ref{sec:adiabatic-coupling}; the unitary $\boldsymbol{U}_{\varepsilon}(t)$ is not,
by itself, the measurement coupling (as also discussed in
\textsection\ref{sec:adiabatic-coupling}); to obtain the measurement coupling, we shall
allow the field to scatter, and thus take a suitably long time $t\gg1$, passing
to the so-called \emph{interaction representation}.
\begin{equation*}
  \widetilde{\boldsymbol{U}}_{\varepsilon}(t)= e^{i \frac{t}{\varepsilon}1\otimes \mathrm{d}\Gamma_{\varepsilon}(-\Delta)} e^{-i \frac{t}{\varepsilon}\boldsymbol{H}_{\varepsilon}(\mathfrak{g})}\;.
\end{equation*}
We take a (large) semiclassical time $t_{\varepsilon}$ to be fixed precisely
below (as anticipated, $t_{\varepsilon}= \mathrm{O}(-\ln
\varepsilon)$); then, we define the measurement coupling $\boldsymbol{U}_{\varepsilon}$
by:
\begin{equation}
  \label{eq:24bis}
  \boldsymbol{U}_{\varepsilon}:= \widetilde{\boldsymbol{U}}_{\varepsilon}(t_{\varepsilon})\;.
\end{equation}
Observe that it could be possible to define another measurement scheme by
taking $\boldsymbol{U}_{\varepsilon}=\Omega^+_{\varepsilon}$, the latter being the
spin-boson wave operator (whenever it exists):
\begin{equation*}
  \boldsymbol{\Omega}^+_{\varepsilon}= \underset{t\to \infty}{\mathrm{s-lim}}\; \widetilde{\boldsymbol{U}}_{\varepsilon}(t)\;.
\end{equation*}
As a matter of fact, our choice approximates this wave operator in
logarithmic times, but only as $\varepsilon\to 0$. We plan to study whether
$\boldsymbol{U}_{\varepsilon}=\boldsymbol{\Omega}^+_{\varepsilon}$ could still define a semiclassical measurement scheme in
future works \citep[drawing inspiration from][where semiclassical scattering
is studied for a related model, but with a different
scaling]{ammari2021arxiv}; it might be that one needs to introduce the
interaction with the environment and thus a Lindblad-type dynamics, to be
able to take the limits $t\to \infty$ and $\varepsilon\to 0$ in any order.

\subsubsection{The approximate pointer} To define the approximate pointer, let us first discuss the classical
scattering theory yielded by the adiabatic mode splitting. The spin-boson
Hamiltonian $\boldsymbol{H}_{\varepsilon}(\mathfrak{g})$ can be written as the operator-valued
matrix
\begin{equation*}
  \boldsymbol{H}_\eps(\mathfrak g) = \begin{pmatrix}
    1 + \di \Gamma_\eps(-\Delta) & \mathfrak g \,\phi_\eps (g) \\ \mathfrak g \,\phi_\eps(g) & -1 + \di \Gamma_\eps(-\Delta)
  \end{pmatrix}\;,
\end{equation*}
whose formal matrix eigenvalues are given by
\begin{equation}\label{def:lambda_eps_pm}
  \lambda_{\pm}^\eps = \di \Gamma_\eps(-\Delta) \pm \sqrt{ 1 + \mathfrak g^2 \,\phi_\eps (g)^2}\;.
\end{equation}
If we substitute $\mathrm{d}\Gamma_{\varepsilon}(-\Delta)$ and $\phi_{\varepsilon}(g)$ with their classical
symbols
\begin{gather*}
  \mathrm{d}\Gamma_{\varepsilon}(-\Delta)\rightsquigarrow \langle\, \cdot \,  , -\Delta\, \cdot \, \rangle_2\;\;\mbox{and}\;\;
  \phi_{\varepsilon}(g) \rightsquigarrow  \sqrt{2}\, \mathrm{Re}\, \langle\, \cdot \,  , g \rangle_2,
\end{gather*}
we obtain two classical Hamiltonians
\begin{equation*}
  \lambda_{\pm}(u) = \langle u  , -\Delta u \rangle_2 \pm \sqrt{1+ 2\mathfrak{g}^2\, \mathrm{Re}\, \langle u  , g \rangle_2^2 }\;,
\end{equation*}
whose corresponding Hamilton equations read
\begin{equation*}
  i \partial_t{u}_\pm(t) = -\Delta u_\pm(t) \pm \mathfrak g^2 \frac{\mathrm{Re}\langle u_\pm(t),g \rangle_2\, g}{\sqrt{1+2 \mathfrak g^2 \mathrm{Re} \langle u_\pm(t),g\rangle_2^2}}\;.
\end{equation*}
From this we can see that as soon as the interaction between $\mathcal{S}$ and $\mathcal{P}_0$ is
turned on, there are two bifurcating trajectories in the phase space,
starting from $\underline{u}$: the (unique) solutions $t\mapsto u_{\pm}(t)$ to Cauchy
problems
\begin{equation}
  \label{eq:1}
  \begin{cases}
    \: i \partial_t{u}_\pm(t) = -\Delta u_\pm(t) \pm \mathfrak g^2 \frac{\mathrm{Re}\langle u_\pm(t),g \rangle_2\, g}{\sqrt{1+2 \mathfrak g^2 \mathrm{Re} \langle u_\pm(t),g\rangle^2_2}}\\
    \: u_{\pm}(0)=\underline{u}
  \end{cases}\;.
\end{equation}
The two solutions $u_+$ and $u_-$ are different at all times $t>0$, and they
both scatter: there exists $\mathfrak{g}_0>0$ such that for all $0<\mathfrak{g}<\mathfrak{g}_0$ and all
$\underline{u}\in \mathfrak{X}_{\mathcal{P}_0}$, there exist $u_{\pm}^{\infty}\in \mathfrak{X}_{\mathcal{P}_0}$ such that
\begin{equation*}
  \lim_{t\to \infty}\lVert u_{\pm}(t) - e^{it\Delta}u_{\pm}^{\infty}  \rVert_2^{}=0\;,
\end{equation*}
see Proposition~\ref{globalwellposedness}. Furthermore, suppose that $g\in \mathscr{S}$, and
let $\underline{u}\in L^2$ be such that there exists $t\in \mathbb{R}$ so that
$\mathrm{Re}\, \langle e^{it\Delta}\underline{u}  , g \rangle_2\neq 0$. Then, there exists
$C(\underline{u},g)>0$  and $\mathfrak{g}_0(\underline{u},g)>0$ such that for all
$0<\mathfrak{g}<\mathfrak{g}_0$,
\begin{equation*}
  \lVert u^{\infty}_+- u^{\infty}_-  \rVert_2^{}\geq C(\underline{u},g)\, \mathfrak{g}^2=:\delta\;,
\end{equation*}
see Proposition~\ref{asymptotic splitting}. The approximate pointer
observable then is chosen to be the one generated by
\begin{equation}
  \label{eq:28bis}
  \mathrm{Z}_{\varepsilon}^+:= \mathrm{Op}^{a-Wick}_{\varepsilon}(\chi^+_\delta)\;,\quad \mathrm{Z}_{\varepsilon}^-:=\mathrm{Op}^{a-Wick}_{\varepsilon}(\chi^-_\delta)\;,
\end{equation}
where $\chi^{\pm}_{\delta}\in \mathscr{D}(\mathfrak{X}_{\mathcal{P}_0})_{\mathrm{cyl}}$ are so that:
\begin{equation*}
  \chi^{\pm}_{\delta}= \tilde{\chi}^{\pm}_{\delta}\circ \mathbb{P}\;,\quad \mathrm{Ran}\mathbb{P}= \mathrm{span}\{u^{\infty}_+,u^{\infty}_-\}\;,\quad \mathscr{D}(\mathbb{P}\mathfrak{X}_{\mathcal{P}_0})\ni\tilde{\chi}^{\pm}_{\delta}(u)=
  \begin{cases}
    \: 1 &\text{if } \lVert u-u^{\infty}_{\pm}  \rVert_2^{}\leq \frac{\delta}{4}\\
    \: 0 &\text{if } \lVert u- u^{\infty}_{\pm}  \rVert_2^{}\geq \frac{\delta}{2}
  \end{cases}\;.
\end{equation*}
The choice above ensures that, in the limit $\varepsilon\to 0$, the two pointer
observables %$\mathrm{Z}_{\varepsilon}^{\pm}$
$Z_\eps=(Z_\eps^+,Z_\eps^-)$ act, respectively, as characteristic functions
of $u_{\pm}^{\infty}$ when acting on the outcome state-valued measure
$\boldsymbol{\Upsilon}(\varrho\, \mathrm{d}\delta_{\underline{u}})$.  \medskip

At the end of this paragraph, we are thus led to the following definition.
\begin{definition}[Concrete semiclassical scheme]
  \label{def:13}
  The strong semiclassical measurement scheme $\mathscr{M}_\varepsilon(\sigma_z)$ for $\sigma_z$ is the quadruple
  \begin{equation*}
    \mathscr{M}_\varepsilon(\sigma_z)= (\mathscr{H}_{\mathcal{P}_\varepsilon}, \Gamma_{\varepsilon}, \mathrm{Z}_{\varepsilon}, \boldsymbol{U}_{\varepsilon})
  \end{equation*}
  defined through \eqref{eq:24}, \eqref{eq:28}, \eqref{eq:24bis}, and
  \eqref{eq:28bis} above.
\end{definition}

\subsection{Schedule of the proof of the main result}
%: the limit $\varepsilon\to 0$}
\label{sec:limit-varepsilonto-0}
Theorem~\ref{thm:2} (and thus Theorem~\ref{thm:6}) is based on the fact that the scheme $\mathscr{M}_\varepsilon(\sigma_z)$ is a strong
$\mathscr{M}_0(\sigma_z)$-semiclassical scheme.
%in \textsection\ref{sec:limit-varepsilonto-0} below.
%Let us now briefly outline how Theorem~\ref{thm:2} is proved, in
%\textsection\ref{sectioncoherentstates}--\ref{sec:longtime} below. 
Without loss of
generality, we focus on an $\mathcal{S}$-vector-state $\psi= \alpha_+\lvert + \rangle +\alpha_-\lvert -\rangle$,
coupled together with the initial coherent vector $\lvert \underline{u}_{\varepsilon}\rangle$, to
form the initial vector $\boldsymbol{\Psi}_{\varepsilon}$ of $\mathscr{H}_{\mathcal{S}}\otimes \mathscr{H}_{\mathcal{P}_{\varepsilon}}$:
\begin{equation*}
  \boldsymbol{\Psi}_{\varepsilon}= \psi\otimes \lvert \underline{u}_{\varepsilon}\rangle\;.
\end{equation*}
Its finite time evolution is given by
\begin{equation*}
  \boldsymbol{\Psi}_{\varepsilon}(t)= \boldsymbol{U}_{\varepsilon}(t)\boldsymbol{\Psi}_{\varepsilon}\;.
\end{equation*}
Also, recall the formal eigenvalues $\lambda^{\varepsilon}_{\pm}$ of $\boldsymbol{H}_{\varepsilon}(\mathfrak{g})$,
and observe that, again formally, the latter can be written as
\begin{equation*}
  \boldsymbol{H}_{\varepsilon}(\mathfrak{g})= \lambda^{\varepsilon}_+\pi^{\varepsilon}_+ + \lambda^{\varepsilon}_-\pi^{\varepsilon}_-\;,
\end{equation*}
with
\begin{equation*}
  \pi^{\varepsilon}_{\pm}= \frac{1}{2}\biggl(1_{\mathcal{S}}\pm \frac{\sigma_z + \mathfrak{g}\, \sigma_x\otimes \phi_{\varepsilon}(\mathfrak{g})}{\sqrt{1+\mathfrak{g}^2\,\phi_{\varepsilon}(g)^2}}\biggr)\;.
\end{equation*}
These ``$\mathscr{H}_{\mathcal{S}}$-projectors'' have a semiclassical
counterpart
\begin{equation*}
  \pi_{\pm}(u)= \frac{1}{2}\biggl(1_{\mathscr{H}_{\mathcal{S}}}\pm \frac{\sigma_z + \sqrt{2}\,\mathfrak{g}\,  \mathrm{Re}\,\langle u  , g \rangle_2\, \sigma_x}{\sqrt{1+2\mathfrak{g}^2\,\langle u  , g \rangle_2^2}}\biggr)
\end{equation*}
that are now truly orthogonal projectors on $\mathscr{H}_{\mathcal{S}}$ for all $u\in
L^2(\mathbb{R}^3)$. Along the trajectories $u_{\pm}(t)$, solution of \eqref{eq:1}, we
thus define the projectors $\pi_{\pm}^t$ as follows:
%\begin{gather*}
 \[
 \pi_{\pm}^0:= \pi_{\pm}(\underline{u})\;\;\mbox{and}\;\;
  \pi_{\pm}^t:= \pi_{\pm}\bigl(u_{\pm}(t)\bigr)\;,\; \forall t>0\;.
\]
We need the assumption that $g$ and $\underline{u}$ are orthogonal (in that
way, the $\mathcal{S}$-wavefunction is collapsed in the right subspaces), yielding
\begin{equation*}
  \pi_{\pm}^0= \lvert\pm \rangle\langle\pm\rvert\;.
\end{equation*}
Such an assumption is not too restrictive or unphysical, since we have
freedom in building our probe $\mathcal{P}_{\varepsilon}$ (the measurement
scheme must be engineered suitably). Under the evolution $e^{-i
  \frac{t}{\varepsilon}\boldsymbol{H}_{\varepsilon}(\mathfrak{g})}$, the vector
$\Psi_{\varepsilon}$ stays coherent in the field's degrees of freedom but it
changes along a specific direction: it is in fact described by the vector
\begin{equation*}
  e^{\frac{i}{\varepsilon}S_{\pm}(t)}\lvert \widetilde{u}_{\pm}(t)_{\varepsilon}\rangle\;,
\end{equation*}
where the \emph{action} $S_{\pm}\in C(\mathbb{R}_+,\mathbb{R})$ is
defined as
\begin{equation}
  \label{eq:2}
  S_{\pm}(t)=  \mp \int_0^t \frac{\mathfrak g^2 (\mathrm{Re}\langle u_\pm(s),g\rangle)^2+1}{\sqrt{1+2\mathfrak g^2 (\mathrm{Re}\langle u_\pm(s),g\rangle)^2}}ds\;,
\end{equation}
and $\lvert \widetilde{u}_{\pm}(t)_{\varepsilon}\rangle$ is a \emph{squeezed
  coherent state}, defined by
\begin{equation}
  \label{eq:3}
  \lvert \widetilde{u}_{\pm}(t)_{\varepsilon}\rangle= e^{\frac{1}{\varepsilon}(a_{\varepsilon}^{*}(u_{\pm}(t))-a_{\varepsilon}(u_{\pm}(t)))}U_{\pm}(t,0)_{\varepsilon}\,\Omega\;,
\end{equation}
where $U_{\pm}(t,0)_{\varepsilon}$ is the two-parameter unitary group on
$\mathscr{H}_{\mathcal{P}_{\varepsilon}}$ generated by the quadratic
time-dependent Hamiltonian
\begin{equation*}
  A_{\pm}(t)_{\varepsilon}= \mathrm{d}\Gamma_{\varepsilon}(-\Delta)\pm \frac{\mathfrak{g}^2}{2}\frac{\phi_{\varepsilon}(g)^2}{(1+2\mathfrak{g}^2\, \mathrm{Re}\, \langle u_{\pm}(t)  , g \rangle_2^2)^{3/2}}\;.
\end{equation*}
In $\mathcal{S}$, the evolution induces the unitary operator
$\mathcal{R}_{\pm}(t)$, which satisfies the intertwining relation
\begin{equation*}
  \mathcal{R}_{\pm}(t)\pi^0_{\pm}= \pi^t_{\pm}\mathcal{R}_{\pm}(t)\;.
\end{equation*}
The remarkable stability of the structure of coherent states leads to the
following result.

\begin{proposition}
  \label{prop:3}
  For all $g\in \mathscr{S}(\mathbb{R}^3)$ and all $\underline{u}\in L^2(\mathbb{R}^3)$ such that $
  \mathrm{Re}\,\langle \underline{u} , g \rangle_2=0$, there exist $\mathfrak{g}_0>0$ and $C_1,C_2>0$
  such that for all $\mathscr{H}_{\mathcal{S}}\ni \psi=\alpha_+\lvert + \rangle +\alpha_-\lvert -\rangle$, $t>0$ and $0<\mathfrak{g}<\mathfrak{g}_0$,
  \begin{equation*}
    \Bigl\lVert \boldsymbol{\Psi}_{\varepsilon}(t)\,-\mspace{-10mu} \sum_{j\in \{+,-\}}^{}\mspace{-10mu}\alpha_j\,\pi_j^t\,\mathcal{R}_j(t)\,\lvert j \rangle\:\otimes\: e^{\frac{i}{\varepsilon}S_{j}(t)}\,\lvert \widetilde{u}_{j}(t)_{\varepsilon}\rangle \Bigr\rVert_{\mathscr{H}_{\mathcal{S}}\otimes \mathscr{H}_{\mathcal{P}_{\varepsilon}}}^{}\leq C_1 \sqrt{\varepsilon}\, e^{C_2t}\;.
  \end{equation*}
\end{proposition}
\begin{remark*}
  The assumption that there exists $t>0$ such that $\mathrm{Re}\,\langle
  e^{it\Delta}\underline{u} , g \rangle_2\neq0$ is only needed to prove separation of the
  asymptotic states: $u_+^{\infty}\neq u_-^{\infty}$, and it is thus not needed for the
  proposition above.
\end{remark*}
The above norm bound is the key ingredient, combined with a quantitative
study of the long time asymptotics of $u_{\pm}(t)$, to prove that $\mathscr{M}_{\varepsilon}^+(\sigma_z)$
is a strong semiclassical scheme that solves the measurement problem within
error limits (and to have a bound on the rate of convergence), leading to the
following result.
% ; it is thus
% clear that, for being able to take the limit $t\to \infty$ ``after'' the limit $\varepsilon\to
% 0$, we must choose an asymptotic time $t_{\varepsilon}$ that diverges at most
% logarithmically in $\varepsilon$
\begin{proposition}
  \label{prop:4bis}
  For all $\varrho\in \mathcal{L}^1(\mathscr{H}_{\mathcal{S}})$, $\varrho\geq 0$ and $\mathrm{tr}_{\mathcal{S}}\,\varrho=1$, we have that
  \begin{equation*}
    \boldsymbol{U}_{\varepsilon}\Bigl(\varrho\otimes \lvert \underline{u}_{\varepsilon}\rangle\langle \underline{u}_{\varepsilon}\rvert\Bigr)\boldsymbol{U}^{*}_{\varepsilon} \:\overset{\mathcal{F}}{=}\, \sum_{j\in \{+,-\}}^{}\langle j  , \varrho\, j \rangle_{\mathscr{H}_{\mathcal{S}}}\, \lvert j\rangle\langle j\rvert\, \delta_{u^{\infty}_j} + \boldsymbol{O}(\lvert\ln\varepsilon\rvert^{-1/2})\;,
  \end{equation*}
  where $\boldsymbol{O}(\lvert\ln\varepsilon\rvert^{-1/2})$ stands for the fact that for all $\xi\in
  \mathscr{S}(\mathbb{R}^3)$ there exists $C(\xi)>0$ (that depends only on $\lVert \xi
  \rVert_{H^1}^{}$) such that for all $K\in \mathscr{B}(\mathbb{C}^2)$,
  \begin{multline*}
    \mathrm{tr}_{\mathcal{S}+\mathcal{P}_{\varepsilon}}\Bigl(\bigl(K\otimes \pi_{\varepsilon}(\xi,0)\bigr)\boldsymbol{U}_{\varepsilon}\Bigl(\varrho\otimes \lvert \underline{u}_{\varepsilon}\rangle\langle \underline{u}_{\varepsilon}\rvert\Bigr)\boldsymbol{U}^{*}_{\varepsilon}\Bigr) \\= \sum_{j\in \{+,-\}}^{} \langle j  , \varrho\, j \rangle_{\mathscr{H}_{\mathcal{S}}}\, \langle j  , K\, j \rangle_{\mathscr{H}_{\mathcal{S}}}\, e^{i \sqrt{2} \mathrm{Re}\langle \xi  , u^{\infty}_j \rangle_{}}+ \frac{C(\xi)\,\lVert K  \rVert_{}^{}}{\sqrt{\lvert \ln \varepsilon\rvert}}\;.
  \end{multline*}
\end{proposition}
Let us finally remark that once Proposition~\ref{prop:4bis} is proved, from
our choice of $\mathrm{Z}_{\varepsilon}$ the convergence
\eqref{eq:29bis} follows immediately by Proposition~\ref{prop:1}.

\subsection{Physical interpretation of the measurement scheme}
\label{sec:phys-interpr-mathscr}

From a physical perspective, our concrete measurement scheme
$\mathscr{M}_{\varepsilon}^+(\sigma_z)$ corresponds to an experimental
setup that is, at least in theory, realizable; let us briefly describe it
here, as a conclusion to \textsection\ref{sec:meas-scheme-thro}.
\smallskip

%\begin{itemize}
 % \setlength{\itemsep}{3mm}
%\item 
The isolated system $\mathcal{S}+\mathcal{P}_{\varepsilon}$ consists of
  a two-level quantum sample (\emph{e.g.} an atom) in adiabatic interaction
  with a force-carrying bosonic field (\emph{e.g.} light) acting as the
  probe. The non-relativistic dispersion $k^2$ for the probe is used purely
  for mathematical convenience, a more realistic dispersion could be taken;
  the only point in which one should take some care is in modifying the
  scattering theory of $u_{\pm}(t)$ accordingly. Let us also remark that the
  adiabatic approximation is very well suited to describe realistic systems:
  it was originally introduced by Born and Oppenheimer \citep{born1927adp},
  and it is of widespread use in quantum chemistry to study molecular
  dynamics.
  \smallskip 
%\item

Before measurement takes place, the probe is in a lasing state $\lvert
  \underline{u}_{\varepsilon}\rangle\langle \underline{u}_{\varepsilon}\rvert$, that is prepared with a very small
  $\varepsilon>0$: such lasing states well describe classical (macroscopic) light
  pulses and can be realized experimentally with a very high level of
  precision. We further assume that our light pulse is highly
  focused in space, and that the sample $\mathcal{S}$ is initially not in its support.
\smallskip

%\item 
The measurement is carried out as the von Neumann scheme $\mathscr{M}_{\varepsilon}^+(\sigma_z)$ by
putting the sample in the region of space where the pulse is supported, and
letting them interact for a sufficiently long time: this is achieved by
placing the light detectors far enough from the pulse's support.
\smallskip

%\item
Due to the absorption and subsequent emission of light (fluorescence) in $\mathcal{S}$,
the pulse is split into two classical traveling pulses, whose trajectories
are described (up to a $\mathrm{O}(\sqrt{\varepsilon})$ error) by the distinct
classical configurations $u_{\pm}(t)$.
\smallskip

%\item
These classical pulse configurations are detected with frequencies
  $p_{\pm}(t)\sim_{\mathrm{O}(t^{-1})} \langle \pm, \varrho\pm\rangle_{\mathcal{S}}$. Moreover, their detection
  forces a collapse in $\mathcal{S}$ onto subspaces $\pi^t_{\pm}\sim_{\mathrm{O}(t^{-1})}
  \lvert\pm\rangle\langle\pm\rvert$.
  \smallskip  

%\item
The distance at which the light detectors are placed from the
  original pulse's support is tuned to make the detection time $t_{\varepsilon}=
  \mathrm{O}(-\ln \varepsilon)$. This way, one obtains the best approximation
  of the ideal Bohr scheme
  $\mathscr{M}_0(\sigma_z)$: we have a realistic quantum measurement scheme
  that does not need to introduce new axioms or interpretations to quantum
  mechanics but still overcomes the measurement problem, at least within an
  accuracy of order $1/\lvert\ln \varepsilon\rvert^{\frac{1}{2}}$, all the while respecting
  the no-go Theorem~\ref{thm:7}.
%\end{itemize}

\section{Adiabatic propagation of coherent states}
\label{sectioncoherentstates}

In this section, we focus on studying the finite time evolution of a coherent
wave packet through the adiabatic evolution $e^{-i
  \frac{t}{\varepsilon}\boldsymbol{H}_{\varepsilon}(\mathfrak{g})}$, and prove
Proposition~\ref{prop:3}. Throughout this section, the existence and
uniqueness of solutions $u_{\pm}(t)$ to the Cauchy problem \eqref{eq:1} will
be taken for granted: they will be proved in
\textsection\ref{sec:longtime}. 

\subsection{Superadiabatic projectors}
\label{sec:classical}

Let us consider $\boldsymbol{\Psi}_{\varepsilon}\in \mathscr{H}:=\mathbb{C}^2\otimes
\mathscr{F}(L^2)$, and define its microscopic evolution as
\begin{equation*}
  \boldsymbol{\Psi}_{\varepsilon}(t):=\boldsymbol{U}_{\varepsilon}(t)\boldsymbol{\Psi}_{\varepsilon}= e^{-i \frac{t}{\varepsilon}\boldsymbol{H}_{\varepsilon}(\mathfrak{g})}\,\boldsymbol{\Psi}_{\varepsilon}\;,
\end{equation*}
that is the unique $\mathscr{H}$-solution of the Schrödinger equation
\begin{equation*}
  i\varepsilon\partial_t\boldsymbol{\Psi}_{\varepsilon}(t)= \boldsymbol{H}_{\varepsilon}(\mathfrak{g})\,\boldsymbol{\Psi}_{\varepsilon}(t)\;
\end{equation*}
with initial datum $\boldsymbol{\Psi}_\eps$.
% If $\Psi_\eps(t)$ is a solution of Schrödinger's equation $i \eps
% \partial_t \Psi_\eps(t)= H_\eps \Psi_\eps(t)$,
From now on, let us omit the explicit dependence on $\mathfrak{g}$ of
$\boldsymbol{H}_{\varepsilon}(\mathfrak{g})$, for ease of notation.

Our first aim is to understand how this state behaves along the
eigenprojectors, that is, if $\pi_\pm^\eps \boldsymbol{\Psi}_\eps (t)$ solves a possibly
different Schrödinger equation.  As discussed in the introduction, since the
operator $\lambda^\eps_+-\lambda^\eps_-\geq 1$, we can use an adiabatic strategy
aiming at separating the modes.
% it is well-studied in the $L^2(\R^d)$-setting
% (see~\cite{Martinez_Sordoni,Spohn_Teufel} and the
% monograph~\cite{Teufel_book}).
We follow the ideas developed in Chapter~14 of~\cite{Combescure_Robert_book}
for coherent states. The formal computation
$$
i \partial_t (\pi_\pm^\eps \boldsymbol{\Psi}_\eps(t)) = \pi_\pm^\eps \boldsymbol{H}_\eps \boldsymbol{\Psi}_\eps(t) =
\left[\pi_\pm^\eps, \boldsymbol{H}_\eps \right] \boldsymbol{\Psi}_\eps (t) + \lambda_\pm^\eps
\boldsymbol{\Psi}_\eps(t)
$$
leads us to study with precision the commutator $ \left[\pi_\pm^\eps, \boldsymbol{H}_\eps
\right]$.
\begin{lem} \label{commutatorpiH} In $D(\mathrm{d}\Gamma_\eps(1))$, the
  commutator
  \[
    B_\eps := \left[\pi_+^\eps, \boldsymbol{H}_\eps \right] =- \left[\pi_-^\eps, \boldsymbol{H}_\eps
    \right]
  \]
  is a skew-symmetric off-diagonal operator, in the sense that $\pi_+^\eps
  B_\eps \pi_+^\eps=0$ and $\pi_-^\eps B_\eps \pi_-^\eps = 0$.
  Furthermore, %\footnote{C. I have added the $\mathfrak g$, and added comments in the proof. }
  \begin{multline*}
    B_\eps = \frac{\eps\mathfrak g}{2 i} \phi_\eps(i \Delta g) \left(  \frac{\sigma_x}{\sqrt{1+\mathfrak g^2\phi_\eps(g)^2}} -(\sigma_z +\mathfrak g \sigma_x \phi_\eps(g)) \frac{\mathfrak g\phi_\eps(g)}{(1+\mathfrak g^2\phi_\eps(g)^2)^{3/2}}  \right) + \boldsymbol{O}(\eps^2)
    \\
    = \eps \Lambda_{+-}^\eps +\eps \Lambda_{-+}^\eps + \boldsymbol{O}(\eps^2)\;;
  \end{multline*}
  % \footnote{Shouldn't it be
  % \[
  %   B_\eps = \frac{\eps\mathfrak g}{2 i} I_2\otimes\frac{\phi_\eps(i \Delta
  %   g)}{\sqrt{1+\mathfrak g^2\phi_\eps(g)^2}} \left( \sigma_x\otimes\Id
  %     -(\sigma_z\otimes\Id +\mathfrak g \sigma_x \otimes \phi_\eps(g))
  %     \left(I_2\otimes \frac{\mathfrak g\phi_\eps(g)}{1+\mathfrak
  %       g^2\phi_\eps(g)^2} \right) \right) + \mathcal O(\eps^2)
  % \]}
  with the off-diagonal components given by
  \begin{equation}\label{def:Lambda}
    \begin{cases}
      \Lambda_{+-}^\eps := \pi_+^\eps B_\eps \pi_-^\eps = \frac{\mathfrak
        g}{2 i} \phi_\eps(i \Delta g) (1+\mathfrak g^2\phi_\eps(g)^2)^{-1/2}
      \pi_+^\eps \sigma_x \pi_-^\eps,
      \\
      \Lambda_{-+}^\eps := \pi_-^\eps B_\eps \pi_+^\eps = \frac{\mathfrak g}{ 2 i}
      \phi_\eps(i \Delta g) (1+\mathfrak g^2\phi_\eps(g)^2)^{-1/2} \pi_-^\eps
      \sigma_x \pi_+^\eps.
    \end{cases}
  \end{equation}
\end{lem}

\begin{remark}\label{rem:dpi}
  The expression inside the parentheses can be formally rewritten as
    $$
    \frac{\sigma_x}{2 \sqrt{1+\mathfrak{g}^2\phi_\eps(g)^2}} -(\sigma_z + \sigma_x
    \mathfrak{g} \phi_\eps(g)) \frac{2\mathfrak{g} \phi_\eps(g)}{(1+\mathfrak{g}^2\phi_\eps(g)^2)^{3/2}} =
    \frac{\partial}{\partial \phi_\eps(g)} \left( \frac{\sigma_z + \sigma_x
        \mathfrak{g}\phi_\eps(g)}{2 \sqrt{1+\mathfrak{g}^2\phi_\eps(g)^2}} \right) = \frac{\partial
      \pi_+^\eps}{\partial \phi_\eps(g)}
    $$
  \end{remark}

\begin{proof}
The result comes from two commutation relations: we have
 on $D(\mathrm{d}\Gamma_{\varepsilon}(1))$,
  \begin{align}
  \label{commutatordiGammaphi} 
& [\phi_\eps(g), \di
   \Gamma_\eps(-\Delta)] = i \eps \phi_\eps (-i \Delta g) = \frac{\eps}{i}
    \phi_\eps(i \Delta g),\\
    \nonumber
   & \frac{i}{\eps} \left[(1+\mathfrak g^2\phi_\eps(g)^2)^{-1/2}, \di \Gamma_\eps(- \Delta) \right] = \frac{\mathfrak g^2}{2}\Bigl(\phi_\eps(g) \phi_\eps(- i \Delta g) +  \phi_\eps(- i \Delta g) \phi_\eps(g) \Bigr) \\
    \label{lemmacommutatorphiH}
   & \qquad \qquad\qquad\qquad \times \frac{1}{(1+\mathfrak g^2\phi_\eps(g)^2)^{3/2}}
%\\
%\qquad 
 -\frac{3 i \eps \mathfrak g^4}{2}\, \langle g, - \Delta g \rangle_2\:  \frac{\phi_\eps(g)^2}{(1+\mathfrak g^2\phi_\eps(g)^2)^{5/2}}\;.
  \end{align}
  The first one is straightforward and the proof of the  second one is postponed at the end of the proof. As a consequence, 
  on $D(\mathrm{d}\Gamma_\eps(1))$ we have
  \begin{multline*}
    \left[\pi_+^\eps, \boldsymbol{H}_\eps \right]= \frac{1}{2} \left[ \frac{\sigma_z + \mathfrak g \sigma_x\otimes \phi_\eps(g)}{\sqrt{1+\mathfrak g ^2\phi_\eps(g)^2}},\di \Gamma_\eps(-\Delta) + \sigma_z +\mathfrak g  \sigma_x \otimes \phi_\eps(g) \right] \\
    =\frac{1}{2} (\sigma_z + \mathfrak g \sigma_x\otimes \phi_\eps(g)) \left[(1+\mathfrak g ^2\phi_\eps(g)^2)^{-1/2},\di\Gamma_\eps(-\Delta) \right]\\+ \frac{\mathfrak g }{2} \left[ \sigma_x \otimes \phi_\eps(g), \di \Gamma_\eps (-\Delta)  \right] (1+\mathfrak g ^2\phi_\eps(g)^2)^{-1/2}
    \\=\frac{1}{2} (\sigma_z +\mathfrak g  \sigma_x \otimes \phi_\eps(g)) \left[(1+\mathfrak g ^2\phi_\eps(g)^2)^{-1/2},\di \Gamma_\eps (-\Delta)  \right] \\+\frac{\eps\mathfrak g }{2 i}\sigma_x \otimes \phi_\eps(i \Delta g)(1+\mathfrak g ^2\phi_\eps(g)^2)^{-1/2}\\
    =\frac{\eps}{2 i} \phi_\eps(-i \Delta g) (\sigma_z + \mathfrak g \sigma_x \phi_\eps(g)) \frac{\phi_\eps(g)}{(1+\mathfrak g ^2\phi_\eps(g)^2)^{3/2}} \\+\frac{\eps\mathfrak g }{2 i}\sigma_x \otimes \phi_\eps(i \Delta g)(1+\mathfrak g ^2\phi_\eps(g)^2)^{-1/2} + \boldsymbol{O}(\eps^2),
  \end{multline*}
  where the $\boldsymbol{O}(\eps^2)$ is a bounded operator that converges in
  operator norm. Since $(\pi_+^\eps)^2 =\pi_+^\eps$, the commutator
  $\left[\pi_+^\eps, H_\eps \right]$ is off-diagonal. We investigate these
  off-diagonal components for the first term on the second line above. By the
  canonical commutation relations, on $D(\mathrm{d}\Gamma_\eps(1))$ it holds that
  \begin{multline*}
    \pi_\pm^\eps \phi_\eps(i \Delta g) (\sigma_z +\mathfrak g  \sigma_x \phi_\eps(g)) \frac{\phi_\eps(g)}{(1+\mathfrak g ^2\phi_\eps(g)^2)^{3/2}}  \pi_\mp^\eps \\= \phi_\eps(i \Delta g) \frac{\phi_\eps(g)}{1+\phi_\eps(g)^2} \pi_\pm^\eps \frac{\sigma_z+\mathfrak g \sigma_x \phi_\eps(g)}{\sqrt{1+\mathfrak g ^2\phi_\eps(g)^2}} \pi_\mp^\eps +\boldsymbol{O}(\eps^2).
  \end{multline*}
  We denote $A_\eps = (\sigma_z+\mathfrak g \sigma_x
  \phi_\eps(g))(1+\mathfrak g ^2\phi_\eps(g)^2)^{-1/2} \in \mathscr
  B(\mathscr H)$. Since $A_\eps^2 = 1$, a straightforward computation shows
  that,
  \begin{equation*}
    \pi_\pm^\eps A_\eps \pi_\mp^\eps = \frac{1}{4} \left( I_2 \pm A_\eps \right) A_\eps  \left( I_2 \mp A_\eps \right) = \frac{A_\eps}{4} (I_2-A_\eps^2)=0.
  \end{equation*}
  Therefore, the only relevant terms in the commutator are %in a sense
  coming from $ \pi_+^\eps \sigma _x \pi_-^\eps $ and $\pi_-^\eps \sigma_x
  \pi_+^\eps$. We write in $D(\mathrm{d}\Gamma_\eps(1))$,
  \begin{align*}
    \pi_\pm^\eps \sigma_x \otimes \phi_\eps(i \Delta g) (1+\mathfrak g ^2\phi_\eps(g)^2)^{-1/2} \pi_\mp^\eps = \phi_\eps(i \Delta g) (1+\mathfrak g ^2\phi_\eps(g)^2)^{-1/2} \pi_\pm^\eps \sigma_x \pi_\mp^\eps+ \boldsymbol{O}(\eps^2).
  \end{align*}
  The exact expression of $\pi_\pm^\eps \sigma_x \pi_\mp^\eps$ can be
  computed, but we will not need it in our analysis, as we only need to
  control its classical limit $\varepsilon\to 0$.
  \smallskip 

  It remains to prove~\eqref{lemmacommutatorphiH}.  We use the following integral formula:
    $$
    (1+\phi^2)^{-1/2}= \frac{1}{\sqrt{\pi}} \int_\R e^{-(1+\phi^2)x^2} \di x.
    $$
    Let us  take $\Psi \in \mathscr F_\mathrm{fin}(L^2) \cap D(\di
    \Gamma_\eps(- \Delta))$, and denote $\Phi := \di \Gamma_\eps(- \Delta)
    \Psi$, so that
    \begin{equation*}
      \langle \Psi, [(1+\mathfrak g^2\phi_\eps(g)^2)^{-1/2}, \di \Gamma_\eps(- \Delta)] \Psi \rangle = 2 i \mathrm{Im} \langle \Psi, (1+\mathfrak g^2\phi_\eps(g)^2)^{-1/2} \Phi \rangle.
    \end{equation*}
    Let $\mu_{\Psi,\Phi}$ be the spectral measure of the observable $\mathfrak
    g\phi_\eps(g)$ associated with the states $\Psi$ and $\Phi$, and denote by $\phi$ the
    variable of the spectral measure. We obtain
    \begin{align*}
      \langle \Psi, [(1+\mathfrak g^2\phi_\eps(g)^2)^{-1/2}, \di \Gamma_\eps(- \Delta)] \Psi \rangle  
      &= 2 i \mathrm{Im} \int_\R (1+\phi^2)^{-1/2} \di \mu_{\Psi,\Phi}(\phi) \\
      &= 2 i \mathrm{Im} \iint_{\R^2} e^{-(1+\phi^2 )x^2} \frac{\di x}{\sqrt{\pi}} \di \mu_{\Psi,\Phi}(\phi)  \\
      &= \int_\R  2 i \mathrm{Im} \langle \Psi, e^{-(1+\mathfrak g^2\phi_\eps(g)^2 )x^2}  \Phi \rangle \frac{\di x}{\sqrt{\pi}} \\
      &= \int_\R \langle \Psi, [e^{-(1+\mathfrak g^2\phi_\eps(g)^2 )x^2} , \di \Gamma_\eps(-\Delta)] \Psi \rangle  \frac{\di x}{\sqrt{\pi}},
    \end{align*}
    where the exchange of integrals is justified since $e^{-(1+\phi)^2 x^2}
    \leq e^{-x^2}$ is $\di x \otimes \abs{ \mu_{\Psi,\Phi}
    }$-integrable. Indeed, $\mu_{\Psi,\Phi}$ is a finite measure because of
    polarization's identity and $\norm{\Psi \pm \Phi}^2 <\infty$. Now,
    differentiating twice the last commutator with respect to $x^2$, one
    checks that in the sense of quadratic forms,
    \begin{multline*}
      [e^{-(1+\mathfrak g^2\phi_\eps(g)^2 )x^2} , \di \Gamma_\eps(- \Delta)] =
      \left( e^{-\mathfrak g^2\phi_\eps(g)^2 x^2} \di \Gamma_\eps(- \Delta) e^{\mathfrak g^2\phi_\eps(g)^2 x^2} - \di \Gamma_\eps (-\Delta) \right) e^{-(1+\mathfrak g^2\phi_\eps(g)^2 )x^2}
      \\= \Bigl( - i \eps  \mathfrak g^2x^2  (\phi_\eps(g) \phi_\eps(- i \Delta g) +  \phi_\eps(- i \Delta g) \phi_\eps(g)) -2 x^4 \eps^2  \mathfrak g^4\langle g, - \Delta g \rangle_2 \phi_\eps(g)^2 \Bigr)
      % \\
      % &\qquad \qquad \qquad \times
      e^{-(1+\mathfrak g^2\phi_\eps(g)^2 )x^2}. 
    \end{multline*}
    Replacing this expression in the previous integral, one obtains, for the
    last term in the parentheses above, by applying Fubini's theorem again,
    \begin{equation*}
      \int_\R \langle \Psi, -2 x^4 \eps^2 \mathfrak g^4 \langle g, - \Delta g \rangle_2 \phi_\eps(g)^2 e^{-(1+\mathfrak g^2\phi_\eps(g)^2 )x^2} \Psi \rangle  \frac{\di x}{\sqrt{\pi}} 
      =- \frac{3}{2} \eps^2 \mathfrak g^4 \langle g, - \Delta g \rangle_2 \frac{\phi_\eps(g)^2}{(1+\mathfrak g^2\phi_\eps(g)^2)^{5/2}}\;.
    \end{equation*}
    For the first term, one needs to be a bit more cautious. We transfer
    $\phi_\eps(g) \phi_\eps(- i \Delta g) + \phi_\eps(- i \Delta g)
    \phi_\eps(g)$ to the left, and then apply Fubini's theorem with the
    spectral measure of the observable $\phi_\eps(g)$ associated with the
    states $(\phi_\eps(g) \phi_\eps(- i \Delta g) + \phi_\eps(- i \Delta g)
    \phi_\eps(g))\Psi$ and $\Phi$, to obtain the wanted term
    \begin{align*}
      \int_\R &\langle \Psi, \left( -  i \eps \mathfrak g^2 x^2  (\phi_\eps(g) \phi_\eps(- i \Delta g) +  \phi_\eps(- i \Delta g) \phi_\eps(g)   \right)e^{-(1+\mathfrak g^2\phi_\eps(g)^2 )x^2} \Psi \rangle  \frac{\di x}{\sqrt{\pi}} \\
      &=- \frac{i \eps \mathfrak g^2}{2}  \bigl(\phi_\eps(g) \phi_\eps(- i \Delta g) +  \phi_\eps(- i \Delta g) \phi_\eps(g)\bigr) \frac{1}{(1+\mathfrak g^2\phi_\eps(g)^2)^{3/2}}\;.
    \end{align*}
\end{proof}

From the above it is clear that we do not get exactly a Schrödinger equation
for $\pi_{\pm}^\eps \Psi_\eps(t)$, since it is impossible to insert
$\pi_{\pm}^\eps$ to the right of $\Lambda_{\pm\mp}^\eps$. To address this
issue, we slightly modify the eigenprojectors in order to absorb this
contribution. Following~\cite{Teufel_book,Combescure_Robert_book}, this
modified projector is called a \emph{superadiabatic projector}.

% \color{black}
\begin{lem}[Superadiabatic projectors]
  \label{lem:superadiproj}
  Let
$$
\pi^{1,\eps} := \left( \Lambda_{+-}^\eps - \Lambda_{-+}^\eps
\right)(\lambda_+^\eps - \lambda_-^\eps)^{-1},
$$
Then, the superadiabatic projectors $\pi_+^\eps + \eps \pi^{1,\eps}$ satisfy
in $D(\di \Gamma_\eps(-\Delta))$ the relation,
$$
(\pi_+^\eps + \eps \pi^{1,\eps}) \boldsymbol{H}_\eps = (\lambda_+^\eps + \eps
(\Lambda_{-+}^\eps- \Lambda_{+-}^\eps)) (\pi_+^\eps + \eps \pi^{1,\eps}) +
\boldsymbol{O}(\eps^2).
$$
\end{lem}
\begin{remark*}\label{rem:diag}
  Observe that, also,
    $$
    (\pi_-^\eps - \eps \pi^{1,\eps}) \boldsymbol{H}_\eps = (\lambda_-^\eps - \eps
    (\Lambda_{-+}^\eps- \Lambda_{+-}^\eps)) (\pi_-^\eps - \eps \pi^{1,\eps})
    + \boldsymbol{O}(\eps^2).
    $$
    Furthermore, at this stage of the proof, we see that the vectors
   $\boldsymbol{\Psi}_{\eps,\pm}(t):=(\pi^\eps_\pm\pm\eps\pi^{1,\eps})\boldsymbol{\Psi}_\eps(t)$ verify
   \[
     i\eps\partial_t \boldsymbol{\Psi}_{\eps,\pm}(t)= (\lambda_\pm^\eps\pm \eps
     (\Lambda_{+-}^\eps- \Lambda_{-+}^\eps))\boldsymbol{\Psi}_{\eps,\pm}(t) +\boldsymbol{O}(\eps^2)\;\;\mbox{in}\;\; \Hi.
   \]
\end{remark*}
\begin{proof}
  % \footnote{The rigorous proof is given in the
  % Appendix~\ref{appendixcommutatorestimates}.  We give here only a formal
  % proof with the main ideas, assuming that each commutator is of size
  % $\mathrm O(\eps)$ (which is what happens in practice).}
  We write
  \begin{equation*}
    (\pi_+^\eps + \eps \pi^{1,\eps}) \boldsymbol{H}_\eps  = \boldsymbol{H}_\eps \pi_+^\eps +[\pi_+^\eps, \boldsymbol{H}_\eps ]+ \eps \pi^{1,\eps} \boldsymbol{H}_\eps= 
    \eps \Lambda_{-+}^\eps + \eps \Lambda_{+-}^\eps + \lambda_{+}^\eps \pi_+^\eps + \eps \pi^{1,\eps} \boldsymbol{H}_\eps + \boldsymbol{O}(\eps^2).
  \end{equation*}
  where we have used (see~\eqref{def:Lambda}) $[\pi^\eps_+,\boldsymbol{H}_\eps]=\eps
  \Lambda_{+-}^\eps+\eps \Lambda_{-+}^\eps + \boldsymbol{O}(\eps^2)$.
  We obtain
  \begin{multline*}
    (\pi_+^\eps + \eps \pi^{1,\eps}) \boldsymbol{H}_\eps= (\lambda_+^\eps + \eps \Lambda_{-+}^\eps) \pi_+^\eps + \eps (\pi^{1,\eps} \boldsymbol{H}_\eps + \Lambda_{+-}^\eps) + \boldsymbol{O}(\eps^2) \\
    = (\lambda_+^\eps + \eps \Lambda_{-+}^\eps-\eps \Lambda_{+-}^\eps) (\pi_+^\eps + \eps \pi^{1,\eps}) \\+ \eps (\pi^{1,\eps} \boldsymbol{H}_\eps  -\lambda_+^\eps \pi^{1,\eps} + \Lambda_{+-}^\eps) -\eps^2 (\Lambda_{-+}^\eps -\Lambda_{+-}^\eps )\pi^{1,\eps}  +  \boldsymbol{O}(\eps^2)
  \end{multline*}
  where we have used $\Lambda_{+-}^\eps \pi^\eps_+=0$ and $\Lambda_{-+}^\eps
  \pi^\eps_-=0$.  We conclude by observing that 
  %Lemma~\ref{pi^1Hcommutator}.
%\end{proof}
%
%
 % \begin{lem} \label{pi^1Hcommutator} We have, 
 on $D(\di
    \Gamma_\eps(-\Delta))$,
    $$\pi^{1,\eps} \boldsymbol{H}_\eps - \lambda_+^\eps \pi^{1,\eps}  = -\Lambda_ {+-}^\eps + \boldsymbol{O}(\eps).$$
%  \end{lem}
 % \begin{proof}
  Indeed, the relation $[\phi_\eps(i\Delta g), \phi_\eps(g)]=\boldsymbol{O}(\eps)$ in
    $\Hi$ and~\eqref{def:Lambda} implies that in $D(\di
    \Gamma_\eps(-\Delta))$
    \[
      \pi^{1,\eps} \boldsymbol{H}_\eps = (\lambda^\eps_+-\lambda^\eps_-)^{-1}
      (\Lambda^\eps_{+-}-\Lambda^\eps_{-+}) \boldsymbol{H}_\eps + \boldsymbol{O}(\eps).
    \]
    We now consider $(\Lambda_{+-}^\eps-\Lambda_{-+}^\eps) \boldsymbol{H}_\eps -
    \lambda_{+}^\eps (\Lambda_{+-}^\eps-\Lambda_{-+}^\eps)$ and it is enough
    to show that it is equal to $(\lambda_-^\eps -\lambda_+^\eps)
    \Lambda_{+-}^\eps$ up to higher order terms in $\eps$.  For this, we use
    $\boldsymbol{H}_\eps=\lambda^\eps_+\pi^\eps_++\lambda^\eps_-\pi^\eps_-$ and we first
    commute $\lambda^\eps_\pm$ with $\pi^\eps_\pm$, and then with
    $\Lambda^\eps_{\pm,\mp}$, writing for example
    \begin{align*}
      \Lambda^\eps_{+-} \lambda^\eps_-\pi^\eps_- = 
      \Lambda^\eps_{+-}\pi^\eps_-\lambda^\eps_-\pi^\eps_-+  &=  \Lambda^\eps_{+-}\lambda_-^\eps \pi^\eps_- + \Lambda^\eps_{+-}[\pi^\eps_- ,\lambda^\eps_-]\pi^\eps_-\\
      &= \lambda^\eps_-\Lambda^\eps_{+-}+ [\Lambda^\eps_{+-},\lambda^\eps_- ]\pi^\eps_- + \Lambda^\eps_{+-}[\pi^\eps_-,\lambda^\eps_-]\pi^\eps_-\;.
    \end{align*}
    Admitting that these commutators are $\boldsymbol{O}(\eps)$, one obtains in
    $D(\di \Gamma_\eps(-\Delta))$
    \begin{multline*}
      (\Lambda_{+-}^\eps-\Lambda_{-+}^\eps)  \boldsymbol{H}_\eps - \lambda_{+}^\eps (\Lambda_{+-}^\eps-\Lambda_{-+}^\eps) \\= \lambda_-^\eps \Lambda_{+-}^\eps - \lambda_+^\eps \Lambda_{-+}^\eps - \lambda_+^\eps \Lambda_{+-}^\eps +\lambda_+^\eps \Lambda_{-+}^\eps +\boldsymbol{O}(\eps)= (\lambda_-^\eps -\lambda_+^\eps) \Lambda_{+-}^\eps + \boldsymbol{O}(\eps)\;,
    \end{multline*}
    and, multiplying by $(\lambda_+^\eps-\lambda_-^\eps)^{-1}$, we have the
    expected result.  It remains to verify that the commutators
    $[\pi^\eps_\pm,\lambda^\eps_+]$ and $[\Lambda_{\pm,\mp},\lambda^\eps_+]$
    are negligible, or, equivalently, that
    $[\pi^\eps_\pm,\mathrm{d}\Gamma_\eps(-\Delta)]$ and
    $[\Lambda_{\pm,\mp},\mathrm{d}\Gamma_\eps(-\Delta)]$ are negligible.
  \end{proof}

We now want to link the Hamiltonian $\lambda_+^\eps + \eps \Lambda_{-+}^\eps$
to a classical evolution. In the spirit of drawing correspondence between
quantum observables and classical ones, and in view of Remark~\ref{rem:dpi},
we define, for $u \in L^2(\R^3)$,
\begin{multline}
  \label{def:Lambda(u)}
  \Lambda_{\pm \mp}(u) :=  \mathfrak g \sqrt{2} \mathrm{Re}\langle u, i \Delta g \rangle  (1+2\mathfrak g^2 (\mathrm{Re}\langle u,g \rangle)^2)^{-1/2} \pi_\pm(u) \sigma_x \pi_\mp(u) \\
  =\mathfrak g \sqrt{2} \mathrm{Re}\langle u, i \Delta g \rangle \pi_\pm(u) \frac{\partial \pi_+(u)}{\partial(\sqrt{2}\mathrm{Re}\langle u,g \rangle)} \pi_\mp(u).
\end{multline}
The next lemma is devoted to the construction of an operator that will play
the role of parallel transport, as explained in Chapter 14
of~\cite{Combescure_Robert_book} (where it is done for matrix-valued
Hamiltonians and coherent states in $L^2(\R^d)$). This operator will encode
the movement generated by the evolution inside the eigenspace of each
mode.
% COMMENT ON PARALLEL TRANSPORT
\begin{lem} \label{propertiesR_j} The equation
  \begin{equation} \label{equationR_j}
    \begin{cases}
      i \partial_t \mathcal R_\pm(t) = (\Lambda_{\pm\mp}(u_\pm(t))-\Lambda_{\mp \pm}(u_\pm(t))) \mathcal R_\pm(t), \\
      \mathcal R_\pm(0) = I_2
    \end{cases}
  \end{equation}
  admits a unique unitary solution $\mathcal R_\pm \in \mathscr C^1(\R,
  M_2(\C))$. Furthermore, the following intertwining relation holds:
  \begin{equation} \label{intertwining} \mathcal R_\pm(t) \pi_\pm(u_0) =
    \pi_\pm(u_\pm(t)) \mathcal R_\pm(t).
  \end{equation}
\end{lem}

\begin{proof}
  We consider only the $+$ case for clarity, as the $-$ case is handled
  similarly. Global well-posedness is ensured by the Cauchy-Lipschitz theorem
  and the continuity of the map $t\in \R \mapsto u_+(t) \in L^2$. The
  unitarity of $\mathcal R_\pm(t)$ follows from observing that $\mathcal
  R_\pm(0)^* \mathcal R_\pm(0) = I_2$, and that the time derivative of $t
  \mapsto \mathcal R_\pm(t)^* \mathcal R_\pm(t)$ is identically zero. For the
  intertwining relation, one again checks that the relation holds at $t=0$,
  and checks that the derivatives are equal:
  \begin{align*}
    i \partial_t ( \mathcal R_+(t)^* \pi_+(u_+(t)) \mathcal R_+(t)) &= \mathcal R_+(t)^* ( \underbrace{\Lambda_{-+}(u_+(t)) \pi_+(u_+(t) )+ \pi_+(u_+(t)) \Lambda_{+-}(u_+(t))}_{= \Lambda_{-+}(u_+(t))+\Lambda_{+-}(u_+(t))= :B_+(u_+(t)) } ) \mathcal R_+(t) \\
    &+ \mathcal R_+(t)^* i \partial_t(\pi_+(u_+(t)) ) \mathcal R_+(t),
  \end{align*}
  and by Remark~\ref{rem:dpi}
$$
i \partial_t(\pi_+(u_+(t)) ) = - \frac{1}{i} \sqrt{2} \mathrm{Re} \langle
u_+(t), i \Delta g \rangle \frac{\partial \pi_+(u)}{\partial (\sqrt{2}
  \mathrm{Re} \langle u,g \rangle)} (u_+(t)) = -B_+(u_+(t)),
$$
which concludes the proof.
\end{proof}

\subsection{Adiabatic propagation in the modes}
\label{sec:adiab-prop-modes}

Remark~\ref{rem:diag} shows that we can exchange the full evolution governed
by $\boldsymbol{H}_\eps$, in the eigenspace of $\pi_\pm^\eps$, by the evolution yielded by
the operator $\lambda_\pm^\eps+\eps
(\Lambda^\eps_{-+}-\Lambda^\eps_{-+})$. The role of $\mathcal R_\pm(t)$ will
be to counter the action of the subprincipal term and to reduce ourselves to
% completely we can now exchange the full evolution governed by $H_\eps$ in
% the eigenspace of $\pi_\pm^\eps$, by
the evolution generated by $\lambda_\pm^\eps$ alone.
% , at the cost of the introduction of $\mathcal R_\pm(t)$.
For this purpose, we first need to be able to ``exchange'' the symbols
$\Lambda_{-+}^\eps$ with
$\Lambda_{-+}(u_\pm(t))$. Proposition~\ref{propfieldexpansion} proves this is
possible, as it shows that the state $e^{-i \frac{t}{\varepsilon}
  \lambda_\pm^\eps}e^{\frac{1}{\varepsilon}(a^{*}(\underline{u})-a(\underline{u}))}\Omega$
is effectively concentrated around the classical configuration
$u_\pm(t)$. Beforehand, let us prove a technical result that justifies the
existence of a two-parameter unitary group on the Fock space, that is crucial
to approximate the action of $e^{-i \frac{t}{\varepsilon} \lambda_\pm^\eps}$
on the vacuum vector.

\begin{lem} \label{equationU_pm} Consider the time dependent
  operator
    $$
    A_\pm(t)_{\varepsilon} = \di \Gamma_\eps(- \Delta) \pm \frac{\mathfrak
      g^2}{2} \frac{ \phi_\eps(g)^2}{(1+2\mathfrak g^2 \mathrm{Re}\langle
      u_\pm(t),g\rangle^2_2)^{3/2}}\;,
    $$
    seen as a bounded operator from $\mathscr D_+ := D((\di
    \Gamma_\varepsilon(-\Delta)+1)^{1/2})$ to $\mathscr D_- := D((\di
    \Gamma_\varepsilon(-\Delta)+1)^{-1/2})$.\footnote{More precisely, the sets
      $\mathscr D_\pm$ are the completions of $D((\di
      \Gamma_\varepsilon(-\Delta)+1)^{\pm1/2})$ with respect to the inner
      product
    $$
    \langle \Psi,\Phi\rangle_{\mathscr D_\pm} := \langle (\di
    \Gamma_\varepsilon(-\Delta)+1)^{\pm 1/2} \Psi, (\di
    \Gamma_\varepsilon(-\Delta)+1)^{\pm 1/2} \Phi \rangle.
    $$
  }

  Then, the non autonomous Schrödinger Cauchy problem
  \begin{equation}
    \begin{cases}
      i \eps \partial_t \Phi(t) = A_\pm (t)_{\varepsilon} \Phi(t) \\
      \Phi(0) = \Phi_0 \in \mathscr F(L^2)
    \end{cases}
  \end{equation}
  admits a unique solution implemented by the two-parameter unitary group
  $(t,s) \in \R^2 \mapsto U_\pm(t,s)_{\varepsilon}$. In other words, the
  map $(t,s) \in \R^2 \mapsto U_\pm(t,s)_{\varepsilon}$ satisfies:
  \begin{itemize}
  \item[(i)] For every $(t,s,r) \in \R^3$, $U_\pm(t,s)_{\varepsilon}$ is
    unitary, $U_\pm(t,t)_{\varepsilon} = 1_{\mathscr F(L^2)}$, and
    $U_\pm(t,s)_{\varepsilon}U_\pm(s,r)_{\varepsilon} =
    U_\pm(t,r)_{\varepsilon}$.
  \item[(ii)] For all $s \in \R$, the map $ U_\pm(\cdot,s)_{\varepsilon}$
    belongs to $ \mathscr C^0 (\R, \mathscr B(\mathscr D_+)) \cap \mathscr
    C^1 (\R, \mathscr B(\mathscr D_+,\mathscr D_-))$ and for all $\Phi_0 \in
    \mathscr D_+$, and all $t \in \R$,
    \begin{align*}
      i \eps \partial_t (U_\pm(t,s)_{\varepsilon}\Phi_0) &= A_\pm(t)_{\varepsilon} U_\pm(t,s)_{\varepsilon} \Phi_0\;
      %\\
      \;\;\mbox{and}\;\;\;
      i \eps \partial_t (U_\pm(s,t)_{\varepsilon}\Phi_0) = - U_\pm(s,t)_{\varepsilon}A_\pm(s)_{\varepsilon} \Phi_0\;.
    \end{align*}
  \end{itemize}
  Furthermore, the following bound holds:
    $$
    \norm{(\di \Gamma_\varepsilon(-\Delta)+\varepsilon)^{1/2}
      U_\pm(t,s)_\varepsilon \Phi_0} \leq e^{C \abs{t-s}} \norm{(\di
      \Gamma_\varepsilon(-\Delta)+\varepsilon)^{1/2} \Phi_0}.
    $$
  \end{lem}

\begin{proof}
  The proof is an application of \cite[Corollary
  C.4]{ammari2012propagation}. Firstly, observe that for all $t\in \mathbb{R}$,
  $A_{\pm}(t)_{\varepsilon}$ is essentially self-adjoint on $\mathscr{F}_{\mathrm{fin}}(-\Delta)$
  \citep[see, \emph{e.g.}][Nelson's commutator theorem could also be
  used]{falconi2015mpag}; then we use the continuity of $t \in \R \mapsto A_\pm(t)_{\varepsilon}
  \in \mathscr B(\mathscr D_+,\mathscr D_-)$, which follows from the continuity
  of $f :t \in \R \mapsto \left( 1+ 2 \mathfrak g^2 \mathrm{Re}\langle u_\pm(t),g\rangle^2_2
  \right)^{-3/2} \in \R$, by writing
  \begin{equation*}
    \norm{A_\pm(t)_{\varepsilon} - A_\pm(s)_{\varepsilon}}_{\mathscr B(\mathscr D_+, \mathscr D_-)} \leq \abs{f(t)-f(s)} \norm{\phi_\eps(g)^2}_{\mathscr B(\mathscr D_+, \mathscr D_-)} %\\
    \lesssim \eps \abs{f(t)-f(s)} \norm{(-\Delta)^{-1/2} g}^2_{2};
  \end{equation*}
  where the bound $\norm{\phi_\eps(g)^2}_{\mathscr B(\mathscr D_+, \mathscr
    D_-)} \lesssim \eps \norm{(-\Delta)^{-1/2}g}_{2}^2$ is a consequence of
  the standard estimate
    $$
    \norm{ \phi_\eps(g) \Phi} \lesssim \norm{(-\Delta)^{-1/2} g}_{2}
    \norm{(\di \Gamma_\eps(-\Delta)+\eps)^{1/2} \Phi},
    $$
    that can be found, \emph{e.g.}, in \cite[Theorem
    5.16]{arai2018analysis}. It remains to estimate the following commutator,
    using \eqref{commutatordiGammaphi},
    \begin{align*}
      [(\di \Gamma_\varepsilon(-\Delta)+\varepsilon), A_\pm(t)_{\varepsilon}] &= \pm \frac{1}{
        2  (1+ 2 \mathfrak g^2 \mathrm{Re}\langle u_\pm(t),g\rangle^2)^{3/2}} [\di \Gamma_\varepsilon(-\Delta), \phi_\eps(g)^2] \\
      &= \mp \frac{\varepsilon}{2 (1+ 2 \mathfrak g^2 \mathrm{Re}\langle u_\pm(t),g\rangle^2)^{3/2}} (\phi_\varepsilon(g) \phi_\varepsilon( \Delta g) + \phi_\varepsilon( \Delta g) \phi_\varepsilon (g))\;,
    \end{align*}
    thus obtaining the desired bound
    $$
    \norm{[(\di \Gamma_\varepsilon(-\Delta)+\varepsilon),
      A_\pm(t)_{\varepsilon}]}_{\mathscr B(\mathscr D_+,\mathscr D_-)}
    \lesssim \norm{(-\Delta)^{-1/2} g}_{2} \norm{(-\Delta)^{3/2} g}_{2}.
    $$
\end{proof}
We can now describe the so-called propagation of chaos, along each mode, for
the Fock-space coherent wavepackets: the coherent structure is preserved
along the dynamics induced by $\lambda_{\pm}^{\varepsilon}$.

\begin{prop}[Propagation of chaos in the modes]
  \label{propfieldexpansion}
  Let $\underline{u} \in L^2(\R^3)$ and consider the unique global solutions
  $u_\pm(t)$ to the Cauchy problems \eqref{eq:1}. Then there exists $C>0$
  (depending only on some Schwartz seminorms of $g$), such that for every $t
  \in \R$,
    $$
    \norm{ e^{-i \frac{t}{\varepsilon} \lambda_\pm^\eps/_\eps}\lvert
      \underline{u}_{\varepsilon}\rangle - e^{\frac{i}{\varepsilon} S_\pm(t)}
      \lvert \widetilde{u}_{\pm}(t)_{\varepsilon}\rangle} \leq C \mathfrak
    g^3\, \sqrt{\eps}\, {\rm exp}\left(C\,\mathfrak g^2 \abs{t}\right)\;;
    $$
    where the action $S_{\pm}(t)$ and the squeezed coherent vector $\lvert
    \widetilde{u}_{\pm}(t)_{\varepsilon}\rangle$ are given, respectively, by
    \eqref{eq:2} and \eqref{eq:3}.
  \end{prop}
  \begin{remark*}
    For convenience, we adopt the following notation:
    \begin{equation}
      \label{def:squeezeWP}
      \Psi^\eps_{\mathrm{swp},\pm}(t) := e^{\frac i\eps S_\pm(t)} \lvert \widetilde{u}_{\pm}(t)_{\varepsilon}\rangle\;.
    \end{equation}
    Observe that Proposition~\ref{propfieldexpansion} implies
    \begin{equation*}
      \Psi^\eps_{\mathrm{swp},\pm}(t)={\rm e}^{-i\frac t\eps
        \lambda^\eps_\pm} \Psi^\eps_{\mathrm{swp},\pm}(0) + \mathrm O(\sqrt \eps)
      \;\;\mbox{in}\;\; \mathscr{F}(L^2).
    \end{equation*}
  \end{remark*}
  \begin{proof}
    Using the shorthand notation $f_{\mathfrak g}(x) := \sqrt{1+\mathfrak g^2
      x^2}$:
    \begin{multline*}
      (*):= i \eps \frac{\di}{\di t} \left( e^{i \frac{t}{\varepsilon} \lambda_\pm^\eps}e^{\frac i\eps S_\pm(t)} \lvert \widetilde{u}_{\pm}(t)_{\varepsilon}\rangle\right) \\
      =  e^{i \frac{t}{\varepsilon} \lambda_\pm^\eps}  e^{\frac{i}{\varepsilon} S_\pm(t)} \left[ -\di \Gamma_\eps(- \Delta)     \mp f_{\mathfrak g}( \phi_\eps(g)) -\dot{S}_\pm(t) \right] \lvert \widetilde{u}_{\pm}(t)_{\varepsilon}\rangle+ e^{i \frac{t}{\varepsilon} \lambda_\pm^\eps}  e^{\frac{i}{\varepsilon} S_\pm(t)} e^{\frac{1}{\varepsilon}(a_{\varepsilon}^{*}(u_{\pm}(t))-a_{\varepsilon}(u_{\pm}(t)))}\\ \left[\mathrm{Re} \langle u_\pm(t), i\partial_t u_\pm(t) \rangle + \sqrt{2} \phi_\eps(i \partial_t u_\pm(t)) +  A_\pm(t)_{\varepsilon} \right] U_{\pm}(t,0)_{\varepsilon} \Omega.
    \end{multline*}
    Now, let us conjugate with the unitary
    $e^{\frac{1}{\varepsilon}(a^{*}_{\varepsilon}(u_{\pm}(t))-a_{\varepsilon}(u_{\pm}(t)))}$ to obtain
    \begin{gather*}
      \begin{multlined}
        e^{\frac{1}{\varepsilon}(a_{\varepsilon}(u_{\pm}(t))-a_{\varepsilon}^{*}(u_{\pm}(t)))}\, \di
        \Gamma_\eps(- \Delta)\,
        e^{\frac{1}{\varepsilon}(a^{*}_{\varepsilon}(u_{\pm}(t))-a_{\varepsilon}(u_{\pm}(t)))} = \di
        \Gamma_\eps(- \Delta) +\sqrt{2} \phi_\eps(- \Delta u_\pm(t))
        \\+\langle u_\pm(t) , - \Delta u_\pm(t) \rangle
      \end{multlined}\;, \\
      e^{\frac{1}{\varepsilon}(a_{\varepsilon}(u_{\pm}(t))-a^{*}_{\varepsilon}(u_{\pm}(t)))}\:
      f_{\mathfrak g}(\phi_\eps(g))\:
      e^{\frac{1}{\varepsilon}(a^{*}_{\varepsilon}(u_{\pm}(t))-a_{\varepsilon}(u_{\pm}(t)))} =
      f_{\mathfrak g}\left(\phi_\eps(g)+\sqrt{2}\mathrm{Re}\langle
        u_\pm(t),g\rangle \right)\;.
    \end{gather*}
    The second equality holds observing that the left-hand side is a positive
    self-adjoint operator, so that we can take the square root of its
    square. Since its square is a polynomial function in $\phi_\eps(g)$, the
    commutation rule follows.  The derivative $(*)$ thus becomes
    \[(*)=e^{i \frac{t}{\varepsilon}t \lambda_\pm^\eps}
      e^{\frac{i}{\varepsilon}
        S_\pm(t)}e^{\frac{1}{\varepsilon}(a^{*}_{\varepsilon}(u_{\pm}(t))-a_{\varepsilon}(u_{\pm}(t)))}
      \Theta^\eps (t)U_\pm(t,0)_{\varepsilon} \Omega\;, \] with
    \begin{multline*}
      \Theta^\eps (t):=-\dot{S}_\pm(t)+\mathrm{Re}\langle u_\pm(t),(i\partial_t +\Delta) u_\pm(t)\rangle +\sqrt 2 \phi_\eps((i\partial_t+\Delta )u_\pm(t))\\
      \;\;  \mp f_{\mathfrak g}(\sqrt{2}\mathrm{Re}\langle u_\pm(t),g\rangle+\phi_\eps(g))+\frac{\phi_\eps(g)^2}{2} f_{\mathfrak g}^{(2)} \left(\sqrt 2 \mathrm{Re}\langle u_\pm(t),g\rangle \right)
      %\\
      -\di \Gamma_\eps(-\Delta) +A_\pm(t)_{\varepsilon}\;.
    \end{multline*}
    Now, we Taylor expand $f_{\mathfrak g}$ around the real number $\phi_\pm
    :=\sqrt{2}\mathrm{Re}\langle u_\pm(t),g\rangle $, and using functional
    calculus we obtain
    \begin{multline*}
      f_{\mathfrak g}(\phi_\pm+\phi_\eps(g)) = f_{\mathfrak g}(\phi_\pm) +  \phi_\eps(g) f_{\mathfrak g}'(\phi_\pm) +\phi_\eps(g)^2 \frac{f_{\mathfrak g}^{(2)}(\phi_\pm)}{2}
      \\
      + \int_0^{\phi_\eps(g)} \int_0^{\phi_1} \int_0^{\phi_2} f_{\mathfrak g}^{(3)}(\phi_\pm+\phi_3) \di \phi_3 \di \phi_2 \di \phi_1\:,
    \end{multline*}
    with
    \[
      f_{\mathfrak g}(\phi_\pm)=\sqrt{1+2\mathfrak g^2 (\mathrm{Re}\langle
        u_\pm(t),g\rangle)^2}\;\;\mbox{and} \;\; f_{\mathfrak
        g}'(\phi_\pm)=\sqrt 2 \mathfrak g^2 \frac{\mathrm{Re}\langle
        u_\pm(t),g\rangle}{\sqrt{1+2\mathfrak g^2( \mathrm{Re}\langle
          u_\pm(t),g\rangle)^2}}\;.
    \]
    Observe that $u_\pm(t)$ solves \eqref{eq:1}, and by definition of
    $f_{\mathfrak g}$, this yields
    \begin{equation} \label{PDErewritten} \sqrt{2} (i \partial_t +
      \Delta)u_\pm(t) = \pm f_{\mathfrak g}' (\phi_\pm)g\;.
    \end{equation}
    Hence,
    \[
      \mathrm{Re} \langle u_\pm(t), (i \partial_t + \Delta ) u_\pm(t)\rangle
      = \pm \frac 1{\sqrt 2} f'_\mathfrak g (\sqrt{2} \mathrm{Re} \langle
      u_\pm(t), g \rangle ) \mathrm{Re}\langle u_\pm(t), g\rangle.
    \]
    Let us now go back to $\Theta^\eps(t)$, and consider its constant terms;
    by definition of $S_\pm(t)$, we have
    \begin{multline*}
      \mathrm{Re} \langle u_\pm(t),  (i \partial_t + \Delta )u_\pm(t) \rangle \mp f_{\mathfrak g} (\sqrt{2} \mathrm{Re} \langle u_\pm(t), g \rangle ) \\
      =\mp (1+2\mathfrak g^2( \mathrm{Re}\langle u_\pm(t),g\rangle)^2)^{-1/2} (\mathfrak g^2( \mathrm{Re}\langle u_\pm(t),g\rangle)^2 +1)= \dot{S}_\pm(t)\;,
    \end{multline*}
    and thus the constant terms cancel each other out.  The first-order terms
    also vanish, using \eqref{PDErewritten}:
$$
\sqrt{2} \phi_\eps ((i \partial_t + \Delta)u_\pm(t)) \mp \phi_\eps(g)
f_{\mathfrak g}'(\phi_\pm) = 0\;.
$$
Finally, the second-order terms vanish by definition of
$A_\pm(t)_{\varepsilon}$:
$$
\mp \frac{\phi_\eps(g)^2}{2} f_{\mathfrak g}^{(2)} (\phi_\pm) - \di
\Gamma_\eps(-\Delta) + A_\pm (t)_{\varepsilon} = 0\; .
$$
Therefore, it suffices to estimate the norm of the integral remainder
    $$
    \int_0^{\phi_\eps(g)} \int_0^{\phi_1} \int_0^{\phi_2} f_{\mathfrak
      g}^{(3)}(\phi_\pm+\phi_3) \di \phi_3 \di \phi_2 \di \phi_1
    U_\pm(t,0)_{\varepsilon} \Omega.
    $$
    Using functional calculus and by the standard estimate
    $\sup_{x\in\R}\abs{f^{(3)}(x)} \leq 3 \mathfrak g^3$, one obtains
    \begin{align*}
      \norm{\int_0^{\phi_\eps(g)} \int_0^{\phi_1} \int_0^{\phi_2} f^{(3)}(\phi_\pm+\mathfrak g\phi_3) \di \phi_3 \di \phi_2 \di \phi_1 U_\pm(t,0)_{\varepsilon} \Omega} &\leq \mathfrak g^3\norm{\frac{\phi_\eps(g)^3}{2} U_\pm(t,0)_{\varepsilon} \Omega} \\
      &\leq C(\norm{g}_{2}) \mathfrak g^3 \norm{(\mathrm{d}\Gamma_\varepsilon(1)+\varepsilon)^{3/2} U_\pm(t,0)_{\varepsilon} \Omega}
    \end{align*}
    where we used three times the well-known bound
    $$
    \norm{\phi_\eps(g) \Phi} \leq \sqrt{2} \norm{g}_{2}
    \norm{(\mathrm{d}\Gamma_{\varepsilon}(1)+\eps)^{1/2} \Phi}, \Phi \in
    D((\mathrm{d}\Gamma_{\varepsilon}(1)+\eps)^{1/2}).
    $$
    Let us conclude by proving that
    \begin{equation} \label{claim}
      \norm{(\mathrm{d}\Gamma_\varepsilon(1)+\varepsilon)^{3/2}
        U_\pm(t,0)_{\varepsilon} \Omega} \leq C \mathfrak g^3 \exp \left(
        C\mathfrak g^2 \abs{t} \right) \lVert
      (\mathrm{d}\Gamma_\varepsilon(1)+\varepsilon)^{3/2}\Omega \rVert_{}^{}=
      C \mathfrak g^3 \exp \left( C\mathfrak g^2 \abs{t}
      \right)\varepsilon^{3/2}\;;
    \end{equation}
    with $C = C(\norm{g}_{2})$: indeed, by integrating $(*)$ and dividing by
    $i\eps$, we use this bound to obtain
    \begin{equation*}
      \norm{ e^{-\frac{i}{\varepsilon} t \lambda_\pm^\eps} \lvert \underline{u}_{\varepsilon}\rangle - e^{\frac{i}{\varepsilon} S_\pm(t)}\lvert \widetilde{u}_{\pm}(t)_{\varepsilon}\rangle} \leq C \mathfrak g^3\eps^{1/2} e^{C \mathfrak g^2 \abs{t}}.
    \end{equation*}
    To prove \eqref{claim}, let us differentiate to obtain
    \begin{multline*}
      i \eps \frac{\di}{\di t} \norm{(\mathrm{d}\Gamma_\varepsilon(1)+\varepsilon)^{3/2} U_\pm(t,0)_{\varepsilon} \Omega}^2 = i \eps \frac{\di}{\di t} \langle \Omega, U_\pm(0,t)_{\varepsilon} (\mathrm{d}\Gamma_{\varepsilon}(1)+\varepsilon)^3 U_\pm(t,0)_{\varepsilon} \Omega \rangle \\
      = \langle \Omega, U_\pm(0,t)_{\varepsilon} [(\mathrm{d}\Gamma_{\varepsilon}(1)+\varepsilon)^3,A_\pm(t)_{\varepsilon}] U_\pm(t,0)_{\varepsilon}\Omega \rangle\; .
    \end{multline*}
    The commutator is given by
$$
[(\mathrm{d}\Gamma_{\varepsilon}(1)+1)^3,A_\pm(t)_{\varepsilon}] = \pm
\frac{\mathfrak g^2}{2(1+2\mathfrak g^2 \mathrm{Re} \langle
  u_\pm(t),g\rangle^2)^{3/2}}[(\mathrm{d}\Gamma_{\varepsilon}(1)+\varepsilon)^3,\phi_\eps(g)^2]\;.
$$
The pre-factor can be estimated as
$$
\abs{ \frac{\mathfrak g^2}{2(1+2\mathfrak g^2 \mathrm{Re} \langle
    u_\pm(t),g\rangle^2)^{3/2}}} \leq \mathfrak g^2,
$$
while the last commutator is explicit: denoting $$C_\eps =
[\mathrm{d}\Gamma_{\varepsilon}(1),\phi_\eps(g)^2] = -i\varepsilon
(\phi_\eps(g) \phi_\eps(ig) + \phi_\eps(ig) \phi_\eps(g)),$$ we find
$$
[(\mathrm{d}\Gamma_{\varepsilon}(1)+\varepsilon)^3,\phi_\eps(g)^2] =
(\mathrm{d}\Gamma_{\varepsilon}(1)+\varepsilon)^2 C_\eps +
(\mathrm{d}\Gamma_{\varepsilon}(1)+\varepsilon) C_\eps
(\mathrm{d}\Gamma_{\varepsilon}(1)+1) + C_\eps
(\mathrm{d}\Gamma_{\varepsilon}(1)+\varepsilon)^2\;,
$$
which is bounded by $C \eps \norm{g}^2_{L^2}
(\mathrm{d}\Gamma_{\varepsilon}(1)+\varepsilon)^3$. Therefore,
\begin{equation*}
  \frac{\di}{\di t} \norm{(\mathrm{d}\Gamma_{\varepsilon}(1)+\varepsilon)^{3/2} U_\pm(t,0)_{\varepsilon} \Omega}^2 \leq C(\norm{g}_{2}) \mathfrak g^2 \norm{(\mathrm{d}\Gamma_{\varepsilon}(1)+\varepsilon)^{3/2} U_\pm(t,0)_{\varepsilon} \Omega}\;,
\end{equation*}
allowing us to conclude using Gronwall's inequality.
\end{proof}

\subsection{Proof of Proposition~\ref{prop:3}}
\label{sec:proof-proposition}

Let us conclude this section by proving the finite-time semiclassical limit,
in the strong form of the norm bound of Proposition~\ref{prop:3}.  This proposition relies on two lemmas. The first, 
Lemma~\ref{lemmaprojectorweyl}, can be understood as a microlocalization result for the coherent states in the Fock space. It allows us to reduce the entire propagator $e^{-i \frac{t}{\varepsilon} \boldsymbol{H}_\eps} $ 
 to the mode propagator $ e^{- i \frac{t}{\varepsilon} \lambda_\pm^\eps} $, as stated in the second lemma, Lemma~\ref{adiabaticprop}.

\begin{lem} [Microlocalisation of a bosonic coherent state] \label{lemmaprojectorweyl}  For any  $ \varphi_{\varepsilon} \in
    D((\mathrm{d}\Gamma_{\varepsilon}(1)+1)^{1/2})\subset \mathscr{F}(L^2)$ and $u \in L^2$,
$$
\norm{ \Bigl(\phi_\eps(g) - \sqrt{2}\,\mathrm{Re}\,\langle u,g \rangle_2\Bigr) e^{\frac{1}{\varepsilon}(a^{*}_{\varepsilon}(u)-a_{\varepsilon}(u))} \varphi_{\varepsilon}
} \leq \norm{g}_2 \norm{(\mathrm{d}\Gamma_\varepsilon(1)+\varepsilon)^{1/2} \varphi_{\varepsilon}}\; .
$$
With the same assumptions as above, for every $\ket{\psi} \in \C^2$ and with
$\varepsilon$-uniform omitted constants,
$$
\norm{ (\pi_\pm^\eps - \pi_\pm(u)) \ket{\psi} \otimes e^{\frac{1}{\varepsilon}(a^{*}_{\varepsilon}(u)-a_{\varepsilon}(u))} \varphi_{\varepsilon} }
\lesssim  \norm{g}_2 \norm{(\mathrm{d}\Gamma_{\varepsilon}(1)+\varepsilon)^{1/2} \varphi_{\varepsilon}}\;,
$$
and
$$
\norm{ (\Lambda_{-+}^\eps - \Lambda_{-+}(u)) \ket{\psi} \otimes
  e^{\frac{1}{\varepsilon}(a^{*}_{\varepsilon}(u)-a_{\varepsilon}(u))} \varphi_{\varepsilon}} \lesssim \norm{g}_{2} \|\Delta g\|_{2}
\norm{(\mathrm{d}\Gamma_{\varepsilon}(1)+\varepsilon)^{1/2} \varphi_{\varepsilon}}\;.
$$
\end{lem}

\begin{proof}
  We have
  \begin{multline*}
    \norm{ (\phi_\eps(g) - \sqrt{2}\,\mathrm{Re}\,\langle u,g \rangle_2) e^{\frac{1}{\varepsilon}(a^{*}_{\varepsilon}(u)-a_{\varepsilon}(u))} \varphi_{\varepsilon} } \\=\norm{e^{\frac{1}{\varepsilon}(a_{\varepsilon}(u)-a^{*}_{\varepsilon}(u))}  (\phi_\eps(g) - \sqrt{2}\,\mathrm{Re}\,\langle u,g \rangle_2) e^{\frac{1}{\varepsilon}(a^{*}_{\varepsilon}(u)-a_{\varepsilon}(u))} \varphi_{\varepsilon} }\\
    = \norm{\phi_\eps(g) \varphi_{\varepsilon}}  \leq \norm{g}_{2} \norm{(\mathrm{d}\Gamma_\eps(1)+\eps)^{1/2} \varphi_{\varepsilon}}\;.
  \end{multline*}
  For the second bound, we simply use the fact that $\pi_\pm(u)$ is Lipschitz as
  a function of $\sqrt{2}\,\mathrm{Re}\langle u, g\rangle_2$. Therefore, by functional
  calculus one can extract the Lipschitz bound and use the result just above:
  \begin{multline*}
    \norm{ (\pi_\pm^\eps - \pi_\pm(u)) \ket{\psi} \otimes e^{\frac{1}{\varepsilon}(a^{*}_{\varepsilon}(u)-a_{\varepsilon}(u))} \varphi_{\varepsilon} } \\\lesssim \norm{ \ket{\psi} \otimes (\phi_\eps(g) - \sqrt{2}\,\mathrm{Re}\langle u,g \rangle_2) e^{\frac{1}{\varepsilon}(a^{*}_{\varepsilon}(u)-a_{\varepsilon}(u))} \varphi_{\varepsilon} } \lesssim \norm{g}_{2} \norm{(\mathrm{d}\Gamma_{\varepsilon}(1)+\varepsilon)^{1/2} \varphi_{\varepsilon}}\;.
  \end{multline*}
  For the last point, we use~\eqref{def:Lambda} and write
  \begin{multline*}
    \Lambda^\eps_{+-}-\Lambda_{-+}(u)= \frac{\mathfrak g}{2 i}  \phi_\eps(i \Delta g) (1+\mathfrak g^2\phi_\eps(g)^2)^{-1/2} \pi_+^\eps \sigma_x \pi_-^\eps\\
    - \mathfrak g \sqrt{2}\, \mathrm{Re}\,\langle u, i \Delta g \rangle_2  (1+2 \mathfrak  g^2(\mathrm{Re}\,\langle u,g \rangle_2)^2)^{-1/2} \pi_+(u) \sigma_x \pi_-(u)\\
    = \left(\frac{\mathfrak g}{2 i}  \phi_\eps(i \Delta g) -\mathfrak g \sqrt{2}\, \mathrm{Re}\,\langle u, i \Delta g \rangle_2 \right)\Bigl(1+2 \mathfrak  g^2(\mathrm{Re}\,\langle u,g \rangle_2)^2\Bigr)^{-1/2} \pi_+(u) \sigma_x \pi_-(u)\\
     +\mathfrak g \sqrt{2}\, \mathrm{Re}\,\langle u, i \Delta g \rangle_2\Bigl(
      (1+\mathfrak g^2\phi_\eps(g)^2)^{-1/2} \pi_+^\eps \sigma_x \pi_-^\eps
      \\-  (1+2 \mathfrak  g^2(\mathrm{Re}\,\langle u,g \rangle_2)^2)^{-1/2} \pi_+(u) \sigma_x \pi_-(u)
    \Bigr)
  \end{multline*}
  and we conclude as before.
\end{proof}

Our next lemma is an adiabatic result for the propagation of coherent states.

\begin{lemma}  [Adiabatic decoupling for coherent states]\label{adiabaticprop}  Let $\psi \in \C^2$; then there exists $C>0$
  such that for every $t \in \R$,
$$
\norm{ e^{-i \frac{t}{\varepsilon} \boldsymbol{H}_\eps} \psi \otimes \lvert
  \underline{u}_{\varepsilon}\rangle - \sum_\pm \mathcal R_\pm(t) \pi_\pm^0
  \psi \otimes e^{- i \frac{t}{\varepsilon} \lambda_\pm^\eps}  \lvert
\underline{u}_{\varepsilon}\rangle}\leq C \,{\eps}\, e^{C\abs{t}}.
$$
\end{lemma}
\begin{proof}
  Without loss of generality, we may assume that $\pi_+^0 \psi =
  \pi_+(\underline{u})\psi=\psi$. We make use of the superadiabatic
  projectors and consider
  \[
    X_\eps(t)= (\pi_+^\eps + \eps \pi_+^{1,\eps}) e^{-i \frac{t}{\varepsilon}
      \boldsymbol{H}_\eps}\psi \otimes \lvert \underline{u}_{\varepsilon}\rangle -
    \mathcal R_+(t) \pi_+^0 \psi \otimes e^{- i \frac{t}{\varepsilon}
      \lambda_+^\eps} \lvert \underline{u}_{\varepsilon}\rangle\; .
  \]
  Then $X_{\varepsilon}(0)=0$, and, in $\mathscr{H}$,
  \begin{multline*}
    i \eps \partial_t X_{\varepsilon}(t)  = (\pi_+^\eps + \eps \pi_+^{1,\eps}) \boldsymbol{H}_\eps e^{-i \frac{t}{\varepsilon} \boldsymbol{H}_\eps}   \psi \otimes \lvert \underline{u}_{\varepsilon}\rangle\\
    -(\lambda_{+}^\eps+\eps (\Lambda_{-+}(u_+(t)) - \Lambda_{+-}(u_+(t))))
    \mathcal R_+(t) \pi_+^0 \psi \otimes 
    e^{- i \frac{t}{\varepsilon} \lambda_+^\eps} \lvert \underline{u}_{\varepsilon}\rangle \\
    = (\lambda_+^\eps + \eps \Lambda_{-+}^\eps-\eps \Lambda_{+-}^\eps) (\pi_+^\eps + \eps \pi_+^{1,\eps}) e^{-i \frac{t}{\varepsilon} \boldsymbol{H}_\eps}  \psi \otimes \lvert \underline{u}_{\varepsilon}\rangle \\
    -  (\lambda_+^\eps + \eps \Lambda_{-+}^\eps-\eps \Lambda_{+-}^\eps) \mathcal R_+(t) \pi_+^0\psi \otimes e^{- i \frac{t}{\varepsilon} \lambda_+^\eps}\lvert \underline{u}_{\varepsilon}\rangle   + \boldsymbol{O}(\eps^2)\\
    =  (\lambda_+^\eps + \eps \Lambda_{-+}^\eps-\eps \Lambda_{+-}^\eps) X_{\varepsilon}(t)  + \boldsymbol{O}(\eps^2)\;.
  \end{multline*}
  Here we used Lemmas \ref{lem:superadiproj} and
  \ref{lemmaprojectorweyl}. Therefore, by the fundamental theorem of
  calculus, together with Gronwall's lemma, we obtain the expected bound.
\end{proof}
We can now complete the proof of Proposition~\ref{prop:3}.
\begin{proof}[Proof of Proposition~\ref{prop:3}]
  Lemma~\ref{adiabaticprop} yields, on $\mathscr{H}$,
  \[
    e^{-i \frac{t}{\varepsilon} \boldsymbol{H}_\eps} \psi \otimes \lvert \underline{u}_{\varepsilon}\rangle = \sum_\pm \mathcal R_\pm(t)
    \pi_\pm^0\psi \otimes e^{- i \frac{t}{\varepsilon} \lambda_\pm^\eps} \lvert \underline{u}_{\varepsilon}\rangle +\boldsymbol{O}(\eps\, e^{C|t|})\; ,
  \]
  while Proposition~\ref{propfieldexpansion} gives
  \begin{equation*}
    e^{-i \frac{t}{\varepsilon} \boldsymbol{H}_\eps}\psi \otimes \lvert \underline{u}_{\varepsilon}\rangle = \sum_\pm \mathcal R_\pm(t) \pi_\pm^0\psi \otimes e^{\frac{i}{\varepsilon} S_\pm(t)}\lvert \widetilde{u}_{\pm}(t)_{\varepsilon}\rangle+\boldsymbol{O}(\eps \,e^{C|t|})+\boldsymbol{O}(\mathfrak g^3\sqrt\eps\, e^{C\mathfrak g^2 |t|})\; .
  \end{equation*}
  Let us remark that the condition on $\mathfrak{g}$, namely that there results holds
  only for $0<\mathfrak g<\mathfrak{g}_0$, is used only to guarantee the well-posedness
  (and separation) of the classical solutions $u_{\pm}(t)$; if one is able to
  prove existence (and separation) for arbitrary values of $\mathfrak{g}$, then this
  bound holds as well.
\end{proof}

\section{Long time effects and semiclassical accuracy of the measurement
  process}
\label{sec:longtime}

In this section, we first prove the existence of the functions $u_\pm(t)$ that
we have used in the preceding section, and the study their scattering
properties. We take advantage of these properties to establish
Proposition~\ref{prop:4bis}, and thus prove Theorem~\ref{thm:2}.

\subsection{Global well-posedness, scattering, and asymptotic splitting of
  trajectories}
\label{sec:glob-well-posedn}

We start with a proof of the well-posedness and scattering properties of the
Cauchy problem~\eqref{eq:1}. To effectively study the system on long time
scales, we need to understand the behavior of the term $\mathrm{Re} \langle
u_\pm(t),g\rangle_2$ appearing in the eigenprojectors and other classical
quantities. For this purpose, we establish the existence of scattering for
some $L^p$ spaces with $p$ sufficiently large to obtain an explicit decay, by
means of dispersive estimates.

\begin{prop}
  \label{globalwellposedness}

  There exists $\mathfrak g_0>0$ such that for every $0<\mathfrak g<\mathfrak{g}_0$, the
  Cauchy problem~\eqref{eq:1}
  % \begin{equation}
  %   i \partial_t{u}_\pm(t) = \nabla_{\overline{u}} \lambda_{\pm} (u^\pm(t)) = -\Delta u_\pm(t) \pm \mathfrak
  %   g^2 \frac{\mathrm{Re}\langle u_\pm(t),g \rangle g}{\sqrt{1+2 \mathfrak g^2
  %   \mathrm{Re} \langle u_\pm(t),g\rangle^2}}, \ \ \ u_\pm(0) = u_0 \in L^2
  % \end{equation}
  is globally well-posed in $\mathscr C(\R_+,L^2) \cap
  L^2(\R_+,L^6)$. Furthermore, there is scattering in $L^2 \cap L^\infty$, that is,
  for all $\underline{u}\in L^2$ there exists
  \begin{equation*}
    u^{\infty}_{\pm}:= \lim_{t\to +\infty}e^{-it\Delta}u_{\pm}(t)\;,
  \end{equation*}
  and the limit can be taken in any $L^p$ norm-topology, with $2\leq p\leq +\infty$. It
  follows that $u_{\pm}(t)$ can be seen as the unique solution of the ``Cauchy
  problem at infinity''
  \begin{equation}
    \label{eq:4}
    u_{\pm}(t)= e^{it\Delta}u_{\pm}^{\infty}\mp i\mathfrak{g}^2\int_{+\infty}^t e^{i(t-s)\Delta}\frac{\mathrm{Re}\langle u_\pm(s),g \rangle_2 \:g}{\sqrt{1+2 \mathfrak g^2 \mathrm{Re} \langle
        u_\pm(s),g\rangle^2_2}}\, \di s\;.
  \end{equation}
\end{prop}
\begin{remark*}
  Observe that equation ~\eqref{eq:1} also reads
  \begin{equation*}
    i \partial_t{u}_\pm(t) = \nabla_{\overline{u}} \lambda_{\pm} (u^\pm(t))\;,
  \end{equation*}
  \emph{i.e.}\ it is the Hamilton equation associated with the classical
  Hamiltonian function $\lambda_{\pm}(\cdot )$. Using Strichartz estimates, one should be
  able to further prove that the Cauchy problem at infinity \eqref{eq:4} is
  globally well-posed on some suitable asymptotic space; this however goes
  beyond the scope of this paper.
\end{remark*}
\begin{proof}
  We first establish local well-posedness using a standard fixed-point
  argument and Strichartz estimates~\cite[Chapter~8]{BCD_book}. The solution
  is to be understood in the Duhamel sense, that is
    $$
    u_\pm(t) = e^{i t \Delta} \underline{u} \mp i \mathfrak g^2 \int_0 ^t e^{i(t-s)\Delta}
    \frac{\mathrm{Re}\langle u_\pm(s),g \rangle_2 \:g}{\sqrt{1+2 \mathfrak g^2 \mathrm{Re} \langle
        u_\pm(s),g\rangle^2_2}}\, \di s\;.
    $$
    Let us consider the (Schrödinger-)admissible pair $(2,6)$ in dimension
    $3$.
    \begin{align*}
      \norm{u_\pm}_{L^2([0,T],L^6)} &\lesssim \norm{\underline{u}}_2 + \mathfrak g^2 \norm{ \frac{\mathrm{Re}\langle u_\pm(s),g \rangle_2\, g}{\sqrt{1+2 \mathfrak g^2 \mathrm{Re} \langle u_\pm(s),g\rangle^2_2}}}_{L^{2}([0,T],L^{\frac{6}{5}})} \\
      & \lesssim \norm{\underline{u}}_{2} + \mathfrak g^2 \norm{ \frac{\mathrm{Re}\langle u_\pm(s),g \rangle_2 }{\sqrt{1+2\mathfrak g^2  \mathrm{Re} \langle u_\pm(s),g\rangle^2_2}}}_{L^2([0,T])} \norm{g}_{{\frac{6}{5}}} \\
      &\lesssim \norm{\underline{u}}_{2} + \mathfrak g^2 \norm{ \norm{u_\pm(s)}_{6} }_{L^2([0,T])} \norm{g}_{{\frac{6}{5}}}^2,    
    \end{align*}
    so that
$$
\norm{u}_{L^2([0,T],L^6)} \leq \frac{C \norm{\underline{u}}_{2}}{1-C\mathfrak g^2
  \norm{g}_{{\frac{6}{5}}}^2},
$$
where $C$ is some Strichartz constant. Therefore, $u \in L^2(\R_+,L^6)$ when
$\mathfrak g$ is sufficiently small. Using again Strichartz estimates, we
prove that $u \in \mathscr C(\R_+,L^2)$ and the integral term
$$
I_t = \mathfrak g^2 \int_0 ^t e^{-i s \Delta} \frac{\mathrm{Re}\langle u_\pm(s),g \rangle_2\,
  g}{\sqrt{1+2 \mathfrak g^2 \mathrm{Re} \langle u_\pm(s),g\rangle^2_2}}\, \di s
$$
defines a Cauchy sequence in $L^2$ and therefore $\underset{t \rightarrow + \infty} \lim
I_t$ exists and there is scattering in $L^2$. Now, compute
$$
i \partial_t (e^{-i t \Delta}u_\pm(t)) =\mathfrak g^2 \frac{\mathrm{Re}\langle  u_\pm(t),
   g \rangle_2 }{\sqrt{1+2 \mathfrak g^2 \mathrm{Re} \langle  u_\pm(t),
     g\rangle^2_2}} e^{-i t \Delta} g.$$ It follows that for every $p >6$,
 \begin{equation}
   \label{eq:13}
   \norm{\partial_t (e^{-i t \Delta} u_\pm(t))}_{{p}} \leq \mathfrak g \norm{e^{-i t \Delta}
  g}_{{p}} \leq {\mathfrak g} (4 \pi \abs{t})^{-3/2(1/p'-1/p)} \norm{g}_{{p'}}
\in L^1([1,+\infty))\;,
 \end{equation}
since $3/2(1/p'-1/p) > 1$ for $p>6$. Therefore, $(e^{-i t \Delta} u_\pm(t))_{t\geq 0}$
is a Cauchy sequence in $L^{p}$ and scattering also holds in $L^{p}$ with the
same limit. Hölder's inequality implies that the limit also exists in $L^p$
for $2 \leq p \leq 6$, which concludes the proof.
\end{proof}

We must now ensure that the trajectories asymptotically split in long time,
to allow for a sharp detection. It turns out that we are able to
quantitatively estimate the gap between the asymptotic $+$ and $-$ states in
some suitable $L^p$ spaces.
\begin{prop}[(Asymptotic) Splitting of trajectories]
  \label{asymptotic splitting}
  Let $g \in \mathscr{S}$, and let $\underline{u}\in L^2$ be such that there exists $t\in \mathbb{R}$
  so that $\mathrm{Re}\langle e^{it\Delta}\underline{u},g\rangle _2\neq 0$.
  Then there exist $C(\underline{u},g), C(t,\underline{u},g)>0$ and
  $\mathfrak g_0 = \mathfrak g_0(\underline{u},g) >0$, such that for every
  $0<\mathfrak g< \mathfrak{g}_0$ and $t\in \R$, the following lower bounds hold:
  \begin{align*}
    \norm{u_+(t)-u_-(t)}_{2} &\geq C(t,\underline{u},g) \mathfrak g^2\;\;\mbox{and}\;\;
    \norm{u_+^\infty - u_- ^\infty}_{{2}} \geq C(\underline{u},g) \mathfrak g^2,
  \end{align*}
  with $\lim_{t\to \infty}C(t,\underline{u},g) > 0$.
\end{prop}
\begin{proof}
  We expand the solution $u_\pm(t)$ as a Taylor series in terms of $\mathfrak
  g$:
    $$
    u_\pm(t) = u^0(t) + \mathfrak g^2 u_\pm^1 (t) +R_{\mathfrak g,\pm}(t),
    $$
    where $u^0(t)$ and $u^1_\pm(t)$ are solutions to the Cauchy problems
    $$
    \begin{cases}i \partial_t u^0 = - \Delta u^0 \\
      u^0(0)=\underline{u}
    \end{cases} \ \ \ \ \ \mathrm{and} \ \ \ \ \
    \begin{cases} i \partial_t u^1_\pm = - \Delta u^1_\pm \pm \mathrm{Re}\, \langle u^0(s), g \rangle_2\, g \\
      u^1(0)=0\end{cases}
    $$
    \textbf{Step 1.} We first show that $u^0(t), u^1_\pm(t)$ (and hence
    $R_{\mathfrak g,\pm}(t)$) scatter, and that there is splitting at first
    order (in $\mathfrak g^2$), both at finite time and asymptotically:
    \begin{align*}
      &u_+^1(t) \neq u_-^1(t) \text{ for every } t>0, \text{ and}\underset{t \rightarrow +\infty} \lim e^{-i t \Delta} u^1_+(t) \neq \underset{t \rightarrow +\infty} \lim e^{-i t \Delta} u^1_-(t).
    \end{align*}
    The existence of scattering for $u^0(t)$ is obvious, and for $u^1_\pm(t)$
    it suffices to remark that $u^1_\pm(t)$ is the solution of a Schrödinger
    equation with a constant source term that belongs to a Strichartz
    space. Regarding the splitting, we write, using the Duhamel formula for
    $t\in \R_+$:
$$
e^{-i t \Delta} u^1_\pm(t) = \mp i \int_0^t e^{-i s \Delta} \mathrm{Re}\,
\langle u^0(s),g \rangle_2\, g\, \di s = \mp i \int_0^t \mathrm{Re}\, \langle \underline{u},
e^{-i s \Delta} g \rangle_2\, e^{-i s \Delta} g\, \di s.
$$
Therefore, splitting occurs if and only if the integral on the right-hand
side is nonzero. We argue by contradiction and suppose that it vanishes. By
taking the scalar product with $\underline{u}$ and then the imaginary part,
one obtains
$$
\int_0^t \left( \mathrm{Re} \langle \underline{u}, e^{-i s \Delta} g \rangle_2 \right)^2
\di s = 0,
$$
so that $ \mathrm{Re} \langle \underline{u}, e^{-i s \Delta} g \rangle_2 =0$ almost everywhere
on $[0,t]$, and hence everywhere by continuity. This contradicts the
assumption that $\mathrm{Re}\, \langle e^{it\Delta}\underline{u} , g \rangle_2\neq 0$ for some
$t>0$. Therefore, for $t>0$,
$$
\norm{u_+^1(t) - u_-^1(t)}_{2} = \norm{e^{-i t \Delta} u^1_+(t)-e^{-i t \Delta}
  u^1_-(t)}_{2} >0.
$$
For the scattering states, the same argument allows to conclude, since
$$
\int_0^\infty \mathrm{Re}\, \langle \underline{u}, e^{-i s \Delta} g \rangle_2 e^{-i s
  \Delta} g\, \di s \neq 0,
$$
thanks again to the assumption on $\underline{u}$.

\medskip

\noindent \textbf{Step 2.} There exists $C>0$ such that for every $t \in \R$,
$$
\norm{e^{i t \Delta} R_{\mathfrak g,\pm}(t)}_{{2}} = \norm{R_{\mathfrak
    g,\pm}(t)}_{{2}} \leq C \mathfrak g^4.$$ $R_{\mathfrak g,\pm}(t)$ is a
solution to
$$
i \partial_t R_{\mathfrak g,\pm}(t) = - \Delta R_{\mathfrak g,\pm}(t) \pm
\mathfrak g^2 \underbrace{\left( \frac{\mathrm{Re}\langle u(t),g\rangle_2
    }{\sqrt{1+2 \mathfrak g^2 \mathrm{Re}\langle u(t),g\rangle^2_2}} -
    \mathrm{Re}\langle u^0(t),g\rangle \right)}_{=: F(t)}g\;.
$$
It is easy to check that
\begin{multline*}
  \abs{F(t)} \leq \abs{\mathrm{Re} \langle u(t)-u^0(t), g \rangle_2} + \mathfrak g^2 \abs{\mathrm{Re} \langle u^0(t),g \rangle_2}^3. \\
  \leq \mathfrak g^2 \left( \abs{ \langle u_\pm^1(t),g \rangle_2 } + \abs{ \langle u^0(t),g \rangle_2 }^3 \right) + \norm{R_{\mathfrak g,\pm}(t)}_{{6}} \norm{g}_{{\frac{6}{5}}}.
\end{multline*}
We notice that $\abs{ \langle u_\pm^1(t),g \rangle_2 } + \abs{ \langle u^0(t),g
  \rangle_2 }^3 \in L^2 (\R_+)$. Indeed, one first has
$$
\norm{\abs{ \langle u^0(t),g \rangle_2 }^3}_{L^2(\R_+)} \leq \norm{e^{i t \Delta} u_0}_{L^6(\R_+,
  L^{\frac{18}{7}})}^3 \norm{g}_{{\frac{18}{11}}}^3 \leq C \norm{\underline{u}}_{2}^3
\norm{g}_{{\frac{18}{11}}}^3.
$$
Secondly,
\begin{align*}
  \norm{\abs{ \langle u_\pm^1(t),g \rangle_2 }}_{L^2(\R_+)} &\leq \norm{u_\pm^1}_{L^2(\R_+,L^6)} \norm{g}_{{\frac{6}{5}}}
  \\& \leq\norm{\int_0^t e^{i(t-s) \Delta} \mathrm{Re}\langle u^0(s),g\rangle_2\, g\,  \di s }_{L^2(\R_+,L^6)} \norm{g}_{\frac{6}{5}} 
  \\&\leq \norm{\mathrm{Re}\langle u^0(s),g\rangle_2\, g   }_{L^2(\R_+,L^{\frac{6}{5}})} \norm{g}_{\frac{6}{5}} \\
  &\leq \norm{e^{i t \Delta}\underline{u}}_{L^2(\R_+,L^6)} \norm{g}_{\frac{6}{5}}^3 \\
  &\leq C \norm{\underline{u}}_2\norm{g}_{\frac{6}{5}}^3\;.
\end{align*}
Therefore, since $R_{\mathfrak g,\pm}(0) = 0$, using Strichartz estimates for
the Duhamel formula one obtains:
\begin{align*}
  \norm{R_{\mathfrak g,\pm}}_{L^2([0,t],L^{6})} &\leq \mathfrak g^2 \norm{\int_0^t e^{i(t- s)\Delta}F(s)  g \di s}_{L^2([0,t],L^6)} \\
  &\leq \mathfrak g^2 \norm{ F}_{L^2([0,t])} \norm{g}_{\frac{6}{5}} 
  \\
  &\leq C \mathfrak g^4 + C \mathfrak g^2\norm{R_{\mathfrak g,\pm}}_{L^2([0,t],L^6)}  \norm{g}_{\frac{6}{5}}^2,
\end{align*}
so that $R_{\mathfrak g,\pm} \in L^2(\R_+,L^6)$ for $\mathfrak g$ small enough,
and the existence of the scattering state $R_{\mathfrak g,\pm}^\infty$ in $L^2$
follows. Finally, by one final application of Strichartz estimates, the
following bounds hold:
$$
\norm{R_{\mathfrak g,\pm}(t)}_{2} + \norm{R_{\mathfrak g,\pm}^\infty}_{2}
\leq C \mathfrak g^4.
$$
\textbf{Step 3.}  We first work at a fixed time $t>0$. Then, for $\mathfrak
g$ small enough,
$$
\norm{u_+(t) - u_-(t)}_{{2}} \geq \mathfrak g^2 \norm{u_+^1 (t)-u_-^1
  (t)}_{{2}} - C \mathfrak g^4 \geq \frac{\mathfrak g^2}{2}
\norm{u_+^1(t)-u_-^1 (t)}_{2}.
$$
For the asymptotic splitting, taking the limit in $L^2$ of $e^{-i t \Delta}
u(t) = e^{-i t \Delta} (u^0(t) + \mathfrak g^2 u^1_\pm(t) + R_{\mathfrak
  g,_\pm}(t))$, we obtain, since scattering holds for each individual
component,
$$
u_\pm^\infty = \underline{u} + \mathfrak g^2 u^{1,\infty}_\pm + R_{\mathfrak
  g,_\pm}^\infty,
$$
with $\norm{R_{\mathfrak g,\pm}^\infty}_{2} \leq C \mathfrak g^4$, by the
uniform estimate of Step 2. Therefore,
$$
\norm{u_+^\infty - u_-^\infty}_{{2}} \geq \mathfrak g^2
\norm{u_+^{1,\infty}-u_-^{1,\infty}}_{{2}} - C \mathfrak g^4 \geq
\frac{\mathfrak g^2}{2} \norm{u_+^{1,\infty}-u_-^{1,\infty}}_{{2}}.
$$
\end{proof}

\subsection{Proof of Proposition~\ref{prop:4bis}}
\label{sec:proof-theor-refthm:2}

We are now in a position to prove the existence of a semiclassical
measurement scheme that solves the measurement problem within error
limits. The proof is split in a number of steps. The first and easiest
concerns the bound for the convergence of the initial coherent state $\varrho\otimes \lvert
\underline{u}_{\varepsilon}\rangle\langle \underline{u}_{\varepsilon}\rvert$ to the state-valued measure $\varrho \,
\mathrm{d}\mu(x)$.

\begin{lemma}
  \label{lemma:1}
  Let $\varrho$ be a state on $\mathscr{H}_{\mathcal{S}}$, and $\underline{u}\in L^2$. Then for all
  compact operator $K$ on $\mathscr{H}_{\mathcal{S}}$ and $\xi\in L^2$,
  \begin{equation*}
   \biggl\lvert\, \mathrm{tr}_{\mathcal{S}}\biggl(K\, \Bigl(\widehat{\bigl(\varrho\otimes \lvert \underline{u}_{\varepsilon}\rangle\langle \underline{u}_{\varepsilon}\rvert\bigr)}(\xi)- \widehat{\bigl(\varrho\,\mathrm{d}\delta_{\underline{u}}\bigr)}(\xi)\Bigr)\biggr)  \,\biggr\rvert_{}^{}  \leq \frac{1}{2}\lVert K  \rVert_{}^{}\, \varepsilon\;.
  \end{equation*}
\end{lemma}
\begin{proof}
  The noncommutative Fourier transform of the coherent state is explicit:
  \begin{equation*}
    \widehat{\lvert \underline{u}_{\varepsilon}\rangle\langle \underline{u}_{\varepsilon}\rvert}(\xi)= e^{- \frac{\varepsilon}{2}\lVert \xi  \rVert_2^2} \,e^{i\sqrt{2} \mathrm{Re}\, \langle \xi  , \underline{u} \rangle_2}\;;
  \end{equation*}
  hence it follows that
  \begin{multline*}
   \biggl\lvert\, \mathrm{tr}_{\mathcal{S}}\biggl(K\, \Bigl(\widehat{\bigl(\varrho\otimes \lvert \underline{u}_{\varepsilon}\rangle\langle \underline{u}_{\varepsilon}\rvert\bigr)}(\xi)- \widehat{\bigl(\varrho\,\mathrm{d}\delta_{\underline{u}}\bigr)}(\xi)\Bigr)\biggr)  \,\biggr\rvert_{}^{} \\= \Bigl\rvert\,\mathrm{tr}_{\mathscr{H}_{\mathcal{S}}}(K\varrho)\,e^{i\sqrt{2} \mathrm{Re}\, \langle \xi  , \underline{u} \rangle_2}\, \Bigl( e^{- \frac{\varepsilon}{2}\lVert \xi  \rVert_2^2}-1\Bigr)\Bigr\rvert\leq  \frac{1}{2}\lVert K  \rVert_{}^{}\lVert \xi  \rVert_2^2\,\varepsilon \;.
  \end{multline*}
\end{proof}
In order to prove a bound for the final state, let us make the following
split. Beforehand, for simplicity let us consider $\varrho=\lvert\psi\rangle\langle\psi\rvert$, with $\psi=\alpha_+\lvert+\rangle+\alpha_-\lvert-\rangle$, as the
general case follows easily. Also, let us use/recall the notations
\begin{gather*}
  \boldsymbol{\Psi}_{\varepsilon}(t_{\varepsilon}):= \boldsymbol{U}_{\varepsilon}(\psi\otimes \lvert \underline{u}_{\varepsilon}\rangle)=e^{-i \frac{t_{\varepsilon}}{\varepsilon}\boldsymbol{H}_{\varepsilon}}(\psi\otimes \lvert \underline{u}_{\varepsilon}\rangle)\;;\\
  \boldsymbol{\varrho}_{\varepsilon}:= \boldsymbol{U}_{\varepsilon} \lvert\psi\otimes  \underline{u}_{\varepsilon}\rangle\langle\psi\otimes  \underline{u}_{\varepsilon}\rvert \boldsymbol{U}_{\varepsilon}^{*}\;;\\
  \boldsymbol{\varrho}_{\mathrm{ideal}}:= \sum_{j\in \{+,-\}}^{} \lvert \langle \psi  , j \rangle_{\mathscr{H}_{\mathcal{S}}}  \rvert_{}^2\, \lvert j\rangle\langle j\rvert \,\delta_{u^{\infty}_j}\;.
\end{gather*}
\begin{equation}
  \label{eq:5}
  \Bigl\lvert\, \mathrm{tr}_{\mathcal{S}}\Bigl(K\, \bigl(\hat{\boldsymbol{\varrho}}_{\varepsilon}(\xi)- \hat{\boldsymbol{\varrho}}_{\mathrm{ideal}}(\xi)\bigr)\Bigr)\,\Bigr\rvert\leq D_{1}(\varepsilon,t_{\varepsilon})+ D_2(\varepsilon,t_{\varepsilon})+D_3(\varepsilon,t_{\varepsilon})+D_4(t_{\varepsilon})+D_5(t_{\varepsilon})\;,
\end{equation}
where $D_1$ encodes the wave packet approximation:
\begin{multline}
  \label{eq:6}
  D_1(\varepsilon,t_{\varepsilon}):= \biggl\lvert\Bigl\langle \Psi_{\varepsilon}(t_{\varepsilon}) -\mspace{-10mu} \sum_{j\in \{+,-\}}^{}\mspace{-10mu}\alpha_j\,\pi_j^{t_{\varepsilon}}\,\mathcal{R}_j(t_{\varepsilon})\,\lvert j \rangle\:\otimes\: \Psi^{\varepsilon}_{\mathrm{swp},j}(t_{\varepsilon})\:,\: K\otimes e^{i\phi_{\varepsilon}(e^{it_{\varepsilon}\Delta}\xi)} \Psi_{\varepsilon}(t_{\varepsilon}) \Bigr\rangle_{\mathbb{C}^2\otimes \mathscr{F}(L^2)}\\
  + \Bigl\langle \Psi_{\varepsilon}(t_{\varepsilon})\:,\: K\otimes e^{i\phi_{\varepsilon}(e^{it_{\varepsilon}\Delta}\xi)} \Bigl(\Psi_{\varepsilon}(t_{\varepsilon}) -\mspace{-10mu} \sum_{j\in \{+,-\}}^{}\mspace{-10mu}\alpha_j\,\pi_j^{t_{\varepsilon}}\,\mathcal{R}_j(t_{\varepsilon})\,\lvert j \rangle\:\otimes\: \Psi^{\varepsilon}_{\mathrm{swp},j}(t_{\varepsilon})\Bigr) \Bigr\rangle_{\mathbb{C}^2\otimes \mathscr{F}(L^2)}\biggr\rvert\;;
\end{multline}
$D_2$ (respectively, $D_3$) takes into account the quantum/classical
approximation of the diagonal (respectively, off-diagonal) components of the
spin-observable:
\begin{multline}
  \label{eq:7}
  D_2(\varepsilon,t_{\varepsilon}):= \sum_{j\in \{+,-\}}^{}\lvert \alpha_j  \rvert_{}^2\Bigl\lvert\langle j  \vert\mathcal{R}_j^{*}(t_{\varepsilon})K\mathcal{R}_{j}(t_{\varepsilon})\rvert j  \rangle_{\mathbb{C}^2}\Bigr\rvert\\\Bigl\lvert\Bigl\langle \Psi^{\varepsilon}_{\mathrm{swp},j}(t_{\varepsilon})  ,  \Bigl(e^{i\phi_{\varepsilon}(e^{it_{\varepsilon}\Delta}\xi)}-e^{i \sqrt{2} \mathrm{Re}\, \langle e^{it_{\varepsilon}\Delta}\xi  , u_j(t_{\varepsilon}) \rangle_2}\Bigr)\Psi^{\varepsilon}_{\mathrm{swp},j}(t_{\varepsilon})\Bigr\rangle_{\mathscr{F}(L^2)}\Bigr\rvert\;,
\end{multline}
\begin{multline}
  \label{eq:10}
  D_3(\varepsilon,t_{\varepsilon}):= \sum_{\substack{j,j'\in \{+,-\}\\j\neq j'}}^{}\Bigl\lvert\bar{\alpha}_j\alpha_{j'}\langle j  \vert\mathcal{R}_j^{*}(t_{\varepsilon})K\mathcal{R}_{j'}(t_{\varepsilon})\rvert j'  \rangle_{\mathbb{C}^2}\Bigr\rvert\\\Bigl\lvert\Bigl\langle \Psi^{\varepsilon}_{\mathrm{swp},j}(t_{\varepsilon})  , \, e^{i\phi_{\varepsilon}(e^{it_{\varepsilon}\Delta}\xi)}\,\Psi^{\varepsilon}_{\mathrm{swp},j'}(t_{\varepsilon})\Bigr\rangle_{\mathscr{F}(L^2)}\Bigr\rvert\;;
\end{multline}
finally, $D_4$ and $D_5$ encode the effect of the scattering process:
\begin{gather}
  \label{eq:8}
  D_4(t_{\varepsilon}):= \sum_{j\in \{+,-\}}^{}\lvert \alpha_j  \rvert_{}^2\Bigl\lvert\langle j  \vert\mathcal{R}_j^{*}(t_{\varepsilon})K\mathcal{R}_{j}(t_{\varepsilon})\rvert j  \rangle_{\mathbb{C}^2}\Bigr\rvert\: \Bigl\lvert e^{i \sqrt{2} \mathrm{Re}\, \langle \xi  , e^{-it_{\varepsilon}\Delta}u_j(t_{\varepsilon}) \rangle_2}- e^{i \sqrt{2} \mathrm{Re}\, \langle \xi  , u_j^{\infty} \rangle_2}\Bigr\rvert\;,\\
  \label{eq:9}
  D_5(t_{\varepsilon}):= \sum_{j\in \{+,-\}}^{}\lvert \alpha_j  \rvert_{}^2 \Bigl\lvert\langle j  \vert\mathcal{R}_j^{*}(t_{\varepsilon})\pi^{t_{\varepsilon}}_j K\pi^{t_{\varepsilon}}_{j}\mathcal{R}_{j}(t_{\varepsilon})- K\rvert j  \rangle_{\mathbb{C}^2}\Bigr\rvert\;.
\end{gather}
Let us focus momentarily forget about $D_1$, since it can be bound using
Proposition~\ref{prop:3}, and focus on the other terms.
\begin{lemma}
  \label{lemma:2}
  There exists constants $C_3,C_4>0$ such that
  \begin{equation*}
    D_2(\varepsilon,t_{\varepsilon})\leq C_3 \lVert K  \rVert_{}^{} \bigl(\lVert \xi  \rVert_{H^1}^{}+\lVert \xi  \rVert_{H^1}^2\bigr)\,\sqrt{\varepsilon}\,e^{C_4 t_{\varepsilon}}\;.
  \end{equation*}
\end{lemma}
\begin{proof}
  Let us start with $D_2$. Observe that
  \begin{multline*}
    D_2(\varepsilon,t_{\varepsilon})\leq \lVert K  \rVert_{}^{}\sum_{j\in \{+,-\}}^{}\Bigl\lvert\Bigl\langle \Psi^{\varepsilon}_{\mathrm{swp},j}(t_{\varepsilon})  ,  \Bigl(e^{i\phi_{\varepsilon}(e^{it_{\varepsilon}\Delta}\xi)}-e^{i \sqrt{2} \mathrm{Re}\, \langle e^{it_{\varepsilon}\Delta}\xi  , u_j(t_{\varepsilon}) \rangle_2}\Bigr)\Psi^{\varepsilon}_{\mathrm{swp},j}(t_{\varepsilon})\Bigr\rangle_{\mathscr{F}(L^2)}\Bigr\rvert\\\leq 2\lVert K  \rVert_{}^{}\sum_{j\in \{+,-\}}^{}\Bigl\lvert\Bigl\langle U_j(t_{\varepsilon},0)_{\varepsilon}\,\Omega  ,  \Bigl(e^{i\phi_{\varepsilon}(e^{it_{\varepsilon}\Delta}\xi)}-1\Bigr)U_j(t_{\varepsilon},0)_{\varepsilon}\,\Omega\Bigr\rangle_{\mathscr{F}(L^2)}\Bigr\rvert\;,
  \end{multline*}
  since by the well-known properties of exponential of creation and
  annihilation operators,
  \begin{multline*}
    \langle \Psi^{\varepsilon}_{\mathrm{swp},\pm}(t_{\varepsilon})  ,e^{i\phi_{\varepsilon}(e^{it_{\varepsilon}\Delta}\xi)}\Psi^{\varepsilon}_{\mathrm{swp},\pm}(t_{\varepsilon})  \rangle_{\mathscr{F}(L^2)}\\= e^{i \sqrt{2} \mathrm{Re}\, \langle e^{it_{\varepsilon}\Delta}\xi  , u_j(t_{\varepsilon}) \rangle_2} \,\langle U_{\pm}(t_{\varepsilon},0)_{\varepsilon}\,\Omega  ,  e^{i\phi_{\varepsilon}(e^{it_{\varepsilon}\Delta}\xi)}U_{\pm}(t_{\varepsilon},0)_{\varepsilon}\,\Omega\rangle_{\mathscr{F}(L^2)}\;.
  \end{multline*}
  Now, define
  \begin{equation*}
    G_{\varepsilon}(\lambda):= \langle U_{\pm}(\lambda t_{\varepsilon},0)_{\varepsilon}\,\Omega  ,  e^{i\phi_{\varepsilon}(e^{it_{\varepsilon}\Delta}\xi)}U_{\pm}(\lambda t_{\varepsilon},0)_{\varepsilon}\,\Omega\rangle_{\mathscr{F}(L^2)}\;.
  \end{equation*}
  We prove that
  \begin{equation*}
    \lvert G_{\varepsilon}(\lambda)  \rvert_{}^{}\leq C\bigl(\lVert \xi  \rVert_{H^1}^{}+\lVert \xi  \rVert_{H^1}^2\bigr)\,\sqrt{\varepsilon}\,e^{C_2 t_{\varepsilon}}
  \end{equation*}
  as follows:
  \begin{multline*}
    i\varepsilon \dot{G}_{\varepsilon}(\lambda) = - \langle U_{\pm}(\lambda t_{\varepsilon},0)_{\varepsilon}\,\Omega  ,  [A_{\pm}(t)_{\varepsilon}\,,\,e^{i\phi_{\varepsilon}(e^{it_{\varepsilon}\Delta}\xi)}]U_{\pm}(\lambda t_{\varepsilon},0)_{\varepsilon}\,\Omega\rangle_{\mathscr{F}(L^2)}\\=- \langle U_{\pm}(\lambda t_{\varepsilon},0)_{\varepsilon}\,\Omega  ,  e^{i\phi_{\varepsilon}(e^{it_{\varepsilon}\Delta}\xi)}\bigl(\phi_{\varepsilon}(-i\varepsilon e^{it_{\varepsilon}\Delta}\Delta\xi)+ \frac{\varepsilon^2}{2}\lVert \xi  \rVert_{H^1}^2 \\+2\phi_{\varepsilon}(g) \mathrm{Re}\langle i\varepsilon e^{it\Delta_{\varepsilon}}\xi  , g \rangle_2+(\mathrm{Re}\langle i\varepsilon e^{it\Delta_{\varepsilon}}\xi  , g \rangle_2)^2\bigr)U_{\pm}(\lambda t_{\varepsilon},0)_{\varepsilon}\,\Omega\rangle_{\mathscr{F}(L^2)}\;.
  \end{multline*}
  Now using the standard Fock space estimates to bound the field operators
  with $(\mathrm{d}\Gamma_{\varepsilon}(-\Delta)+\varepsilon)^{1/2}$, and then the bound at the end of
  Lemma~\ref{equationU_pm} for the propagation of
  $(\mathrm{d}\Gamma_{\varepsilon}(-\Delta)+\varepsilon)^{1/2}$ through $U_{\pm}(t,0)_{\varepsilon}$, one is able to
  conclude.
\end{proof}
Let us now turn our attention to $D_3(\varepsilon,t_{\varepsilon})$. The fact these cross terms
vanish in the semiclassical limit $\varepsilon\to 0$ is crucial to have a working
semiclassical measurement scheme $\mathscr{M}_{\varepsilon}^+(\sigma_z)$. It is a crucial consequence
of the general semiclassical Proposition~\ref{prop:superpsuppr}; the explicit
bound on the accuracy on the other hand is obtained exploiting the special
structure of (squeezed) coherent states.
\begin{lemma}
  \label{lemma:4}
  For all $\varepsilon\in (0,1/2)$, there exists constants $C_5,C_6>0$ such that
  \begin{equation*}
    D_3(\varepsilon,t_{\varepsilon})\leq C_5 \lVert K  \rVert_{}^{} e^{- \frac{C_6}{\varepsilon}}\;.
  \end{equation*}
\end{lemma}
\begin{proof}
  The exponential decay of $D_3$ can be proved by an explicit computation,
  thanks to the special structure of squeezed coherent states
  \citep[see][\textsection3.2 and \textsection8.5]{Combescure_Robert_book}, that yields
  \begin{equation*}
    D_3(\varepsilon,t_{\varepsilon})\leq C_5\lVert K  \rVert_{}^{} e^{- \frac{C}{\varepsilon}\lVert u_+(t_{\varepsilon})- u_-(t_{\varepsilon})  \rVert_2^2}\;.
  \end{equation*}
  Now, since we will choose $t_{\varepsilon}= \mathrm{O}(-\ln \varepsilon)$, by
  Proposition~\ref{asymptotic splitting} we have 
  \begin{equation*}
    \lVert u_+(t_{\varepsilon})- u_-(t_{\varepsilon})  \rVert_2^2\geq \inf_{\varepsilon\in (0,1/2)} C(t_{\varepsilon},\underline{u},g) \mathfrak{g}^2>0\;,
  \end{equation*}
  and the result thus follows (in fact, we only need separation from zero of
  $t_{\varepsilon}$ to have a strictly positive lower bound, so it holds in any
  interval $\varepsilon\in (0,\delta)$, $\delta<1$).

  Let us conclude by remarking how one can apply
  Proposition~\ref{prop:superpsuppr} in this case to have the convergence as
  $\varepsilon\to 0$, albeit with no explicit bound on the accuracy: due to the
  $\varepsilon$-dependence in $t_{\varepsilon}$, that this convergence indeed holds might not be
  obvious at first sight. In fact, we can prove that the semiclassical
  measure of $e^{i
    \frac{t_{\varepsilon}}{\varepsilon}\mathrm{d}\Gamma_{\varepsilon}(-\Delta)}\lvert\Psi_{\mathrm{swp},\pm}(t_{\varepsilon})\rangle\langle
  \Psi_{\mathrm{swp},\pm}(t_{\varepsilon})\rvert e^{-i \frac{t_{\varepsilon}}{\varepsilon}\mathrm{d}\Gamma_{\varepsilon}(-\Delta)}$ is
  $\delta_{u_{\pm}^{\infty}}$:
  \begin{equation*}
    e^{i
    \frac{t_{\varepsilon}}{\varepsilon}\mathrm{d}\Gamma_{\varepsilon}(-\Delta)}\lvert\Psi_{\mathrm{swp},\pm}(t_{\varepsilon})\rangle\langle\Psi_{\mathrm{swp},\pm}(t_{\varepsilon})\rvert e^{-i
    \frac{t_{\varepsilon}}{\varepsilon}\mathrm{d}\Gamma_{\varepsilon}(-\Delta)}\xrightarrow[\varepsilon_n\to 0]{\mathcal{F}} \delta_{u_{\pm}^{\infty}}\;.
  \end{equation*}
  Once we prove this, the fact that the transition amplitude between $e^{i
    \frac{t}{\varepsilon}\mathrm{d}\Gamma_{\varepsilon}(-\Delta)}\Psi_{\mathrm{swp},+}(t_{\varepsilon})$ and $e^{i
    \frac{t}{\varepsilon}\mathrm{d}\Gamma_{\varepsilon}(-\Delta)}\Psi_{\mathrm{swp},-}(t_{\varepsilon})$ vanish is a
  direct consequence of the separation $u_-^{\infty}\neq u_+^{\infty}$ proved in
  Proposition~\ref{asymptotic splitting}, applying
  Proposition~\ref{prop:superpsuppr}.

  To prove that $\delta_{u_{\pm}^{\infty}}$ is the semiclassical measure, from the proof of
  Lemma~\ref{lemma:2}, we have 
  \begin{multline*}
    \Bigl\lvert \langle \Psi^{\varepsilon}_{\mathrm{swp},\pm}(t_{\varepsilon})  ,e^{i\phi_{\varepsilon}(e^{it_{\varepsilon}\Delta}\xi)}\Psi^{\varepsilon}_{\mathrm{swp},\pm}(t_{\varepsilon})  \rangle_{\mathscr{F}(L^2)}- e^{i \sqrt{2} \mathrm{Re}\langle \xi  , u_{\pm}^{\infty} \rangle_2}\Bigr\rvert\\\leq \Bigl\lvert \langle \Psi^{\varepsilon}_{\mathrm{swp},\pm}(t_{\varepsilon})  ,e^{i\phi_{\varepsilon}(e^{it_{\varepsilon}\Delta}\xi)}\Psi^{\varepsilon}_{\mathrm{swp},\pm}(t_{\varepsilon})  \rangle_{\mathscr{F}(L^2)}- e^{i \sqrt{2} \mathrm{Re}\langle e^{it_{\varepsilon} \Delta}\xi  , u_{\pm}(t_{\varepsilon}) \rangle_2}\Bigr\rvert\\+\Bigl\lvert e^{i \sqrt{2} \mathrm{Re}\langle \xi  , e^{-it_{\varepsilon} \Delta}u_{\pm}(t_{\varepsilon}) \rangle_2} - e^{i \sqrt{2} \mathrm{Re}\langle \xi  , u_{\pm}^{\infty} \rangle_2}\Bigr\rvert \\
    \leq C_1 \lVert K  \rVert_{}^{}(\lVert \xi  \rVert_{H^1}^{}+ \lVert \xi  \rVert_{H^1}^2)\, \sqrt{\varepsilon}\, e^{C_2 t_{\varepsilon}}+ \Bigl\lvert 1 - e^{i \sqrt{2} \mathrm{Re}\langle \xi  , u_{\pm}^{\infty}-e^{-it_{\varepsilon} \Delta}u_{\pm}(t_{\varepsilon}) \rangle_2}\Bigr\rvert\;.
  \end{multline*}
  Therefore, by choosing $t_{\varepsilon}< \frac{-\ln \varepsilon}{2C_2+\delta}$, $\delta>0$, we get from
  Proposition~\ref{globalwellposedness} that for all $\xi\in \mathscr{S}$,
  \begin{equation*}
    \lim_{\varepsilon\to 0}\Bigl\lvert \langle \Psi^{\varepsilon}_{\mathrm{swp},\pm}(t_{\varepsilon})  ,e^{i\phi_{\varepsilon}(e^{it_{\varepsilon}\Delta}\xi)}\Psi^{\varepsilon}_{\mathrm{swp},\pm}(t_{\varepsilon})  \rangle_{\mathscr{F}(L^2)}- e^{i \sqrt{2} \mathrm{Re}\langle \xi  , u_{\pm}^{\infty} \rangle_2}\Bigr\rvert=0\;.
  \end{equation*}
\end{proof}
The last two terms, $D_4$ and $D_5$ are already quasi-classical, and only the
long time asymptotics is taken into account, and in fact the $\varepsilon$-dependence
is only through $t_{\varepsilon}$.
\begin{lemma}
  \label{lemma:5}
  There exist a constant $C_7>0$ such that
  \begin{equation*}
    D_4(t_{\varepsilon})\leq C_7 \lVert K  \rVert_{}^{}\lVert \xi  \rVert_1^{}\:t^{-\frac{1}{2}}_{\varepsilon}\;.
  \end{equation*}
\end{lemma}
\begin{proof}
  To bound $D_4$, we just need to bound
  \begin{multline*}
    \Bigl\lvert e^{i \sqrt{2} \mathrm{Re}\, \langle \xi  , e^{-it_{\varepsilon}\Delta}u_\pm(t_{\varepsilon}) \rangle_2} - e^{i \sqrt{2} \mathrm{Re}\, \langle \xi  , u_\pm^{\infty} \rangle_2}\Bigr\rvert\\= \Bigl\lvert 1- e^{i \sqrt{2} \mathrm{Re}\, \langle \xi  , u_\pm^{\infty}-e^{-it_{\varepsilon}\Delta}u_{\pm}(t_{\varepsilon}) \rangle_2}\Bigr\rvert\leq  \sqrt{2}\lVert \xi  \rVert_{1}\lVert u_{\pm}^{\infty}- e^{-it_{\varepsilon}\Delta}u_{\pm}(t_{\varepsilon})  \rVert_{\infty}\;, 
  \end{multline*}
  where we used Hölder's inequality (recall that $\xi\in \mathscr{S}(\mathbb{R}^3)$).  Now,
  recalling that $u_{\pm}^{\infty}$ is the $L^{\infty}$-limit point of
  $e^{-it_{\varepsilon}\Delta}u_{\pm}(t)$ as $t\to \infty$, we can write
  \begin{multline*}
    \Bigl\lVert u_{\pm}^{\infty}- e^{-it_{\varepsilon}\Delta}u_{\pm}(t_{\varepsilon})  \Bigr\rVert_{\infty}= \Bigl\lVert \int_{t_{\varepsilon}}^{\infty}\partial_s\bigl(e^{-is\Delta}u_{\pm}(s)\bigr)  \mathrm{d}s  \Bigr\rVert_{\infty}^{}\\\leq \int_{t_{\varepsilon}}^{\infty}\Bigl\lVert \partial_s\bigl(e^{-is\Delta}u_{\pm}(s)\bigr)  \Bigr\rVert_{\infty}^{}  \mathrm{d}s\leq \mathfrak{g}\lVert g  \rVert_1^{} \int_{t_{\varepsilon}}^{\infty}  (4\pi s)^{-\frac{3}{2}} \mathrm{d}s\leq C\: t_{\varepsilon}^{-\frac{1}{2}}\;,
  \end{multline*}
  where in the last inequality we used \eqref{eq:13}.
\end{proof}
\begin{lemma}
  \label{lemma:3}
  There exists a constant $C_8>0$ such that,
  \begin{equation*}
    D_5(t_{\varepsilon})\leq C_8 \lVert K  \rVert_{}^{}\: t_{\varepsilon}^{-1}\;.
  \end{equation*}
\end{lemma}
\begin{proof} We first show that 
\begin{equation}
   \label{limpipmt}
\pi_\pm^t = \lvert \pm \rangle \langle \pm \rvert + \mathrm{O} (t^{-1}).
\end{equation}
Since $\pi_\pm(u)$ is Lipschitz in $\mathrm{Re}\langle u,g \rangle_2$, it suffices to show
that
$$
\phi(t) := \mathrm{Re}\langle u_\pm(t),g \rangle =  \mathrm{O} (t^{-1})\; .
$$
Recall that since $u_\pm \in L^2(\R_+,L^6)$, $\phi \in L^2(\R_+)$. Using the Duhamel formula, 
$$
\phi(t) = \mathrm{Re}\langle e^{i t \Delta} \underline{u}, g \rangle \mp \mathfrak g^2 \int_0^t \mathrm{Re} \langle  i e^{i(t-s)\Delta} g ,g \rangle  \frac{\phi(s)}{\sqrt{1+2 \mathfrak g^2 \phi(s)^2}} \di s,
$$
from which follows the estimate
$$
\abs{\phi(t)} \leq \abs{\langle e^{i t \Delta} \underline{u}, g \rangle} + \mathfrak g^2 \int_0^t \abs{\langle e^{i(t-s)\Delta} g ,g \rangle} \abs{\phi(s)} \di s.
$$
We estimate the free term by
\begin{equation}
    \abs{\langle e^{i t \Delta} \underline{u}, g \rangle} \leq \norm{e^{i t \Delta} \underline{u}}_{6} \norm{g}_{\frac{6}{5}} \leq (4 \pi \abs{t})^{-1} \norm{ \underline{u}}_{\frac{6}{5}}  \norm{g}_{\frac{6}{5}}.
\end{equation}
We then split the integral between $[0,t/2]$ and $[t/2,t]$. The first one can
be estimated using the dispersive estimates of the Schrödinger group and the
square integrability of $\phi$:
\begin{equation}
\begin{aligned}
     \int_{0}^{t/2} \abs{\langle e^{i(t-s)\Delta} g ,g \rangle} \abs{\phi(s)} \di s &\leq \int_{0}^{t/2} (4 \pi \abs{t-s})^{-\frac{3}{2}} \norm{g}_{L^1}^2 \abs{\phi(s)} \di s \\
     &\leq (2 \pi t)^{-3/2} \norm{g}_{L^1}^2 \int_0^{t/2} \abs{\phi(s)} \di s \\
     &\leq (2 \pi t)^{-3/2}\norm{g}_{L^1}^2   t^{1/2}   \norm{\phi}_{L^2(\R_+)}  =  \, \mathrm{O}(t^{-1}).
\end{aligned}
\end{equation}
where from the first to second line we used that $t-s \geq t/2$ under the integral. For the second integral we use that $\mathfrak g$ is small to take it to the left-hand side. Indeed, for $t/2 \leq s \leq t$, assuming that $t \geq 1$,
$$
\abs{\phi(s)} \leq \frac{ \norm{(1+\,\cdot \,) \phi(\,\cdot \,) }_{L^\infty([1,t])} }{1+s  } \leq \frac{\norm{(1+\,\cdot \,)\phi(\,\cdot \,) }_{L^\infty([1,t])}}{1+t/2}\;.
$$
Therefore,
\begin{multline}
    \int_{t/2}^t  \abs{\langle e^{i(t-s)\Delta} g ,g \rangle} \abs{\phi(s)} \di s \leq \frac{\norm{(1+\,\cdot \,) \phi(\,\cdot \,) }_{L^\infty([1,t])}}{1+t/2} \int_{\R_+} \abs{\langle e^{i s \Delta}g , g \rangle} \di s \\
\leq C  \frac{\norm{(1+\,\cdot \,) \phi(\,\cdot \,) }_{L^\infty([1,t])}}{1+t/2}, 
\end{multline}
since 
$$
 \int_{\R_+} \abs{\langle e^{i s \Delta}g , g \rangle} \di s \leq \norm{g}_{2}^2 + \int_1^\infty (4 \pi s)^{-3/2} \di s \norm{g}_{1}^2 < \infty.
$$
Putting everything together, we obtain
$$
(1+t) \abs{\phi(t)} \leq C (1+t) t^{-1} + \mathfrak g^2 C \norm{ (1+\,\cdot \,) \phi(\,\cdot \,) }_{L^\infty([1,t])},
$$
from which we can infer that
$$
\norm{ (1+\,\cdot \,) \phi(\,\cdot \,) }_{L^\infty([1,t])} \leq \frac{C}{1-C\mathfrak g^2},
$$
and therefore $ (1+\,\cdot \,) \phi(\,\cdot \,) $ is bounded in $[1,\infty)$, proving the expected decay for $\phi$ and thus for \eqref{limpipmt}.
Next, we need to examine $\mathcal R_\pm(t)$ with components denoted (we do not
stress the $\pm$ dependence)
$$
\mathcal R_j(t) = \begin{pmatrix}
    a_t & b_t \\ c_t & d_t
\end{pmatrix}.
$$
From the property of intertwining \eqref{intertwining}, as well as the estimate \eqref{limpipmt} above, we obtain
$$
[\mathcal R_\pm(t), \lvert \pm \rangle \langle \pm \rvert] = \mathrm{O}(t^{-1})\;.
$$
Computing the explicit commutator in $M_2(\C)$ it follows that the off diagonal terms must be small: $b_t,c_t = \mathcal O(t^{-1})$. Therefore, $\mathcal R_\pm(t)$ takes the form
\begin{equation}
    \label{asymptoticR_pm}
    \mathcal R_\pm(t) = a_t \lvert + \rangle \langle + \rvert + d_t \lvert - \rangle \langle - \rvert + \mathrm{O}(t^{-1})\;.
\end{equation}
By unitarity,
$$
I_2 = \mathcal R_\pm(t)^* \mathcal R_\pm(t) = \begin{pmatrix}
     \abs{a_t}^2 + \abs{c_t}^2 & \ast \\ \ast & \abs{b_t}^2 + \abs{d_t}^2 
\end{pmatrix}\;,
$$
thus $\abs{a_t}^2 = 1 + \mathrm{O}(t^{-1})$ and the same goes for
$d_t$. We finally come back to $D_5(t_\eps)$. Using \eqref{limpipmt} together
with \eqref{asymptoticR_pm} yields, \emph{e.g.}\ for the $+$ term,
\begin{equation*}
     \langle + \vert\mathcal{R}_+^{*}(t_{\varepsilon})\pi^{t_{\varepsilon}}_+ K\pi^{t_{\varepsilon}}_{+}\mathcal{R}_{+}(t_{\varepsilon})\rvert +  \rangle_{\mathbb{C}^2}  = \abs{\alpha_{t_\eps}}^2 \langle + \lvert K \rvert + \rangle  + \mathrm{O}(\norm{K} t_\eps^{-1}) = \langle + \lvert K \rvert + \rangle + \mathrm{O}(\norm{K} t_\eps^{-1})\;,
\end{equation*}
so that 
$$
D_5(t_\eps) = \mathrm{O}(\norm{K} t_\eps^{-1})\;.
$$
\end{proof}

\begin{proof}[Proof of Proposition~\ref{prop:4bis}]
  We can now combine the bounds above with the bound for $D_1$,
  the latter being given by a direct application of Proposition~\ref{prop:3}:
  \begin{equation*}
    D_1(\varepsilon,t_{\varepsilon})\leq 2\lVert K  \rVert_{}^{}C_1\,\sqrt{\varepsilon}\, e^{C_2\,t_{\varepsilon}}\;.
  \end{equation*}
  Combining all the bounds together one obtains that, indeed,
  \begin{multline*}
    \Bigl\lvert\, \mathrm{tr}_{\mathcal{S}}\Bigl(K\, \bigl(\hat{\boldsymbol{\varrho}}_{\varepsilon}(\xi)- \hat{\boldsymbol{\varrho}}_{\mathrm{ideal}}(\xi)\bigr)\Bigr)\,\Bigr\rvert\\\leq \lVert K  \rVert_{}^{}\Bigl(2C_1\,\sqrt{\varepsilon}\, e^{C_2\,t_{\varepsilon}}+ C_3\bigl(\lVert \xi  \rVert_{H^1}^{}+\lVert \xi  \rVert_{H^1}^2\bigr)\,\sqrt{\varepsilon}\,e^{C_4 t_{\varepsilon}}+C_5\,e^{-\frac{C_6}{\varepsilon}}+C_7\lVert \xi  \rVert_1^{}\,t_{\varepsilon}^{-\frac{1}{2}}+C_8\,t_{\varepsilon}^{-1}\Bigr)\;.
  \end{multline*}
  This bound yields both results that are stated in the theorem, choosing
  $t_{\varepsilon}= -\frac{1}{2C_2+2C_4}\ln \varepsilon$.
\end{proof}

  \appendix

  \section{Decoherence and mixed states} \label{appendixcommutatorestimates}

  In this appendix, we prove the results of Section~\ref{sec:introduction-1}.
  We start with Theorem~\ref{thm:5}, which states the measurement problem.

  \begin{proof}[Proof of Theorem~\ref{thm:5}]
  First, let $\Psi\in \mathscr{H}_{\mathcal{P}}$ be the initial state vector of the probe, and suppose that
  the initial state of $\mathcal{S}$ is one of the eigenvectors $\psi_j\in \mathscr{H}_{\mathcal{S}}$ of the
  measured observable $\mathcal{O}$. Then by collapse rule,
  \begin{equation*}
    \boldsymbol{U}(\psi_j\otimes \Psi)= \psi_j\otimes \Psi_j\;,
  \end{equation*}
  where $\boldsymbol{U}$ is the unitary measurement map, and $\Psi_j\in \mathscr{H}_{\mathcal{P}}$ is the probe's
  readout-state associated to the eigenvalue $\lambda_j$. For physical reasons, we
  might assume that the readout states $\{\Psi_j\}_j$ are all mutually
  orthogonal (even though this does not play a role in the proof). Now, any
  state vector $\psi\in \mathscr{H}_{\mathcal{S}}$ can be expanded in the orthonormal basis
  $\{\psi_j\}_{j\in \mathbb{N}}$ (we are supposing, for simplicity, that $\mathcal{O}$ is
  self-adjoint and with purely discrete spectrum), and therefore
  \begin{equation*}
    \psi=\sum_{j=1}^n \pi_j \psi_j\;,
  \end{equation*}
  where $\pi_j\in \mathbb{C}$ are so that $p_j= \lvert \pi_j \rvert_{}^2= \lvert \langle \psi_j , \psi \rangle_{\mathcal{S}}
  \rvert_{}^2$ gives, according to Born's rule, the probability to observe $\lambda_j$
  in the state $\psi$. So, on one hand, since $U$ is linear, it follows that
  \begin{equation*}
    \boldsymbol{U}(\psi\otimes \Psi) = \sum_j^{} \pi_j \psi_j\otimes \Psi_j\;;
  \end{equation*}
  and, therefore, its associated density matrix is
  \begin{align}
    \label{eq:14}
    \boldsymbol{U}(\lvert \psi\otimes \Psi\rangle\langle\psi\otimes \Psi\rvert)\boldsymbol{U}^{*}= \sum_j^{} p_j \lvert\psi_j\otimes \Psi_j\rangle\langle\psi_j\otimes \Psi_j&\rvert\\
    \label{eq:15}
    &+\sum_{j\neq k}^{}\bar{\pi}_k \pi_j \lvert\psi_j\otimes \Psi_j\rangle\langle\psi_k \otimes \Psi_k\rvert\;;
  \end{align}
  where $\lvert\psi_j\rangle\langle\psi_j\rvert$ and $\lvert\Psi_j\rangle\langle\Psi_j\rvert$ are orthogonal projectors and \eqref{eq:14} is
  called the diagonal part, while the operator \eqref{eq:15} is called the \emph{coherence} or
  off-diagonal part. If the off-diagonal part does not vanish, then both Born
  and collapse rules cannot hold; and clearly, \eqref{eq:15} cannot vanish unless we
  are in the trivial case that all but one of the $\pi_j$ are zero to start
  with. To prove that, observe that the operator $\boldsymbol{U}(\lvert \psi\otimes \Psi\rangle\langle\psi\otimes \Psi\rvert)\boldsymbol{U}^{*}$ has
  rank one, while \eqref{eq:14} has a rank of at least two (we are implicitly supposing
  that there are at least two distinct eigenvalues of $\mathcal{O}$, otherwise we would
  not be measuring a meaningful quantum observable). Hence the equality
  \begin{equation*}
    \boldsymbol{U}(\lvert \psi\otimes \Psi\rangle\langle\psi\otimes \Psi\rvert)\boldsymbol{U}^{*}= \sum_j^{} p_j \lvert\psi_j\otimes\Psi_j\rangle\langle\psi_j\otimes \Psi_j\rvert
  \end{equation*}
  can never be satisfied. Even if we suppose that the initial state of $\mathcal{P}$ is
  a density matrix $\varsigma$, the problem persists: a similar argument shows that
  there still can be no unitary map $\boldsymbol{U}$ such that $\boldsymbol{U}(\varrho\otimes \Gamma)\boldsymbol{U}^{*}$ satisfies
  the Born and collapse rules \citep[see, \emph{e.g.},][\textsection II for a detailed proof in
  this case]{weidlich1967zfp}.
\end{proof}

  Let us continue with the proof of Proposition~\ref{prop:superpsuppr} that we write concretely in its Fock space
  version. 
  %Lemma~\ref{lemmaprojectorweyl} can be understood as a
  %localisation result for a wave packet. The two other ones are related with
  %the commutation of $\mathrm{d}\Gamma_\eps(-\Delta)$ with operators defined as
  %functions of $\phi_\eps(g)$ by the functional calculus. Let us start with
  %Proposition~\ref{prop:superpsuppr}, that we prove in its Fock space form
  %for convenience.
  \begin{proposition*}
  Let $\mathscr{K}$ be a separable Hilbert space, and let $\psi_{\varepsilon},\varphi_{\varepsilon}\in \mathscr{F}(\mathscr{K})$ be
  normalized and such that there exist $C>0$ and $\delta>0$ so that
  \begin{equation*}
    \langle \psi_{\varepsilon}  , \bigl(\mathrm{d}\Gamma_{\varepsilon}(1)+1)^{\delta}\psi_{\varepsilon} \rangle_{\mathscr{F}(\mathscr{K})}+ \langle \varphi_{\varepsilon}  , \bigl(\mathrm{d}\Gamma_{\varepsilon}(1)+1)^{\delta}\varphi_{\varepsilon} \rangle_{\mathscr{F}(\mathscr{K})}\leq C\;.
  \end{equation*}
  Furthermore, suppose that $\mu,\nu\in \mathscr{P}(\mathscr{K})$ are such  that $\lvert\psi_{\varepsilon}\rangle\langle\psi_{\varepsilon}\rvert \xrightarrow[\varepsilon\to 0]{\mathcal{F}}\mu$ and $\lvert\varphi_{\varepsilon}\rangle\langle\varphi_{\varepsilon}\rvert \xrightarrow[\varepsilon\to 0]{\mathcal{F}}\nu$ with
%  \begin{gather*}
%\lvert\psi_{\varepsilon}\rangle\langle\psi_{\varepsilon}\rvert \xrightarrow[\varepsilon\to 0]{\mathcal{F}}\mu\;\;\mbox{and}\;\;
%\lvert\varphi_{\varepsilon}\rangle\langle\varphi_{\varepsilon}\rvert \xrightarrow[\varepsilon\to 0]{\mathcal{F}}\nu\;.
 % \end{gather*}
 %If in addition
 $\mu\perp\nu$. %
 Then, %, then:
  \begin{align*}
&\forall \alpha,\beta\in \mathbb{C}, \;\;
      \lvert \alpha \psi_{\varepsilon}+\beta\varphi_{\varepsilon}\rangle\langle\alpha\psi_{\varepsilon}+\beta\varphi_{\varepsilon}\rvert\xrightarrow[\varepsilon\to 0]{\mathcal{F}} \lvert \alpha  \rvert_{}^2\mu+\lvert \beta  \rvert_{}^2\nu,\\
&\forall \xi\in \mathscr{K},\;\;
      \lim_{\varepsilon\to 0}\langle \varphi_{\varepsilon}  ,e^{i\phi_{\varepsilon}(\xi)} \psi_{\varepsilon} \rangle_{\mathscr{F}(\mathscr{K})}=0\;.
  \end{align*}
\end{proposition*}
  \begin{proof}
    Firstly, by \citep[][Thm.\ 6.2]{ammari2008} observe that given
    $\psi_{\varepsilon},\varphi_{\varepsilon}$ satisfying the above bounds, there exist measures $\mu,\nu$
    such that the convergence in the sense of Fourier transforms holds, at
    least up to a subsequence extraction.
  Consider now the states
  \begin{equation*}
    \Gamma_{\varepsilon}(\alpha,\beta):= \lvert \alpha\psi_{\varepsilon}+\beta\varphi_{\varepsilon}\rangle\langle\alpha\psi_{\varepsilon}+\beta\varphi_{\varepsilon}\rvert\;,
  \end{equation*}
  as $\alpha,\beta\in \mathbb{C}$. Such states can be described by a projection-valued matrix
  \begin{equation*}
    M[\psi_{\varepsilon},\varphi_{\varepsilon}]=
    \begin{pmatrix}
      \lvert\psi_{\varepsilon}\rangle\langle\psi_{\varepsilon}\rvert & \lvert\psi_{\varepsilon}\rangle\langle\varphi_{\varepsilon}\rvert\\
      \lvert\varphi_{\varepsilon}\rangle\langle\psi_{\varepsilon}\rvert & \lvert\varphi_{\varepsilon}\rangle\langle\varphi_{\varepsilon}\rvert
    \end{pmatrix}\;,
  \end{equation*}
  since $\displaystyle{\Gamma_{\varepsilon}(\alpha,\beta)= \langle (\alpha,\beta)  , M[\psi_{\varepsilon},\varphi_{\varepsilon}](\alpha,\beta) \rangle_{\mathbb{C}^2}} $.
%  \begin{equation*}
%    \Gamma_{\varepsilon}(\alpha,\beta)= \langle (\alpha,\beta)  , M[\psi_{\varepsilon},\varphi_{\varepsilon}](\alpha,\beta) \rangle_{\mathbb{C}^2}\;.
 % \end{equation*}
  Furthermore, the state $\Gamma_{\varepsilon}$ and the ``complex states'' $\lvert\psi_{\varepsilon}\rangle\langle\varphi_{\varepsilon}\rvert$
  (non-positive trace class operators) satisfy:
  \begin{gather}
    \label{eq:11}
    \mathrm{tr}_{\mathscr{F}(\mathscr{K})}\bigl(\Gamma_{\varepsilon}(\mathrm{d}\Gamma_{\varepsilon}(1)+1)^{\delta}\bigr)\leq 2(\lvert \alpha  \rvert_{}^2+\lvert \beta  \rvert_{}^2)C\;,\\
    \label{eq:12}
    \mathrm{tr}_{\mathscr{F}(\mathscr{K})}\lvert(\mathrm{d}\Gamma_{\varepsilon}(1)+1)^{\delta/2}\lvert\psi_{\varepsilon}\rangle\langle\varphi_{\varepsilon}\rvert (\mathrm{d}\Gamma_{\varepsilon}(1)+1)^{\delta/2}\rvert\leq C\;.
  \end{gather}
  Let us start from the second inequality:
  \begin{align*}
   & \mathrm{tr}_{\mathscr{F}(\mathscr{K})}\lvert(\mathrm{d}\Gamma_{\varepsilon}(1)+1)^{\delta/2}\lvert\psi_{\varepsilon}\rangle\langle\varphi_{\varepsilon}\rvert (\mathrm{d}\Gamma_{\varepsilon}(1)+1)^{\delta/2}\rvert\\
    &\qquad \leq  \sqrt{\langle \psi_{\varepsilon}  , (\mathrm{d}\Gamma_{\varepsilon}(1)+1)^{\delta}\psi_{\varepsilon} \rangle_{\mathscr{F}(\mathscr{K})}} \mathrm{tr}_{\mathscr{F}(\mathscr{K})}\sqrt{(\mathrm{d}\Gamma_{\varepsilon}(1)+1)^{\delta/2}\lvert\varphi_{\varepsilon}\rangle\langle\varphi_{\varepsilon}\rvert (\mathrm{d}\Gamma_{\varepsilon}(1)+1)^{\delta/2}}\\
    &\qquad \leq \sqrt{\langle \psi_{\varepsilon}  , (\mathrm{d}\Gamma_{\varepsilon}(1)+1)^{\delta}\psi_{\varepsilon} \rangle_{\mathscr{F}(\mathscr{K})}\langle \varphi_{\varepsilon}  , (\mathrm{d}\Gamma_{\varepsilon}(1)+1)^{\delta}\varphi_{\varepsilon} \rangle_{\mathscr{F}(\mathscr{K})}}\leq C\;.
  \end{align*}
  Then, the first inequality is  a consequence of the latter and of the writing
  $\Gamma_{\varepsilon}(\alpha,\beta)= \langle (\alpha,\beta) , M[\psi_{\varepsilon},\varphi_{\varepsilon}](\alpha,\beta) \rangle_{\mathbb{C}^2}$.

  Thanks to \eqref{eq:12}, the complex state $\lvert\psi_{\varepsilon}\rangle\langle\varphi_{\varepsilon}\rvert$ satisfies the hypothesis of
  \citep[][Prop.\ 6.4]{ammari2008}, so there exists a sequence $\varepsilon_n\to 0$ and
  a complex measure $\eta_{\mathbb{C}}$ such that in the sense of Fourier transforms,
  \begin{equation*}
    M[\psi_{\varepsilon_n},\varphi_{\varepsilon_n}]\xrightarrow[\varepsilon_n\to 0]{\mathcal{F}}
    \begin{pmatrix}
      \mu&\eta_{\mathbb{C}}\\\bar{\eta}_{\mathbb{C}}&\nu
    \end{pmatrix}=: M[\mu,\nu,\eta_{\mathbb{C}}]\;;
  \end{equation*}
  that yields
  \begin{equation*}
    \Gamma_{\varepsilon_n}(\alpha,\beta)\xrightarrow[\varepsilon_n\to 0]{\mathcal{F}} \langle (\alpha,\beta)  , M[\mu,\nu,\eta_{\mathbb{C}}](\alpha,\beta) \rangle_{\mathbb{C}^2}\;.
  \end{equation*}
  However, in view of \eqref{eq:11}, all Wigner measures of $\Gamma_{\varepsilon}(\alpha,\beta)$ must
  be Borel Radon measures; hence for all Borel measurable $\mathfrak{B}\subseteq \mathscr{K}$,
  \begin{equation*}
   M[\mu,\nu,\eta_{\mathbb{C}}](\mathfrak{B}):=\begin{pmatrix}
      \mu(\mathfrak{B})&\eta_{\mathbb{C}}(\mathfrak{B})\\\bar{\eta}_{\mathbb{C}}(\mathfrak{B})&\nu(\mathfrak{B})
    \end{pmatrix}  
  \end{equation*}
  must be positive-definite, yielding
  \begin{equation*}
    \lvert \eta_{\mathbb{C}}(\mathfrak{B})  \rvert_{}^2\leq \mu(\mathfrak{B})\nu(\mathfrak{B})\;.
  \end{equation*}
  Hence, $\eta_{\mathbb{C}}\ll\mu$ and $\eta_{\mathbb{C}}\ll\nu$. Therefore, $\mu\perp\nu$ yields that $\eta_{\mathbb{C}}=0$,
  concluding the proof.% ;
  % this implies that:
  % \begin{itemize}
  % \item along all subsequences $\varepsilon_n\to 0$ of convergence in the sense of
  %   Fourier transforms,
  %   \begin{equation*}
  %     \Gamma_{\varepsilon_n}(\alpha,\beta)\xrightarrow[\varepsilon_n\to 0]{\mathcal{F}} \lvert \alpha  \rvert_{}^2\mu+\lvert \beta  \rvert_{}^2\nu\;,
  %   \end{equation*}
  %   and thus
  %   \begin{equation*}
  %     \Gamma_{\varepsilon}(\alpha,\beta)\xrightarrow[\varepsilon\to 0]{\mathcal{F}} \lvert \alpha  \rvert_{}^2\mu+\lvert \beta  \rvert_{}^2\nu\;;
  %   \end{equation*}
  % \item for all $\xi\in \mathscr{K}$,
  %   \begin{equation*}
  %     \lim_{\varepsilon\to 0} \langle \varphi_{\varepsilon}  , e^{i\phi_{\varepsilon}(\xi)}\psi_{\varepsilon} \rangle_{\mathscr{F}(\mathscr{K})}=0\;.
  %   \end{equation*}
  % \end{itemize}  
  \end{proof}

\bibliographystyle{abbrv} {\footnotesize\bibliography{biblio}}

\end{document}